\documentclass[10pt,twocolumn,letterpaper]{article}
\usepackage{usenix2019_v3}
\usepackage{url}
\usepackage{adjustbox}
\usepackage{amssymb}
\usepackage{amsmath}
\usepackage{amsthm}
\usepackage[notcomma,notperiod,notquote,notexcl,notcolon,notscolon]{hanging}
\usepackage{booktabs}
\usepackage[font={small}]{caption}
\usepackage{pifont}
\usepackage{color}
\usepackage{colortbl}
\usepackage{enumitem}
\usepackage{fontawesome5}
\usepackage{framed}
\PassOptionsToPackage{hyphens}{url}
\PassOptionsToPackage{sort}{cite}
\usepackage[pdfstartview=FitH,hidelinks,
    pdfpagelabels,bookmarksdepth=3,bookmarksopen=true,          bookmarksnumbered
]{hyperref}
\usepackage{nicefrac}
\usepackage{mathtools}
\usepackage{microtype}
\usepackage{multicol}
\usepackage{multirow}
\usepackage{scalefnt}
\usepackage{setspace}
\usepackage{thm-restate}
\usepackage{tikz}
\usetikzlibrary{decorations.pathreplacing}
\usetikzlibrary{calc}
\usetikzlibrary{hobby}
\usetikzlibrary{decorations.markings, decorations.pathmorphing}
\usetikzlibrary{shapes}
\usetikzlibrary{positioning,fit}
\definecolor{few-gray-bright}{HTML}{010202}
\definecolor{few-red-bright}{HTML}{EE2E2F}
\definecolor{few-green-bright}{HTML}{008C48}
\definecolor{few-blue-bright}{HTML}{185AA9}
\definecolor{few-orange-bright}{HTML}{F47D23}
\definecolor{few-purple-bright}{HTML}{662C91}
\definecolor{few-brown-bright}{HTML}{A21D21}
\definecolor{few-pink-bright}{HTML}{B43894}
\definecolor{few-gray}{HTML}{737373}
\definecolor{few-red}{HTML}{F15A60}
\definecolor{few-green}{HTML}{7AC36A}
\definecolor{few-blue}{HTML}{5A9BD4}
\definecolor{few-orange}{HTML}{FAA75B}
\definecolor{few-purple}{HTML}{9E67AB}
\definecolor{few-brown}{HTML}{CE7058}
\definecolor{few-pink}{HTML}{D77FB4}
\definecolor{few-gray-light}{HTML}{CCCCCC}
\definecolor{few-red-light}{HTML}{F2AFAD}
\definecolor{few-green-light}{HTML}{D9E4AA}
\definecolor{few-blue-light}{HTML}{B8D2EC}
\definecolor{few-orange-light}{HTML}{F3D1B0}
\definecolor{few-purple-light}{HTML}{D5B2D4}
\definecolor{few-brown-light}{HTML}{DDB9A9}
\definecolor{few-pink-light}{HTML}{EBC0DA}
 \usepackage{xcolor}
\usepackage{xspace}
\usepackage{titlesec}
\usepackage{wrapfig}
\usepackage{siunitx}
\usepackage[capitalise,noabbrev]{cleveref}
\usepackage{algpseudocode}
\usepackage[available,functional,reproduced]{usenixbadges}
\usepackage{nopageno}

\newcommand{\ffail}{f_\mathsf{live}}
\newcommand{\fevil}{f_\mathsf{secret}}

\newcommand{\CircA}{{\color{brown} \ding{202}}}
\newcommand{\CircB}{{\color{brown} \ding{203}}}
\newcommand{\CircC}{{\color{brown} \ding{204}}}
\newcommand{\CircD}{{\color{brown} \ding{205}}}
\newcommand{\CircE}{{\color{brown} \ding{206}}}
\newcommand{\CircF}{{\color{brown} \ding{207}}}
\newcommand{\CircG}{{\color{brown} \ding{208}}}

\newcommand{\Cover}{\mathsf{Cover}_{N,n,\Phi}}

\newcommand{\Ssalt}{S_\mathsf{salt}}

\newcommand{\data}{\mathsf{data}}
\newcommand{\PKEenc}{\mathsf{PKE.Encrypt}}
\newcommand{\PKEdec}{\mathsf{PKE.Decrypt}}
\newcommand{\PKEgen}{\mathsf{PKE.KeyGen}}
\newcommand{\block}{\mathsf{block}}
\newcommand{\addr}{\mathsf{addr}}

\newcommand{\Read}{\mathsf{Read}}
\newcommand{\Delete}{\mathsf{Delete}}

\newcommand{\LHEncAdv}{\mathsf{LHEncAdv}}
\newcommand{\CDHAdv}{\mathsf{CDHAdv}}

\newcommand{\CDH}{\mathsf{CDH}}

\newcommand{\Zgz}{\Z^{>0}}

\newcommand{\hash}{\ms{Hash}}

\newcommand{\Digest}{\ms{Digest}}
\newcommand{\val}{\ms{val}}

\newcommand{\ProveIncludes}{\ms{ProveIncludes}}
\newcommand{\VerifyIncludes}{\ms{DoesInclude}}
\newcommand{\ProveExtends}{\ms{ProveExtends}}
\newcommand{\VerifyExtends}{\ms{DoesExtend}}
\newcommand{\pfext}{\pi_\ms{Ext}}
\newcommand{\pfinc}{\pi_\ms{Inc}}

\newcommand{\msg}{\ms{msg}}

\newcommand{\Sshare}{\ms{Shamir.Share}}
\newcommand{\reconst}{\ms{Reconstruct}}
\newcommand{\Sreconst}{\ms{Shamir.Reconst}}

\newcommand{\id}{\mathsf{id}}
\newcommand{\F}{\mathbb{F}}

\newcommand{\N}{\Z^{>0}}
\newcommand{\Z}{\mathbb{Z}}

\newcommand{\G}{\mathbb{G}}

\newcommand{\Hash}{\mathsf{Hash}}

\newcommand{\user}{\mathsf{user}}
\newcommand{\username}{\user}

\newcommand{\pin}{\mathsf{pin}}

\newcommand{\setup}{\mathsf{Setup}}
\newcommand{\backup}{\mathsf{Backup}}

\newcommand{\puncture}{\mathsf{Puncture}}

\newcommand{\enc}{\mathsf{Encrypt}}
\newcommand{\dec}{\mathsf{Decrypt}}

\newcommand{\keygen}{\mathsf{KeyGen}}

\newcommand{\select}{\mathsf{Select}}

\newcommand{\recover}{\mathsf{Recover}}

\newcommand{\salt}{\mathsf{salt}}

\newcommand{\AEenc}{\mathsf{AE.Encrypt}}
\newcommand{\AEdec}{\mathsf{AE.Decrypt}}
\newcommand{\AEAdv}{\mathsf{AEAdv}}
\newcommand{\AEscheme}{\mathsf{AE}}

\newcommand{\calA}{\ensuremath{\mathcal{A}}}
\newcommand{\calB}{\ensuremath{\mathcal{B}}}

\newcommand{\calE}{\ensuremath{\mathcal{E}}}

\newcommand{\calH}{\ensuremath{\mathcal{H}}}

\newcommand{\calK}{\ensuremath{\mathcal{K}}}

\newcommand{\calM}{\ensuremath{\mathcal{M}}}

\newcommand{\calP}{\ensuremath{\mathcal{P}}}

\newcommand{\calS}{\ensuremath{\mathcal{S}}}

\newcommand{\deq}{\mathrel{\mathop:}=}
\newcommand{\zo}{\ensuremath{\{0,1\}}}

\newtheorem{challenge}{Challenge}

\theoremstyle{plain}
\newtheorem{theorem}{Theorem}

\newtheorem{lemma}[theorem]{Lemma}

\theoremstyle{definition}

\newtheorem{defn}[theorem]{Definition}
\newtheorem{experiment}[theorem]{Experiment}

\theoremstyle{remark}
\newtheorem{remark}[theorem]{Remark}

\newtheoremstyle{goal}
  {\topsep}
  {\topsep}
  {\normalfont}
  {0pt}
  {\bfseries}
  {: } 
  { }
  {\thmname{#1}\thmnumber{ #2}\thmnote{ (#3)}}
\theoremstyle{goal}

\newcommand{\esm}[1]{\ensuremath{#1}}
\newcommand{\ms}[1]{\esm{\mathsf{#1}}}

\newcommand{\poly}{\operatorname{poly}}

\newcommand{\negl}{\operatorname{negl}}

\newcommand{\getsr}{\rgets}
\newcommand{\rgets}{\mathrel{\mathpalette\rgetscmd\relax}}
\newcommand{\rgetscmd}{\ooalign{$\leftarrow$\cr
    \hidewidth\raisebox{1.2\height}{\scalebox{0.5}{\ \rm R}}\hidewidth\cr}}

\newcommand{\abs}[1]{\left| #1 \right|}

\newcommand{\ct}{\ms{ct}}

\newcommand{\sk}{\ms{sk}}
\newcommand{\pk}{\ms{pk}}

\newcommand{\mpk}{\ms{mpk}}

\newcounter{ExperimentCount}

\newcounter{PropertyCount}
\newcommand{\Property}[1]{\refstepcounter{PropertyCount}\paragraph{Property~\arabic{PropertyCount}: #1.}}

\newcommand{\name}{True2F\xspace}

\newif\iffull
\fulltrue

\makeatletter
\g@addto@macro{\UrlBreaks}{\UrlOrds}
\makeatother

\newlength{\defhangindent}
\usepackage[T1]{fontenc}
\usepackage{newtxmath}
\usepackage{newtxtext}

\usepackage[scaled=.8]{beramono}

\DeclareMathAlphabet\mathcal{OMS}{cmsy}{m}{n}
\SetMathAlphabet\mathcal{bold}{OMS}{cmsy}{b}{n}

\SetMathAlphabet{\mathsf}{normal}{OT1}{cmss}{m}{n}
\SetMathAlphabet{\mathsf}{bold}{OT1}{cmss}{bx}{n}

\DeclareSymbolFont{AMSb}{U}{msb}{m}{n}
\DeclareSymbolFontAlphabet{\mathbb}{AMSb}

\DeclareSymbolFont{numbers}{T1}{ptm}{m}{n}
\SetSymbolFont{numbers}{bold}{T1}{ptm}{bx}{n}
\DeclareMathSymbol{0}\mathalpha{numbers}{"30}
\DeclareMathSymbol{1}\mathalpha{numbers}{"31}
\DeclareMathSymbol{2}\mathalpha{numbers}{"32}
\DeclareMathSymbol{3}\mathalpha{numbers}{"33}
\DeclareMathSymbol{4}\mathalpha{numbers}{"34}
\DeclareMathSymbol{5}\mathalpha{numbers}{"35}
\DeclareMathSymbol{6}\mathalpha{numbers}{"36}
\DeclareMathSymbol{7}\mathalpha{numbers}{"37}
\DeclareMathSymbol{8}\mathalpha{numbers}{"38}
\DeclareMathSymbol{9}\mathalpha{numbers}{"39}

\SetSymbolFont{numbers}{bold}{T1}{ptm}{b}{n}
\makeatletter
\renewcommand{\operator@font}{\mathgroup\symnumbers}
\makeatother

\makeatletter

\definecolor{LightCyan}{rgb}{0.88,1,1}

\long\def\com#1{}

\long\def\xxx#1{}

\newcommand{\itpara}[1]{\smallskip\noindent\textit{#1}}
\newcommand{\sitpara}[1]{\smallskip\noindent\textit{#1}}
\renewcommand{\paragraph}[1]{\smallskip\noindent\textbf{#1}}

\makeatletter

\let\c@table\c@figure
\makeatother 

\setlist[description]{leftmargin=\parindent,topsep=0ex,itemsep=0ex,partopsep=1ex,parsep=1ex}

\LetLtxMacro{\oldtextsc}{\textsc}
\renewcommand{\textsc}[1]{\oldtextsc{\scalefont{1.2}#1}}

\setitemize{noitemsep,topsep=2pt,parsep=2pt,partopsep=2pt,leftmargin=4ex}
\setenumerate{noitemsep,topsep=2pt,parsep=2pt,partopsep=2pt,leftmargin=4ex}

\newcolumntype{R}[2]{
  >{\adjustbox{angle=#1,lap=\width-(#2)}\bgroup}
    c
    <{\egroup}
}

\newcommand{\rott}[1]{{\adjustbox{angle=90,lap=\width-1em}{#1}}}

\renewcommand{\name}{SafetyPin\xspace}
\newcommand{\Name}{SafetyPin\xspace}

\begin{document}

\title{\Name: Encrypted Backups with Human-Memorable Secrets}

\author{Emma Dauterman\\
\emph{UC Berkeley}
\and
Henry Corrigan-Gibbs\\
\emph{EPFL and MIT CSAIL}
\and
David Mazi\`eres\\
\emph{Stanford}}
\date{}
\maketitle

\hyphenation{Safety-Pin}

\frenchspacing
\paragraph{Abstract.}
We present the design and implementation of \name,
a system for encrypted mobile-device backups.
Like existing cloud-based mobile-backup systems, including
those of Apple and Google, \name
requires users to remember only a short PIN and
defends against brute-force PIN-guessing attacks using
hardware security protections.
Unlike today's systems, \name splits trust over a cluster of
hardware security modules (HSMs) in order to provide security
guarantees that scale with the number of HSMs.
In this way, \name protects backed-up user data
even against an attacker that can adaptively
compromise many of the system's constituent HSMs. 
SafetyPin provides this protection without sacrificing
scalability or fault tolerance.
Decentralizing trust while respecting the resource
limits of today's HSMs requires a synthesis of
systems-design principles and cryptographic tools.
We evaluate \name on a cluster of 100 low-cost HSMs 
and show that a \name-protected recovery takes 1.01 seconds.
To process 1B recoveries a year, we estimate that
a \name deployment would need 3,100 low-cost HSMs.
 \section{Introduction}
\label{sec:intro}

Modern mobile phones and tablets back up sensitive data to the cloud.
To protect users' privacy, this data must be encrypted under keys that
are not available to the cloud provider.  Unfortunately, with
3.8~billion smartphone users, it is impractical to expect them all to
store, say, a 128-bit AES backup key.  Not everyone has a computer, or
trustworthy friends who can keep shares of a backup key, or even a
safe place to store a backup key on paper.  As a result, mobile OSes
have fallen back to protecting backups with the least common
denominator: device screen-lock PINs.
Using PINs is good for security
because a user's screen-lock PIN never leaves her device 
(so the cloud provider never learns it). Using PINs is good for
usability because users generally remember them.

Unfortunately, PINs have such low entropy (e.g,
six decimal digits) that no feasible amount of key
stretching can protect against brute-force PIN-guessing attacks.
Instead, modern backup systems---such as those from
Apple~\cite{apple-key-vault}, Google~\cite{google-key-vault}, and
Signal~\cite{signal-recovery}---rely on hardware-security modules
(HSMs) in their data centers to thwart brute-force attacks.
Specifically, devices encrypt their backup keys under the public keys
of HSMs, but each device includes a hash of its screen-lock PIN as part 
of the plaintext.  
HSMs return decrypted plaintext only to clients that can supply this PIN hash.  
Furthermore, HSMs limit the number of decryption attempts for
any given user account.  For fault tolerance, a device typically
encrypts its backup key to the public keys of five HSMs, 
allowing any one of the five to recover the backup key.

This status quo still falls short of acceptable privacy for two
reasons.  First, HSMs are not perfect, yet each HSM in these systems
is a single point of security failure for millions of users' backup keys.
Second, these systems make it difficult for clients to detect security breaches.
For instance, if a malicious insider working in a
data center physically steals an HSM, then to anyone outside the
company it looks like an unremarkable single hardware failure.
Alternatively, if an insider successfully guesses someone's PIN, the
victim may have no idea her backup was ever compromised.

This paper presents \name, a PIN-based encrypted-backup system with
stronger security properties.
The key idea behind \name is that recovering any user's backed-up 
data either requires (a) guessing the user's PIN or (b)
compromising a very large number of HSMs---e.g.,
6\% of all HSMs operated by a provider.
(The 6\% figure here is a tunable system parameter.)
Such large-scale attacks
would typically need to span multiple data centers, be harder for
insiders to pull off undetected against physical devices, cost more,
and also likely cause service disruptions visible to end users.

\begin{figure}
\centering
\includegraphics[width=\columnwidth]{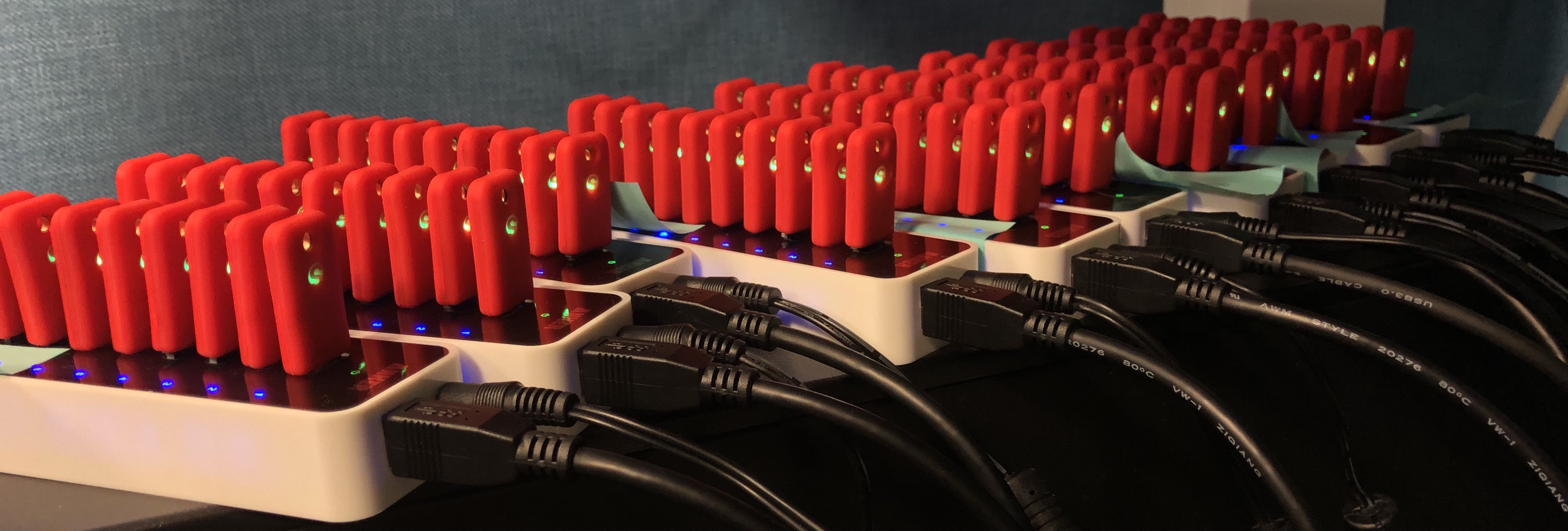}
\caption{Our cluster of 100 low-cost hardware security modules
        (SoloKeys~\cite{solokeys}) on which we evaluate \name.} 
\label{fig:datacenter}
\end{figure}

One way to achieve \name's security goal would be to threshold-encrypt
the client's hashed PIN and backup key in such a way that decrypting
the client's backup key would require the participation of 6\% of all HSMs
in the system.
Unfortunately, this approach lacks scalability.
If each client recovering a backup must interact with 
6\% of the system's HSMs, adding more HSMs improves
security without improving throughput.  
As the number of HSMs in the system increases, 
we would like the system's overall throughput to increase in tandem
with its security (i.e., the attacker's cost).

To achieve scalability, \name takes a different
approach: devices threshold-encrypt their backup
keys to a small cluster of $n$ HSMs such that
decryption requires the participation of most HSMs in the cluster.
The cluster size $n$ is independent of the total number of HSMs in the system,
and depends on both the fraction of compromised HSMs the system can tolerate
and the fraction of HSMs that can fail-stop.
(For example, to tolerate the compromise of 6\% of
HSMs where half of a cluster is allowed to fail-stop,
we can set the cluster size $n=40$.)
This design achieves our scalability goal, since each device need
only communicate with a small fixed number of HSMs during recovery.
This design also achieves our security goal because the cluster of $n$ HSMs that can
decrypt a client's backup depends on the client's secret PIN,
via a primitive we introduce called \emph{location-hiding encryption}.  
Hence, even if an attacker compromises 6\% of the HSMs in the system as a whole,
the chances that the attacker compromises a ``useful'' set of HSMs---i.e., 
at least half of the HSMs in the device's chosen cluster---is very small.
More precisely, we show that if the total number
of HSMs in the system is large enough (a few
hundred or more), the probability that an attacker
can decrypt a backup via HSM compromise is not
much higher than the probability of simply
guessing the client's PIN.

In modern backup systems, each HSM only needs to monitor the number of
PIN attempts for a small subset of users, but because of our
location-hiding encryption primitive, every HSM needs to be able to
verify the number of PIN attempts for every user.
To maintain this information scalably, the HSMs use a new type
of distributed log.
Third parties can monitor this log to alert users whenever a backup-recovery
attempt is underway.
Since a compromised service provider may see which
HSMs a mobile device interacts with during recovery (and could compromise
those HSMs to recover the users' backed-up data), HSMs revoke their
ability to decrypt backups after completing the recovery process.
Implementing this revocation requires adapting ``puncturable
encryption''~\cite{GM15} to storage-limited HSMs.  While our prototype
is focused on PIN-protected backups, these primitives have potentially
broader applicability to problems such as private storage in
peer-to-peer systems and cryptocurrency ``brain wallets.''

We implemented \name on low-cost SoloKey HSMs~\cite{solokeys}.  We
evaluate the system using a cluster of 100 SoloKeys
(\cref{fig:datacenter}) and an Android phone (representing the client
device).  Generating a recovery ciphertext on the client, excluding
the time to encrypt the disk image, takes 0.37 seconds.
To process 1B recoveries a year, or 123K recoveries per hour, we estimate that
we would need 3,100 SoloKeys.
In a \name deployment of 3,100 HSMs, tolerating the compromise of 6\% of
the HSMs (i.e., 194 HSMs),
the client must interact with a cluster of 40 HSMs during recovery.
Running our backup-recovery protocol across a cluster of this size
takes 1.01 seconds.

\paragraph{Limitations.}
A limitation of \name is that the set of HSMs a device uses for
recovery can leak information about the user's PIN.
In particular, an attacker who controls the data center can learn a salted
hash of the user's PIN during recovery.
This is unfortunate in the common case that
people re-use the same PIN after
recovery~\cite{DBC+14,HWS13,GF06,SKK+10}.  
We discuss one mitigation in \cref{par:leakage}.
Also, while it is possible to
detect when PINs can safely be re-used, we have not yet implemented
this functionality.

In addition, \name is more expensive than today's PIN-based backup systems.
\name requires the data center operator to operate a much larger
fleet of HSMs (roughly $50-100\times$ larger) than the standard HSM-based 
backup systems require.  
\Name clients must also download roughly 2MB of keying material per day
in a \name deployment supporting one billion recoveries per year,
due to the periodic rotation of large HSM keys.
Even so, we expect that the cost of storing and transferring disk images
(GBs/user) will dwarf these costs.
 \section{The setting}\label{sec:setting}
\paragraph{Entities.}
Our encrypted-backup system involves three entities, whose
roles we describe here.

\sitpara{Client.}
Initially, the client holds 
(1) a username with the service provider,
(2) a human-memorable passphrase or PIN, 
(3) a disk image to be backed up, and
(4) the public keys of the service provider's HSMs.
Later on, the client should be able to recover
her backed-up data using only her
username, her PIN, and access to the other components of 
the backup system.

In \name, as in today's PIN-based backup systems,
security depends on the client having access to
the HSMs' true public keys: If a malicious service
provider can swap out the HSMs' true
public keys for its own
public keys without detection, the service
provider can immediately break security.
Using a distributed log (\cref{sec:log}) can ensure that
all clients see a common set of HSM public keys, to 
prevent targeted attacks.
Hardware-attestation techniques, as used in the FIDO~\cite{fido-att}
and SGX~\cite{sgx-att} specs, can provide another defense.

We also assume the provider has traditional account
authentication (e.g., Gmail passwords) to prevent random third parties
from consuming PIN guesses, but we omit this from the discussion for
simplicity.

\sitpara{Service provider.}
The service provider offers the encrypted-backup
service to a pool of clients and it maintains
the data centers in which the backup system runs.
For example, the service provider could be
a mobile-phone vendor, such as Apple or Google.
The service provider's data centers contain the network
infrastructure that connects the HSMs.  They also
contain large amounts of (potentially
untrustworthy) storage and computing resources.
Our security properties will hold against a
service provider that becomes compromised at any
point after the system is set up.

\sitpara{Hardware security modules (HSMs).}
The service provider's data centers contain 
thousands of hardware security modules. 
An HSM is a tamper-resistant computing device
meant for storing cryptographic secrets.
HSMs have fully programmable processors but 
are typically resource-poor (see \cref{tab:hsm}).
It is possible to lock an HSM's firmware before
deployment, which makes remote compromise and
key-extraction attacks more difficult. 
Each HSM has a public key and stores the
corresponding secret key in its secure memory.

\begin{table}
\centering
{\small
\begin{tabular}{rrrrc}
  \toprule
  Device & Price & $g^x$/sec. & Storage & FIPS\\ \midrule
  SoloKey~\cite{solokeys} & \$20 & 8 & 256 KB$^\star$\hspace{-0.57em} & \\
  YubiHSM 2~\cite{yubico-hsm} & \$650 & 14 & 126 KB &\\
  SafeNet A700~\cite{safenet} & \$18,468 & 2,000 & 2,048 KB& \checkmark \\\midrule
  Intel i7-8569U (CPU) & \$431 & 22,338 & n/a \\
  \bottomrule
\end{tabular}
}
  \caption[HSM]{Hardware security modules offer physical security protections but
  are computationally weak compared~to~a~standard~CPU\@.\\[-6pt]

  {\footnotesize {\setstretch{0.8} The $g^x$/sec is NIST P256 elliptic-curve 
  point-multiplications~per second. ``FIPS?'' refers to whether the device
  meets the FIPS 140-2 standard for HSMs.
  ($^\star$\,The 256 KB storage on the SoloKey is shared between code and data.)\par}}}
  \label{tab:hsm}
\end{table}

\paragraph{The attack scenario.}
The service provider (Apple, Google, etc.) 
spends vast amounts of money acquiring
a large user base for products that store user data in the cloud.
The provider risks reputational damage and 
journalistic scrutiny if it cannot ensure the durability and
confidentiality of user data.

A service provider can deploy \name as a way to build 
trust among its user base and to protect its own 
infrastructure against future compromise.
By enlisting third-party organizations to monitor the \name
deployment's public distributed log, the provider can build
further public trust in the system.

At some point after the provider deploys \name, 
a powerful attacker wishes to
steal user data.  The attacker may have malicious insiders working for
the provider.  It may physically compromise data centers to steal
HSMs.  It may intercept shipments to tamper with some of the HSMs on
their way to the data center.  The attacker could also be a state
actor employing legal pressure to gain access to data centers.
Nonetheless, the attacker is sensitive to both the cost of attacks and
the risk of public exposure.

Both the attack cost and risk of exposure increase with the number of HSMs
the attacker must compromise.  For instance, while a malicious insider
working at a data center may be able to abscond with a single
HSM---passing the missing device off as a hardware failure---removing
100 HSMs is a much riskier proposition.  A state actor who can order
the provider to hand over HSMs may be dissuaded if doing so will
attract press coverage either by making non-targeted clients' data
unrecoverable or creating a damning public audit trail.

The attacker may compromise clients as well as the provider.  For
instance, the attacker may have a good guess at a target user's PIN,
perhaps because of CCTV footage showing the user unlocking a mobile
device.  While \name cannot prevent the attacker from gaining access
to the data with the correct PIN, the risk will be higher to the
attacker if stolen PINs cannot be used without exposing the attack in
\name's public distributed log.

\itpara{Notation.}
The set $\Zgz$ refers to the set of natural numbers $\{1, 2, 3, \dots\}$.
For a positive integer $n$, we let $[n] = \{1, \dots, n\}$
and we use $\bot$ to denote a failure symbol.
For strings $a$ and $b$, we write their concatenation
as $a\|b$.
Throughout, we use $\lambda$ to denote the security parameter,
and we typically take $\lambda = 128$ (i.e., for 128-bit security).
 \section{System goals}
\label{sec:design}

\Name implements an encrypted-backup functionality, 
which consists of two routines:
\begin{itemize}
  \item the \textit{backup} algorithm, which the client uses to
    produce its encrypted backup, and

  \item the \textit{recovery} protocol, in which the
        client uses HSMs to recover the backup plaintext from
        ciphertext.
\end{itemize}

We define these protocols with respect to
a number of HSMs $N \in \Zgz$ and a finite PIN space 
$\calP \subseteq \zo^*$.
For convenience, we define the \textit{master public key}
$\mpk$ for a data center to be all $N$ HSMs' 
public keys: $\mpk = (\pk_1, \dots, \pk_N)$.
The syntax of an encrypted-backup system 
is then as follows:

\smallskip

\begin{hangparas}{\defhangindent}{1}
    $\backup(\mpk, \username, \pin, \msg) \to \ct$.
        Given the master public key $\mpk$, 
        a client username $\user$,
        the client's PIN $\pin \in \calP$, and a message $\msg \in \zo^*$ 
        to be backed up, output a recovery ciphertext $\ct$.
        This routine runs on the client and requires no
        interaction with HSMs.
        The client uploads the resulting ciphertext $\ct$
        to the service provider.

    $\recover^{\calS, \calH_1, \dots, \calH_N}(\mpk, \user, \pin, \ct) \to \msg$ or $\bot$.
        The client initiates the recovery routine, which 
        takes as input the master public key $\mpk$,
        a client username $\user$, 
        a PIN $\pin \in \calP$, and a recovery ciphertext $\ct$.
\end{hangparas}
\begin{hangparas}{\defhangindent}{0}
        During the execution of $\recover$, the client interacts with
        the service provider $\calS$ and a subset of the HSMs $\calH_1, \dots, \calH_N$.
        Each HSM $\calH_i$ holds the master public key $\mpk$, 
        and its secret decryption key~$\sk_i$.
        During recovery, the data center provides
        the client's username~$\user$ to each HSM.
\end{hangparas}
\begin{hangparas}{\defhangindent}{0}
        The recovery routine outputs a backed-up message $\msg \in \zo^*$
        or a failure symbol $\bot$.
\end{hangparas}

\medskip
We now describe the security properties that such a
system should satisfy.
We work in an asynchronous network model; we use
standard cryptographic primitives to set up 
authenticated and encrypted channels between 
the client, service provider, and HSMs.

\Property{Security} \label{prop:sec}
If the client obtains the HSMs' true public keys, 
then even an attacker that: 
\begin{itemize}
    \item controls the service provider (in particular, is an active network attacker inside
        the data centers and has control of the service provider's servers and storage),
      \item compromises an $\fevil$ (e.g., $\fevil=\tfrac{1}{16}$) fraction of HSMs in the data center 
        \emph{before} the client begins the recovery process, and
      \item compromises all of the HSMs in the data center 
        \emph{after} the recovery protocol completes,
\end{itemize}
still should learn nothing 
about any honest client's encrypted message 
(in a semantic-security sense~\cite{GM84}) beyond what it
can learn by guessing that client's PIN\@.

\sitpara{Discussion:}  The adversary can inspect all clients' recovery
ciphertexts and then choose to compromise a large set of HSMs that
depends on these ciphertexts.  Such attacks are relevant when, for
example, a state actor with the power to compromise many HSMs 
targets the backed-up data of a specific set of users.

Two important caveats are: (1) \name does not
protect against an attacker compromising HSMs
while recovery is in progress (see
\cref{fig:timeline}) and
(2) as implemented, \name does not protect the PIN: an
adversary that observes which HSMs the client
contacts during recovery may learn a salted hash
of the PIN after recovery completes.
Section~\ref{sec:log:audit} discusses how to detect and
mitigate this leakage by protecting the salt.

\Property{Scalability} 
The recovery protocol should require the client to
interact with a constant number of HSMs, independent of the
number of HSMs in the data center.
(This constant may depend on the security parameter and 
on the fraction of HSMs whose compromise the system can tolerate.)
Hence, providers can deploy additional HSMs to scale capacity.
Concretely, when we configure the system to
tolerate the
compromise of $\fevil = \tfrac{1}{16}$ of the data center's HSMs,
our protocol requires the client to 
communicate with $40$ HSMs during recovery.

\Property{Fault tolerance} 
Every client should be able to recover her encrypted
message even if a constant fraction $\ffail$ (e.g. $\ffail = \tfrac{1}{64}$)
of the HSMs in the data center fail-stop.

\paragraph{Setting parameters.}
For the remainder of this paper, we set the fraction of compromised HSMs that the 
system can tolerate to $\fevil = \tfrac{1}{16}$ and the fraction
of HSMs that can fail while still allowing the client to recover her
backup to $\ffail = \tfrac{1}{64}$.
This choice is reasonable because large companies have more than 16
data centers, while smaller companies can collaborate on a shared
deployment with 16 physical security perimeters.
By adjusting the other parameters, it is possible to achieve
any $0 < \fevil < 1$ or $0 < \ffail < 1$.
(In \cref{sec:eval:e2e-costs}, we discuss how the choice
of these values affects other system parameters.)
 \begin{figure*}
  \centering
  \resizebox{0.22\textwidth}{!}{\tikzset{
  status/.style={fill=white,inner sep=0.5pt,midway,circle},
  clean/.style={fill=white,inner sep=-2pt,circle}
}
\begin{tikzpicture}[use Hobby shortcut]
\useasboundingbox (-2.6,-0.3) rectangle (2,3.1);
\definecolor{few-gray-bright}{HTML}{010202}
\definecolor{few-red-bright}{HTML}{EE2E2F}
\definecolor{few-green-bright}{HTML}{008C48}
\definecolor{few-blue-bright}{HTML}{185AA9}
\definecolor{few-orange-bright}{HTML}{F47D23}
\definecolor{few-purple-bright}{HTML}{662C91}
\definecolor{few-brown-bright}{HTML}{A21D21}
\definecolor{few-pink-bright}{HTML}{B43894}
\definecolor{few-gray}{HTML}{737373}
\definecolor{few-red}{HTML}{F15A60}
\definecolor{few-green}{HTML}{7AC36A}
\definecolor{few-blue}{HTML}{5A9BD4}
\definecolor{few-orange}{HTML}{FAA75B}
\definecolor{few-purple}{HTML}{9E67AB}
\definecolor{few-brown}{HTML}{CE7058}
\definecolor{few-pink}{HTML}{D77FB4}
\definecolor{few-gray-light}{HTML}{CCCCCC}
\definecolor{few-red-light}{HTML}{F2AFAD}
\definecolor{few-green-light}{HTML}{D9E4AA}
\definecolor{few-blue-light}{HTML}{B8D2EC}
\definecolor{few-orange-light}{HTML}{F3D1B0}
\definecolor{few-purple-light}{HTML}{D5B2D4}
\definecolor{few-brown-light}{HTML}{DDB9A9}
\definecolor{few-pink-light}{HTML}{EBC0DA}
 \fill [color=black!10,rounded corners,thick] (-1.5,0.25) rectangle (2,2.75);

\node[fill=pink,rounded corners, inner sep=1.5pt] at (-2.2,0.3) {\footnotesize $\msg$};
\node at (-2.25,0.5) {\fontsize{40pt}{2pt} \selectfont \faMobile*};
\node at (-0.75,2) {\color{green!20!black} \Huge \faDatabase};
\node at (-2.2,0.9) {\footnotesize $\mpk$};
\node at (-2.2,0.62) {\footnotesize $\pin$};

\newcounter{hsmcount}\setcounter{hsmcount}{1}
\foreach \b in {3,...,1} {
  \foreach \d in {1,...,3} {
    \node at (0.6*\d-0.2,0.75*\b-0.1) {\rotatebox{90}{\Large \color{white!70!black} \faMicrochip}};
    \node at (0.6*\d-0.2,0.75*\b-0.1) {\color{white!90!black} {\footnotesize$\sk_\thehsmcount$}};
    \stepcounter{hsmcount}
  }
}

\node[anchor=south west] at (-1.5,2.75) {\textbf{Service provider}};
\node[anchor=south] at (1.4,2.375) {\color{white!70!black} \footnotesize \textbf{HSMs}};
  \node[anchor=south] at (-0.75,0.9) {\color{green!20!black} \footnotesize \parbox{1in}{\setstretch{0.7} \centering \textbf{External\\storage}}};

\draw[thick,<-] (-1.25,2.2) to[bend right] node[status] {\footnotesize $\ct$} (-2.2,1.2);
\node[clean] at (-2.15,1.4) {\CircA};

\end{tikzpicture}
 }~~~~\vline{}\vline{}\vline{}\vline{}~~~~
  \resizebox{0.22\textwidth}{!}{\tikzset{
  status/.style={fill=white,inner sep=0.5pt,midway,circle},
  textbox/.style={fill=white,inner sep=0.5pt,midway,rectangle},
  clean/.style={fill=white,inner sep=-2pt,circle}
}
\begin{tikzpicture}[use Hobby shortcut]
\useasboundingbox (-2.6,-0.3) rectangle (2,3.1);
\definecolor{few-gray-bright}{HTML}{010202}
\definecolor{few-red-bright}{HTML}{EE2E2F}
\definecolor{few-green-bright}{HTML}{008C48}
\definecolor{few-blue-bright}{HTML}{185AA9}
\definecolor{few-orange-bright}{HTML}{F47D23}
\definecolor{few-purple-bright}{HTML}{662C91}
\definecolor{few-brown-bright}{HTML}{A21D21}
\definecolor{few-pink-bright}{HTML}{B43894}
\definecolor{few-gray}{HTML}{737373}
\definecolor{few-red}{HTML}{F15A60}
\definecolor{few-green}{HTML}{7AC36A}
\definecolor{few-blue}{HTML}{5A9BD4}
\definecolor{few-orange}{HTML}{FAA75B}
\definecolor{few-purple}{HTML}{9E67AB}
\definecolor{few-brown}{HTML}{CE7058}
\definecolor{few-pink}{HTML}{D77FB4}
\definecolor{few-gray-light}{HTML}{CCCCCC}
\definecolor{few-red-light}{HTML}{F2AFAD}
\definecolor{few-green-light}{HTML}{D9E4AA}
\definecolor{few-blue-light}{HTML}{B8D2EC}
\definecolor{few-orange-light}{HTML}{F3D1B0}
\definecolor{few-purple-light}{HTML}{D5B2D4}
\definecolor{few-brown-light}{HTML}{DDB9A9}
\definecolor{few-pink-light}{HTML}{EBC0DA}
 \fill [color=black!10,rounded corners,thick] (-1.5,0.25) rectangle (2,2.75);

\node at (-2.25,0.5) {\fontsize{40pt}{2pt} \selectfont \faMobile*};
\node at (-0.75,2) {\color{green!20!black} \Huge \faDatabase};
\node at (-2.2,0.9) {\footnotesize $\mpk$};
\node at (-2.2,0.62) {\footnotesize $\pin$};

\setcounter{hsmcount}{1}
\foreach \b in {3,...,1} {
  \foreach \d in {1,...,3} {
    \node at (0.6*\d-0.2,0.75*\b-0.1) {\rotatebox{90}{\Large \color{white!70!black} \faMicrochip}};
    \node at (0.6*\d-0.2,0.75*\b-0.1) {\color{white!90!black} {\footnotesize$\sk_\thehsmcount$}};
    \stepcounter{hsmcount}
  }
}

\node[anchor=south west] at (-1.5,2.75) {\textbf{Service provider}};
\node[anchor=south] at (1.4,2.375) {\color{white!70!black} \footnotesize \textbf{HSMs}};
  \node[anchor=south] at (-0.75,0.9) {\color{green!20!black} \footnotesize \parbox{1in}{\setstretch{0.7} \centering \textbf{External\\storage}}};

\draw[thick,->] (-1.25,2.3) to[bend right,out=-40,in=-135] node[status] {$\ct$} (-2.4,1.25);
  \draw[thick,->] (-1.9,0) to[out=0,in=-90] node[below right=-3pt] {\footnotesize 
  \parbox{0.8in}{\raggedright ``$\langle$ Start recovery, \texttt{user=joe}, $\dots \rangle$''}} (-0.75,1); 
\node[clean] at (-1.55,0.05) {\CircC};
\node[clean] at (-1.4,2.3) {\CircB};

\end{tikzpicture}
 }~~~~\vline{}~~~~
  \resizebox{0.22\textwidth}{!}{\tikzset{
  status/.style={fill=white,inner sep=0.5pt,midway,circle},
  textbox/.style={fill=white,inner sep=0.5pt,midway,rectangle},
  clean/.style={fill=white,inner sep=-2pt,circle}
}
\begin{tikzpicture}[use Hobby shortcut]
\useasboundingbox (-2.6,-0.3) rectangle (2,3.1);
\definecolor{few-gray-bright}{HTML}{010202}
\definecolor{few-red-bright}{HTML}{EE2E2F}
\definecolor{few-green-bright}{HTML}{008C48}
\definecolor{few-blue-bright}{HTML}{185AA9}
\definecolor{few-orange-bright}{HTML}{F47D23}
\definecolor{few-purple-bright}{HTML}{662C91}
\definecolor{few-brown-bright}{HTML}{A21D21}
\definecolor{few-pink-bright}{HTML}{B43894}
\definecolor{few-gray}{HTML}{737373}
\definecolor{few-red}{HTML}{F15A60}
\definecolor{few-green}{HTML}{7AC36A}
\definecolor{few-blue}{HTML}{5A9BD4}
\definecolor{few-orange}{HTML}{FAA75B}
\definecolor{few-purple}{HTML}{9E67AB}
\definecolor{few-brown}{HTML}{CE7058}
\definecolor{few-pink}{HTML}{D77FB4}
\definecolor{few-gray-light}{HTML}{CCCCCC}
\definecolor{few-red-light}{HTML}{F2AFAD}
\definecolor{few-green-light}{HTML}{D9E4AA}
\definecolor{few-blue-light}{HTML}{B8D2EC}
\definecolor{few-orange-light}{HTML}{F3D1B0}
\definecolor{few-purple-light}{HTML}{D5B2D4}
\definecolor{few-brown-light}{HTML}{DDB9A9}
\definecolor{few-pink-light}{HTML}{EBC0DA}
 \fill [color=black!10,rounded corners,thick] (-1.5,0.25) rectangle (2,2.75);

\node at (-2.25,0.5) {\color{white!70!black} \fontsize{40pt}{2pt} \selectfont \faMobile*};
\node (db) at (-0.75,2) {\color{green!20!black} \Huge \faDatabase};
\node at (-2.2,0.9) {\color{white!70!black} \footnotesize $\mpk$};
\node at (-2.2,0.62) {\color{white!70!black} \footnotesize $\pin$};

\setcounter{hsmcount}{1}
\foreach \b in {3,...,1} {
  \foreach \d in {1,...,3} {
    \node (hsm-\b-\d) at (0.6*\d-0.2,0.75*\b-0.1) {\rotatebox{90}{\Large \color{blue!70!black} \faMicrochip}};
    \node at (0.6*\d-0.2,0.75*\b-0.1) {\color{white} {\footnotesize$\sk_\thehsmcount$}};
    \stepcounter{hsmcount}
  }
}

\setcounter{hsmcount}{1}
\foreach \b in {3,...,1} {
  \foreach \d in {1,...,3} {
    \draw[thick,densely dotted,red!80!black] ($(db.east) + (-0.1,0.2)$) -- (hsm-\b-\d.center); 
    \stepcounter{hsmcount}
  }
}
\setcounter{hsmcount}{1}
\foreach \b in {3,...,1} {
  \foreach \d in {1,...,3} {
    \node (hsm-\b-\d) at (0.6*\d-0.2,0.75*\b-0.1) {\rotatebox{90}{\Large \color{blue!70!black} \faMicrochip}};
    \node at (0.6*\d-0.2,0.75*\b-0.1) {\color{white} {\footnotesize$\sk_\thehsmcount$}};
    \stepcounter{hsmcount}
  }
}

\node[anchor=south west] at (-1.5,2.75) {\textbf{Service provider}};
\node[anchor=south] at (1.4,2.375) {\color{blue!70!black} \footnotesize \textbf{HSMs}};

  \node[anchor=south] at (-0.75,0.9) {\color{green!20!black} \footnotesize \parbox{1in}{\setstretch{0.7} \centering \textbf{External\\storage}}};
  \node[clean] at ($(db.east) + (0.1,0.4)$) {\CircD};

\end{tikzpicture}
 }~~~~\vline{}~~~~
  \resizebox{0.22\textwidth}{!}{\tikzset{
  status/.style={fill=white,inner sep=0.5pt,midway,circle},
  statusg/.style={fill=black!10,inner sep=0.5pt,midway,circle},
  textbox/.style={fill=white,inner sep=0.5pt,midway,rectangle},
  clean/.style={fill=white,inner sep=-2pt,circle}
}
\begin{tikzpicture}[use Hobby shortcut]
\useasboundingbox (-2.6,-0.3) rectangle (2,3.1);
\definecolor{few-gray-bright}{HTML}{010202}
\definecolor{few-red-bright}{HTML}{EE2E2F}
\definecolor{few-green-bright}{HTML}{008C48}
\definecolor{few-blue-bright}{HTML}{185AA9}
\definecolor{few-orange-bright}{HTML}{F47D23}
\definecolor{few-purple-bright}{HTML}{662C91}
\definecolor{few-brown-bright}{HTML}{A21D21}
\definecolor{few-pink-bright}{HTML}{B43894}
\definecolor{few-gray}{HTML}{737373}
\definecolor{few-red}{HTML}{F15A60}
\definecolor{few-green}{HTML}{7AC36A}
\definecolor{few-blue}{HTML}{5A9BD4}
\definecolor{few-orange}{HTML}{FAA75B}
\definecolor{few-purple}{HTML}{9E67AB}
\definecolor{few-brown}{HTML}{CE7058}
\definecolor{few-pink}{HTML}{D77FB4}
\definecolor{few-gray-light}{HTML}{CCCCCC}
\definecolor{few-red-light}{HTML}{F2AFAD}
\definecolor{few-green-light}{HTML}{D9E4AA}
\definecolor{few-blue-light}{HTML}{B8D2EC}
\definecolor{few-orange-light}{HTML}{F3D1B0}
\definecolor{few-purple-light}{HTML}{D5B2D4}
\definecolor{few-brown-light}{HTML}{DDB9A9}
\definecolor{few-pink-light}{HTML}{EBC0DA}
 \fill [color=black!10,rounded corners,thick] (-1.5,0.25) rectangle (2,2.75);

\node[fill=pink,rounded corners, inner sep=1.5pt] at (-2.2,0.3) {\footnotesize $\msg$};
\node at (-2.25,0.5) {\fontsize{40pt}{2pt} \selectfont \faMobile*};
\node (db) at (-0.75,2) {\color{green!20!black} \Huge \faDatabase};
\node at (-2.2,0.9) {\footnotesize $\mpk$};
\node at (-2.2,0.62) {\footnotesize $\pin$};

\draw [dashed,
      fill=yellow,
      fill opacity=.45]
    (1,0.35) .. (0.5, 0.35) .. (0.2,0.35) .. (0.25,1).. (1.3,1.1) .. (1.3,1.4) .. (1.3, 1.55) .. (1.75,1.75) .. (1.95, 1.65) .. (1.95, 1.4) .. (1.95, 1) .. (1.95, 0.75) .. (1.7, 0.35) .. (1.4, 0.35) .. (1,0.35);

\setcounter{hsmcount}{1}
\foreach \b in {3,...,1} {
  \foreach \d in {1,...,3} {
    \node at (0.6*\d-0.2,0.75*\b-0.1) {\rotatebox{90}{\Large \color{blue!70!black} \faMicrochip}};
    \node at (0.6*\d-0.2,0.75*\b-0.1) {\color{white} {\footnotesize$\sk_\thehsmcount$}};
    \stepcounter{hsmcount}
  }
}

\node[anchor=south west] at (-1.5,2.75) {\textbf{Service provider}};
\node[anchor=south] at (1.4,2.375) {\color{blue!70!black} \footnotesize \textbf{HSMs}};
\node[anchor=south] at (-0.75,0.9) {\color{green!20!black} \footnotesize \parbox{1in}{\setstretch{0.7} \centering \textbf{External\\storage}}};

\draw[thick,->] (-1.9,0.6) to[out=0,in=-195] node[statusg] {\footnotesize $\ct, \pi$} (0,0.55);
  \draw[thick,->] (1.05,0.35) to[out=-120,in=0] node[below=0.2ex,textbox] {\footnotesize $\mathsf{Shares\ of\ decrypted\ }\ct$} (-1.7,0.2);

\draw[thick,->] (-1.25,1.8) to[bend right] node[status] {$\pi$} (-2,1.25);
\node[clean] at (-1.4,1.8) {\CircE};
\node[clean] at (-1.55,0.6) {\CircF};
\node[clean] at (0.9,0.1) {\CircG};

\end{tikzpicture}
 }
    \caption{
  An overview of the recovery-protocol flow.
  Each HSM $i$ holds a secret key $\sk_i$.
  The client holds a vector $\mpk$ of all HSMs' public keys.
  \CircA{} During backup, the client uses its PIN and the master public key 
            to encrypt its data $\msg$ into a recovery ciphertext $\ct$.
            The client then uploads this recovery ciphertext $\ct$ to the 
            service provider.
  \CircB{} During recovery, the client downloads its recovery ciphertext. 
  \CircC{} The client asks the data center to log its recovery attempt.
  \CircD{} The service provider collects a batch of client log-insertion requests,
           updates the log, and aggregates the new log into a Merkle tree.
           The service provider and HSMs run a log-update protocol.
           At the end of this protocol, each HSM holds the root of the Merkle
           tree computed over the latest log.
  \CircE{} The service provider sends the client a Merkle proof $\pi$ that the client's recovery
           attempt is included in the latest log (i.e., in the latest Merkle root).
  \CircF{} The client sends the recovery ciphertext $\ct$ and log-inclusion proof $\pi$
           to the subset of HSMs needed to decrypt the recovery ciphertext. 
  \CircG{} The HSMs check the proof and return shares of the decrypted ciphertext 
            to the client. The client uses these to recover the backed-up data $\msg$.
  }
  \label{fig:overview}
\end{figure*}
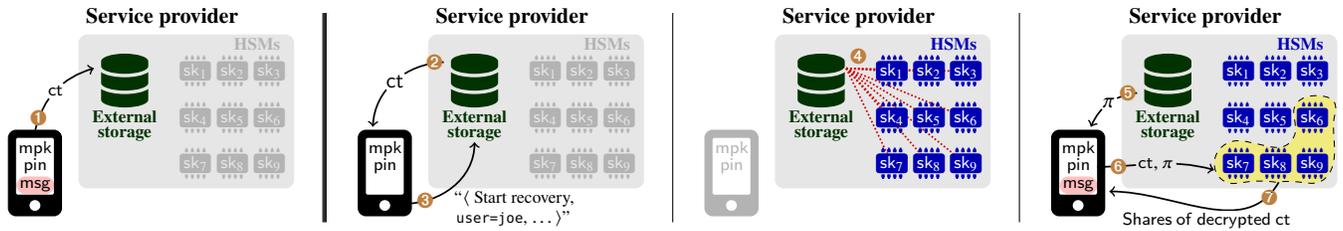

\section{Architecture overview}
\label{sec:arch}

We now describe our encrypted-backup protocol (\cref{fig:overview}) and
explain how it satisfies the design goals
of \cref{sec:design}.
We will discuss possible extensions and deployment considerations in \cref{app:deploy}.

\subsection{The back-up process}
The client begins the back-up process holding 
\begin{itemize}
\item the public keys of all HSMs in the data center,
\item its secret PIN, and 
\item a disk image to be backed up (the ``message'').
\end{itemize}
To back up its disk image, 
the client samples a subset of $n$ HSMs out
of the $N$ total HSMs in the data center where $n \ll N$.
The client chooses this subset by hashing 
(a) public information: the service name, its username, and a public
salt the client chooses at random, and (b) its secret PIN.
The client then encrypts its message with a random
AES encryption key, and then splits this AES key into $n$ threshold 
shares using Shamir secret sharing~\cite{S79}, such that 
any threshold $t$ of the shares suffice to recover the AES key.
The client prepends each share with 
the client's username 
to ensure that the
ciphertexts are bound to the client's username.
The client then encrypts one share to the public
key of each HSM in its chosen subset.

The client's recovery ciphertext then consists of:
its public salt, the AES-encrypted message, 
the $n$ encrypted shares of the AES key,
and a configuration-epoch number that the service
provider can use to identify the set of HSMs that were
in service at the time the client created its backup.
The client computes the ciphertext locally and uploads it to the
backup service provider, with no HSM interactions required.

To explain why this construction is scalable: 
since only a constant
number of HSMs $n \ll N$
participate in the decryption process, the system
scales well as the number of HSMs in the data
center increases.

To explain why this construction should be secure:
if the attacker cannot guess the client's PIN, the
attacker does not know which set of $n$ 
HSMs (out of the $N$ total)
it needs to compromise to recover the
client's AES key. So, the best attacks are either
to: guess the client's PIN or compromise a large
fraction of the data center.

This argument requires that each individual key-share ciphertext leak
no information about which HSM can decrypt it---a cryptographic property known as
``key privacy''~\cite{BBDP01}.  However, even key-private
encryption schemes do not always remain secure against an
adversary that adaptively compromises secret keys, which leads to
our first technical challenge:

\begin{challenge}
How can we ensure that the client's recovery
ciphertext ``leaks nothing'' about which HSMs are
required to decrypt the client's message, even
against an attacker who can adaptively compromise HSMs?
\end{challenge}
In \cref{sec:hide}, we explain how to solve this problem
using \emph{location-hiding encryption}, a new
cryptographic primitive. 

\subsection{The recovery process}
The client begins the recovery process holding:
\begin{itemize}
\item the public keys of all HSMs in the data center,
\item its secret PIN, and 
\item its recovery ciphertext (which the client can fetch
      from the service provider).
\end{itemize}

First, the client asks the service provider to record
its recovery attempt
in the \emph{append-only log}, implemented collectively by the 
service provider and HSMs.
The log holds a mapping of identifiers to values.
The service provider can insert new identifier-value pairs 
into the log but 
the service provider
cannot modify or delete the values of defined identifiers,
ensuring that there is at most one immutable value for each identifier.

The recovery attempt is logged as follows.  
The client begins by using
public information (service name, username, and
salt in the recovery ciphertext) along with its secret PIN
to recover the subset of $n$ HSMs
it picked during backup.
The client then hashes these values together with some randomness
to produce a cryptographic commitment $h$
to the identities of these HSMs and to its recovery ciphertext.
The client then asks the service provider to insert 
the identifier-value pair $(\username, h)$ into the log, where $\username$ 
is the client's username.
(In this discussion, we use the client's username as the key for simplicity.
In practice, to preserve privacy, we might use an opaque
device-install UUID\@.)

The service provider collects a batch of these log-insertion
requests, produces a Merkle-tree~\cite{M89} digest over the
updated log, and runs a log-update protocol with the HSMs.
At the end of this protocol, the HSMs hold the updated log digest.
The service provider then returns to the client a Merkle
proof $\pi$ proving that the pair
$(\username, h)$ appears in the latest log digest.

Since the service provider and HSMs run the log-update protocol
periodically (e.g., every 10 minutes), the client will have
to wait a few minutes on average to decrypt its backup. 
The client already has to download its large encrypted disk image,
which will likely take minutes, so these steps can proceed in parallel.

The client then contacts its chosen set of $n$ HSMs over an encrypted channel, such as TLS.
The client sends to each HSM: 
its username, the opening of its commitment $h$ 
(i.e., the values and randomness used to construct the commitment $h$), and
the Merkle inclusion proof~$\pi$.
Each HSM 
\begin{itemize}
\item recomputes the commitment $h$ and
      checks the inclusion proof $\pi$ (to confirm that the recovery attempt is logged), and
\item decrypts its share of the client's AES key, 
      confirms that the username in the decrypted plaintext
      matches the one provided by the client 
      (which prevents user $A$ from attempting to decrypt user $B$'s ciphertext,
      in collusion with a malicious service provider).
\end{itemize}
If both of these checks pass, the HSM returns the AES-key
share to the client.

Given any $t$ of these decryption-key shares, the client can
recover the AES key used to encrypt its backup.
The client can then use this AES key to decrypt its backed-up
message.

Since at most one log entry can exist per
username, the use of the log ensures that each
user can make at most one recovery attempt. In
this way, the system defeats brute-force
PIN-guessing attacks. With a slight modification,
it is possible to allow each user to make a fixed
number (e.g., 5) guesses, or a fixed number of
guesses per time period (e.g., 5 per
month).

A counter-intuitive property of this scheme is that the
client never explicitly provides its PIN to the HSMs. 
The fact that the client knows which subset of the HSMs to contact
implicitly proves the client's knowledge of the PIN because
the set of $n$ HSMs is much smaller than the total number of HSMs
$N$.

\medskip

This overview leaves some technical details unexplained.
In particular:
\begin{challenge}
How do the HSMs implement the append-only log 
without sacrificing scalability or security?
\end{challenge}
A straightforward way to implement the log would be to have
each HSM store the entire state of the log.
But then every HSM would have to participate in every recovery
attempt, which would not meet our scalability goals.
Another implementation would be to have the data-center operators
maintain the log, but then malicious data centers could violate
the append-only property, and thus mount brute-force PIN-guessing
attacks, without HSMs noticing.

In \cref{sec:log}, we explain how the HSMs can collectively
maintain such an append-only log in a scalable and secure manner.
At a high level, the (potentially adversarial) data center maintains 
the state of the log, which we represent as a list of identifier-value pairs.
Every time the data center wants to insert an identifier-value pair into the log, 
the data center must prove to a random subset of the HSMs that the
identifier to be inserted is undefined in the current log.
Provided that at least one honest HSM audits each log-insertion, we can
guarantee that the values associated with log identifiers are immutable
(i.e., that we maintain the log's append-only property). 
In this way, (a) each HSM needs to participate in only a vanishing fraction
of the recovery attempts and (b) even an attacker who can compromise many
of the HSMs cannot break the append-only nature of the log.

\medskip

One remaining issue is that an attacker who observes the data center
network may see which HSMs a client interacts with during recovery and
decide to compromise that exact set of HSMs after recovery completes.

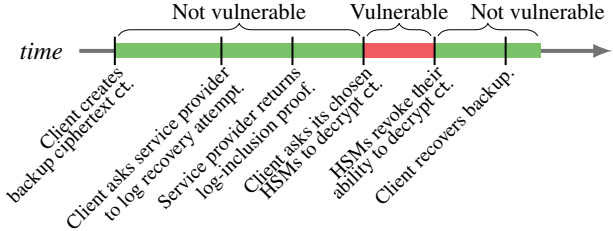
\begin{figure}
  \resizebox{\columnwidth}{!}{\begin{tikzpicture}[
    Larrow/.style={>=latex, ->, draw=few-gray, text=black,ultra thick},
    timelabel/.style={font=\fontsize{8}{6}\selectfont,align=right,rotate=45,below left},
    boxlabel/.style={font=\small,align=center,above,yshift=3pt},
  ]
\coordinate (arrowStart) at (0,0);
\coordinate (arrowEnd) at (7.5,0);
\draw[Larrow] (arrowStart) -- node[above,draw=none, inner ysep=0pt]{} (arrowEnd);
\node[anchor=east] (time) at ([xshift=-3pt] arrowStart) {\emph{time}};

\def\timept#1#2{
  \draw[thick] ($ (arrowStart) + (#1, 0.2) $) -- node[timelabel]{#2} ($ (arrowStart) + (#1, -0.2) $);
}

\def\rect#1#2#3#4#5{
  \draw[draw=none,fill=#3] ($ (arrowStart) + (#1, 0.1) $) rectangle ($ (arrowStart) + (#2, -0.1) $);
  \draw[decorate,decoration={brace,amplitude=5pt}] ($ (arrowStart) + (#1 + 0.05, 0.21) $) -- node[xshift=#5,boxlabel]{#4} ($ (arrowStart) + (#2 - 0.05, +0.21) $);
}

\rect{0.5}{4}{few-green}{Not vulnerable}{0};
\rect{4}{5}{few-red}{Vulnerable}{0};
\rect{5}{6.5}{few-green}{Not vulnerable}{20pt};

\timept{0.5}{Client creates\\ backup ciphertext $\ct$.}
\timept{2}{Client asks service provider\\
          to log recovery attempt.}
\timept{3}{Service provider returns\\ log-inclusion proof.}
\timept{4}{Client asks its chosen\\ HSMs to decrypt $\ct$.}
\timept{5}{HSMs revoke their\\ ability to decrypt $\ct$.}
\timept{6}{Client recovers backup.}

\end{tikzpicture}
 }
  \\[-32pt]
  \caption{
    Since HSMs in \name revoke their ability to
    decrypt a client's recovery ciphertext, \name
    protects against HSM compromise attacks that
    take place before recovery begins and after it
    completes. An attacker who can compromise HSMs
    while recovery is in progress can break
    security.
}
  \label{fig:timeline}
\end{figure}

\begin{challenge}
For scalability, the client should only communicate with a small
number of HSMs during recovery. But then how can we protect against
an attacker who compromises these HSMs after recovery completes?
\end{challenge}

Our idea is as follows: after a client runs the recovery protocol, 
each participating HSM \emph{revokes} its ability to decrypt that client's
recovery ciphertext.
So, even if an after-the-fact attacker compromises the HSMs that participated
in recovery, the attacker learns no useful information. 
The only window of vulnerability is at the moment
after the client contacts its HSMs and 
before the HSMs complete revocation (\cref{fig:timeline}).
We describe how to make this work on resource-limited HSMs in
\cref{sec:punc}.
 \section{Protecting the mapping of users to HSMs\\
with location-hiding encryption}
\label{sec:hide}

In this section, we define and construct \emph{location-hiding encryption},
which the client uses to encrypt its backup data.

The location-hiding encryption routine takes as input
(1) a set of $N$ public keys, 
(2) a short PIN, and 
(3) a message,
and outputs a ciphertext.
In our application, the $N$ public keys are the public keys
of the $N$ HSMs in the data center.

The cryptosystem has three main properties, which
we formalize in \iffull \cref{sec:hide-defs}: \fi

\itpara{1. Security.}
To successfully decrypt the ciphertext, an attacker
must either (a) guess the PIN or (b) control more than a
constant fraction $\fevil$ of the $N$ total secret keys.
This security property must hold even if the attacker can
adaptively compromise an $\fevil$ fraction of the
$N$ secret keys.
In our application, this implies that unless an adversary can guess
the PIN or compromise a constant $\fevil$ 
fraction of the HSMs in the data center, 
it learns nothing about the client's backed-up data.

\itpara{2. Scalability.}
Given the PIN used to encrypt the message, it
is possible to decrypt the message using a
small subset of the $N$ secret keys 
corresponding to the $N$ public keys used during encryption.
In our application, a client who knows the correct PIN
can recover its backup by interacting with only 
a small cluster of $n$ HSMs (for some parameter $n \ll N)$ 
out of the $N$ total HSMs,
So as $N$ grows, each HSM needs to participate in a vanishing fraction
of the total recovery attempts.

\itpara{3. Fault tolerance.}
Given the PIN, it is possible to decrypt a ciphertext 
even if a random fraction $\ffail$ of all secret keys
are unavailable.
In our application, this implies clients can recover their
backups even if an $\ffail$ fraction of all HSMs fail.

\smallskip
We call this primitive ``location-hiding encryption''
because there is a small set of $n$ HSMs
that the attacker could compromise to decrypt the 
ciphertext, but the cryptosystem \emph{hides the location}
of these HSMs within the larger pool of $N$ HSMs.

\subsection*{Our construction}
\label{sec:hide:const}

Our construction of location-hiding encryption 
is just a careful composition of existing primitives.
However, it takes some analysis to prove that the
composition provides the desired security
properties. 
We describe our construction here in prose
and we include the security definitions and proofs in \iffull \cref{sec:hide-defs}.
\fi
The construction makes use of a public-key encryption scheme
(hashed ElGamal encryption~\cite{E85,BSbook})
and an authenticated encryption scheme (e.g., AES-GCM).

\paragraph{Setup.}
In our construction, each 
HSM $i$, for $i \in [N]$, holds a keypair $(\pk_i, \sk_i)$ for the
public-key encryption scheme.
Let $t \in \N$ be a threshold such that if each HSM
fails with probability $\ffail$, then in a random sample of
$n$ HSMs, there are at least $t$
non-failed HSMs with extremely high probability.
Our instantiation takes $t = n/2$
for $\ffail = \tfrac{1}{64}$.

\paragraph{Encryption.}
The encryption routine takes as input
a list of $N$ public keys $(\pk_1, \dots, \pk_N)$,
a PIN, and a message $\msg$.
To encrypt the message using our location-hiding encryption scheme:
\begin{enumerate}
  \item Sample a random AES key $k$ and a random salt.
  \item Split $k$ into $t$-out-of-$n$-Shamir 
        secret shares $k_1, \dots, k_n$~\cite{S79}.
  \item Hash the PIN and salt and use the result as a seed
        to generate a list of $n$~random indices
        $i_1, \dots, i_n \in [N]$.
  \item Encrypt each key-share $k_j$ with public key $\pk_{i_j}$.
  \item Finally, return (a) the salt, (b) the $n$ public-key
    ciphertexts, and (c) the AES encryption of $\msg$ under key $k$.
\end{enumerate}

\paragraph{Decryption.}
To decrypt given the ciphertext and PIN:
\begin{enumerate}
  \item Hash the salt and PIN to reconstruct the set
        of indices $i_1, \dots, i_n \in [N]$ used during encryption.
  \item Use secret keys $\sk_{i_1}, \dots, \sk_{i_n}$ to decrypt
        the $n$ shares of the AES key $k$.
        (In fact, only $t$ of the shares are necessary.)
  \item Using the recovery routine for Shamir secret sharing, recompute
        the AES key $k$ from its shares.
  \item Decrypt and return $\msg$ using the AES key $k$.
\end{enumerate}

Notice that the decryption routine only uses the
PIN to sample the set of secret keys used for
decryption.
In our application, this implies that
the client never needs to explicitly provide its PIN (or
even a hash of its PIN) to the HSMs; contacting the
right subset of HSMs is enough to ensure that the
client provided the correct PIN.

The intuition behind the security analysis is straightforward:
with hashed ElGamal encryption,
the ciphertext reveals no information about which $n$ public
keys (out of the $N$ total where $n \ll N$) were used during encryption.
Thus, the ciphertext reveals no information about
which secret keys the attacker must compromise unless the attacker can
guess the PIN\@.
Without these secret keys, the attacker cannot learn anything about $k$,
and therefore cannot decrypt the message.

\iffull
The following theorem, which we prove as \cref{thm:lhe-sec}
in \cref{sec:hide-defs:proof} makes this argument precise:

\smallskip
\paragraph{Theorem (Informal).}
\emph{The location-hiding encryption scheme of \cref{fig:lhe-const}
instantiated with the hashed ElGamal encryption scheme (\cref{sec:hide-defs:elgamal}) 
over a group $\G$
is \emph{secure} (in the sense of \cref{defn:lhe-sec}) for certain values of $n$ and $N$, 
provided that:
\begin{itemize}
  \item the computational Diffie-Hellman problem
is hard in $\G$, 
  \item the authenticated-encryption scheme is secure, and
  \item we model the hash functions used in the construction as random oracles.
\end{itemize}}
\fi

There are two reasons why the security analysis is non-trivial:
First, we must ensure that the ciphertext leaks nothing
about the $n$ keys to which it was encrypted (i.e., that it
is \emph{key-private}~\cite{BBDP01}).
Second, we must ensure that the encryption scheme
remains secure even if an attacker can adaptively compromise
secret keys. This is known as security under 
\emph{selective-opening attack}~\cite{BHY09,FHKW10,HR14}.
Showing that both properties hold at once is the source of
the technical complexity.
 \section{The distributed log}\label{sec:log}

\begin{figure*}
  {\centering {\small \begin{tikzpicture}[
    agent color/.style={draw=few-blue, fill=few-blue-light!25},
    agent/.style={agent color, thick, rounded corners,
      minimum height=6cm, text width=6cm},
    msg/.style={line width=1.5mm, >=latex, few-blue, text=black},
    actiontext/.style={anchor=center, text width=5.5cm, align=left},
    Larrow/.style={>=latex, <->, draw=few-blue-bright, text=black},
    tree/.style={draw=few-purple,very thick},
    log/.style={thick, draw=few-blue-bright, fill=few-green},
    new log/.style={log, fill=few-green!50},
  ]
  \node[agent, label=above:{\bfseries Service provider}] (provider) {};
  \node[agent, right=4cm of provider, label=above:{\bfseries HSM$_i$}] (hsm) {};

\def\msg#1#2#3{
  \draw[msg,#2] ([yshift=-#1] provider.north east) --
  node[above,align=center]{#3} ([yshift=-#1] hsm.north west);
}
\def\act#1#2#3{
  \node[actiontext] at ([yshift=-#1] #2.north) {#3};
}

\msg{5mm}{->}{$(d,d',R)$}
\act{7mm}{provider}{Holds stale log $L$, new log $L'$.\\
Build Merkle tree over log-chunk digests and extension proofs.}
\act{5mm}{hsm}{Holds stale digest $d = \Digest(L)$.}

\msg{15mm}{<-}{``Audit $(1,3)$''}
\act{15mm}{hsm}{Choose $\lambda$ random chunks\\
  in the range $\{1,\ldots,N\}$ to audit.}

\msg{35mm}{->}{$(d,d_1,\pi_1),(d_2,d_3,\pi_3)$,\\ $\langle$Merkle
  proofs that these\\ values are included in $R\rangle$}
  \act{35mm}{hsm}{If $\VerifyExtends(d,d_1,\pi_1)$ and\\
  \ \ \ $\VerifyExtends(d_2,d_3,\pi_3)$,\\
  then sign $(d,d',R)$ using secret key.}

\msg{45mm}{<-}{Signature $\sigma_i$}

\coordinate (logpos) at ([yshift=-45mm] provider.north) {};

\msg{55mm}{->}{Aggregate $\sigma$}
\act{54mm}{provider}{Collect sigs from all online HSMs.\\
Aggregate signatures into $\sigma$.}
\act{54mm}{hsm}{If $\sigma$ is a valid signature of $(d,d',R)$ under
  aggregate public key,\\then accept $d'$ as new digest.}

\node[log, anchor=south, minimum width=5.5cm, minimum height=1cm]
  (log) at ([yshift=5pt] logpos) {};
\node[new log, draw=none, inner sep=0pt, outer sep=0pt,
  name=lognew, fit={(log.south east)
    ($(log.north west)!.4!(log.north east)$)},
] {};
\node[log, anchor=south, minimum width=5.5cm, minimum height=1cm, fill=none]
  (log) at ([yshift=5pt] logpos) {};

\foreach \x in {0,1,...,5} {
  \coordinate (pi\x) at ($(lognew.north west)!\x/5.0!(lognew.north east)$) {};
}
\foreach \x in {0,1,2,3,5} {
  \draw[new log, very thick] (pi\x) -- (pi\x |- lognew.south);
}

\draw[Larrow] (logpos -| log.west) --
    node[agent color,draw=none, inner ysep=0pt]{$L$} (logpos -| pi0);
\draw[Larrow,transform canvas={yshift=-5pt}]
      (logpos -| log.west) -- node[agent color,draw=none]{$L'$} (logpos -| pi5);

\foreach \x/\d/\p in {0/d/,
  1/d_1/\pi_1, 2/d_2/\pi_2, 3/d_3/\pi_3,
  4//\cdots, 5/d'/\pi_N } {
  \node[above=0pt of pi\x, inner ysep=2pt] {$\p$};
  \node[above=1.5ex of pi\x, inner sep=0pt, outer sep=0pt] (d\x) {$\strut\d$};
}

\def\rparent#1(#2)(#3)#4; {
  \coordinate (xtmp) at ($(#2)!.5!(#3)$) {};
  \coordinate (ytmp) at ($(#3) + (0,4mm)$) {};
  \node[outer sep=0pt, inner sep=1pt] (#1) at (xtmp |- ytmp) {#4};
  \draw[tree] (#3) -- (#1);
}
\def\parent#1(#2)(#3)#4; {
  \rparent#1(#2)(#3)#4;
  \draw[tree] (#2) -- (#1);
}
\parent d0d1(d0)(d1);
\parent d2d3(d2)(d3);
\parent d0d4(d0d1)(d2d3);
\rparent d4d5(d4)(d5);
\rparent d4d4d5( $(d4)!.5!(d3)$ |-d2d3)(d4d5);
\parent root(d0d4)(d4d4d5){$R$};
\end{tikzpicture}
 \\}}
\caption{The protocol that the service provider and HSMs use to update the HSM's log digest.}
\label{fig:logflow}
\end{figure*}
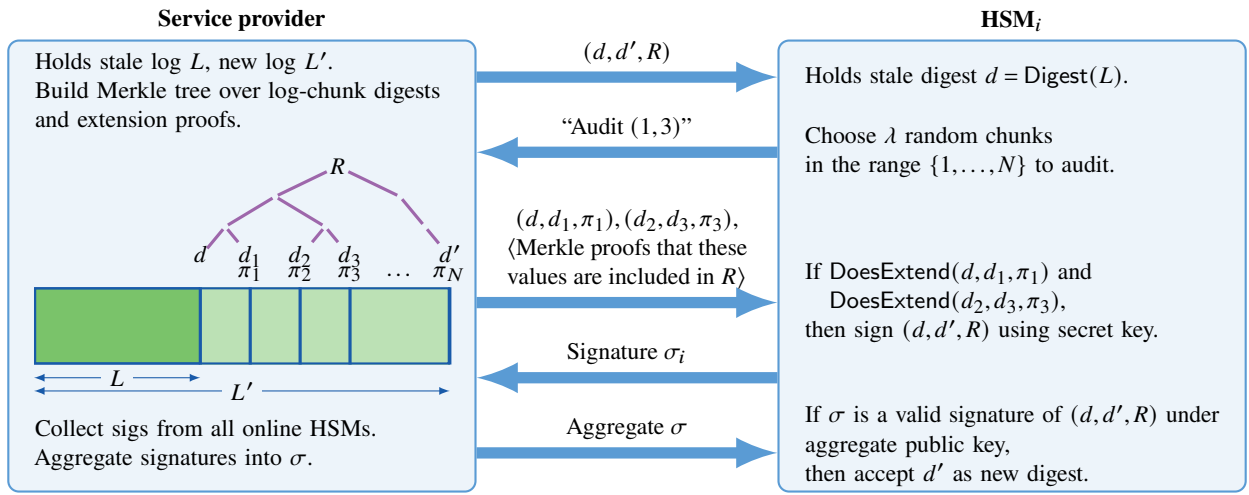

In \name, the HSMs collectively maintain a \emph{distributed log},
which any external party can read and replay.
The service provider maintains the log state and
the HSMs monitor log insertions to ensure that the
service provider does not violate the log's
append-only property. 

We use this log for two primary purposes:
\begin{enumerate}
\item \textbf{Limiting PIN guesses.}
To prevent an attacker from brute-force guessing
a client's PIN, we use the log
(as described in \cref{sec:arch})
to enforce a global limit on the number of recovery attempts
that the HSMs allow per username.  

\item \textbf{Monitoring recovery attempts.}
The service provider logs each recovery attempt, so any \name client
can inspect the log to learn whether someone
(e.g., a foreign attacker or snooping
acquaintance) has tried to recover their backed-up data.
A client could then take mitigating action---such as contacting
their service provider, a law-enforcement agency, or the press.

\end{enumerate}

A third use for the log---which comes directly from 
related work~\cite{ford2006persistent}
and which we have not yet implemented---is to
\textbf{manage HSM group membership}.
Whenever the service provider wants to add or remove an 
HSM from the data center, the service provider operator could
record this information in the log before the other HSMs
will accept the change.
All \name clients can thus verify that they are
communicating with the same set of HSMs.
In addition, clients can also detect suspicious
changes in the set of HSMs in the data center.
(For example, if the service provider replaces all
HSMs in the data center over the course of a day.)

The log is simply a list of identifier-value pairs maintained by
the service provider. Clients can insert identifier-value
pairs in order to record recovery attempts, and HSMs
maintain a digest of the log state.
Our distributed log must satisfy the following key property:
\begin{quote}
\textit{
If any honest HSM ever accepts that an
identifier-value pair $(\id, \val)$ is included in the log,
the HSM should never accept that $(\id, \val')$ is included
in the log, for any value $\val' \neq \val$.
}
\end{quote}

\subsection{Underlying data structure}
\label{sec:log:syntax}

\itpara{Terminology.}
The log $L$ is a list of key-value pairs.
Since we use the word ``key'' in this paper to refer to cryptographic keys, 
we call log keys ``identifiers.''
We say that a log $L'$ ``extends'' a log $L$ if
(a) $L$ is a prefix of $L'$ and 
(b) every identifier in $L'$ appears at most once.

Our distributed log uses an
authenticated data structure~\cite{Tamassia03,NN98,aad} 
that implements the following five routines:

\smallskip
\begin{hangparas}{\defhangindent}{1}

$\Digest(L) \to d$. Return a constant-size digest $d$ representing the current state of the log.

$\ProveIncludes(L, \id, \val) \to \{\pfinc, \bot\}$. 
Output a proof $\pfinc$ that attests to the fact that 
the identifier-value pair $(\id, \val)$ is 
in the log represented by digest $d = \Digest(L)$.

$\VerifyIncludes(d, \id, \val, \pfinc) \to \zo$. 
Return ``$1$'' iff $\pfinc$ proves that
the log that digest $d$ represents contains $(\id, \val)$. 

$\ProveExtends(L, L') \to \{\pfext, \bot\}$.
Output a proof $\pfext$ that $d' = \Digest(L')$ represents a log that 
extends the log that digest $d = \Digest(L)$ represents.

$\VerifyExtends(d, d', \pfext) \to \zo$. 
Return ``$1$'' iff $\pfext$ proves that
the log that digest $d'$ represents extends
the log that digest $d$ represents.

\end{hangparas}

\smallskip
\noindent

The inclusion and extension proofs must be
complete (honest verifiers accept valid proofs) and
sound (honest verifiers reject invalid proofs),
as we define in \iffull \cref{app:logx:sec}. \fi

\paragraph{Implementing the data structure.}
Nissim and Naor~\cite{NN98} show that it is possible to implement 
these log primitives using only Merkle trees~\cite{M89}. 
We summarize their construction in \iffull \cref{app:logx:impl}.
\fi
At a very high level: the digest of the log is just the root of 
a Merkle tree computed over all of the entries of the log,
represented as a binary search tree indexed by $\id$.
A log-inclusion proof $\pfinc$ is a Merkle proof of inclusion
relative to this root.
A log-extension proof $\pfext$ is a proof that: 
(1) every identifier inserted to the new log did not exist in the old log and 
(2) the new digest represents the old log tree with the new values inserted.
It is possible to prove both assertions using a number of Merkle proofs
proportional to the number of log insertions.

\subsection{Building a distributed log}

We now explain how to use the primitives of \cref{sec:log:syntax} to build our
distributed append-only log.

\paragraph{Initializing the log.}
The service provider maintains the entire state of the log $L$.
Each HSM stores a log digest $d$ which,
in steady state, is the digest of the log $L$
that the service provider holds.
Initially, the log $L$ is empty and each HSM holds the digest of the empty log.

\paragraph{Inserting into the log.}
A client can insert an entry $(\id, \val)$ into the log
by simply sending the pair to the service provider.
The service provider adds this entry to its log state $L$.

\paragraph{Proving log membership to HSMs.}
Before the HSMs allow a client to begin the recovery process, the HSMs
require proof that the client's recovery attempt is logged.
Assume for the moment that the service provider holds a log $L$ 
and all HSMs hold the up-to-date digest $d = \Digest(L)$.
(We will explain how the HSMs get the latest log digest in a moment.)
Then, a client can prove inclusion of any pair $(\id, \val)$ in the
log by asking the service provider for an inclusion proof.
The service provider computes $\pfinc = \ProveIncludes(L, \id, \val)$
and returns the inclusion proof to the client.
The client then sends $(\id, \val, \pfinc)$ to the HSM, 
which can check $\VerifyIncludes(d, \id, \val, \pfinc)$ to be
convinced that $(\id, \val)$ is in the log represented by its digest~$d$.
This inclusion check is fast---logarithmic in the log length.

\paragraph{Updating the log digest at the HSMs.}
After a sequence of log-insertions, the service provider holds a log state $L'$.
The HSMs will be holding a digest $d = \Digest(L)$ of a stale log $L$. 
If the service provider is honest, the new log $L'$ extends the old log $L$.

To update the log digest at the HSMs, the service provider will first 
send the new digest $d' = \Digest(L')$ to every HSM.
Next, the data center must convince each HSM that this new digest $d'$
represents a log that extends the log $L$ that the old digest $d$ represents.

One non-scalable way to achieve this would be for the service provider
to send an extension proof $\pfext = \ProveExtends(L, L')$ to every HSM.
The problem is that the time required to check this extension proof grows
linearly with the number of new log entries. 
So if every HSMs checked the entire extension
proof, the throughput of the system would not
increase as the number of HSMs increases.

Instead, we use a randomized-checking approach, as in \cref{fig:logflow}.
If there have been $I$ insertions to the log since the last update,
the service provider divides the updates into $N$ chunks, each containing
$I/N$ insertions.
The service provider then applies these chunks of updates to the old 
log $L$ one at a time, producing a digest $d_i$ and extension proof $\pi_i$ for 
each of the $N$ intermediate logs ($i \in \{1, \dots, N\}$).
The service provider then 
sends the root $R$ of a Merkle-tree commitment to these digests to each HSM.

Each HSM then asks the service provider 
for a random $\lambda$-size subset of the intermediate digests and extension proofs,
where $\lambda$ is a security parameter.
The service provider returns the requested digests and extension proofs and proves
that these values are included in the Merkle root $R$.
Each HSM checks its requested intermediate extension proofs using $\VerifyExtends(\cdot)$ 
and checks the Merkle proof relative to the root $R$.
The HSMs auditing the first and last chunks also ensure that the intermediate
digests match the old digest $d$ and 
the new digest $d'$, respectively.

If these extension and Merkle proofs are valid, each HSM signs the tuple $(d, d', R)$
using an aggregate signature scheme~\cite{BGLS03},
and returns the signature to the service provider.
Once all online HSMs have signed, 
the service provider aggregates these signatures and
broadcasts the aggregated signature to all HSMs.
If any HSM fails during this process,
the service provider notifies the HSMs and they
restart this log-update process. 
(In \iffull \cref{sec:logx:fail}, \fi
we describe how the log can make
progress even if HSMs fail during the log-update protocol.)

The HSMs check the aggregate signature on $(d, d', R)$ relative
to the HSMs' aggregate public key.
If the signature is valid, the HSMs accept the new digest $d'$.

\itpara{Security.}
If there are at most $\fevil$ compromised HSMs,
then even if $\fevil$ honest HSMs are slow,
$(1 - 2 \fevil) N$ honest HSMs will participate in any successful
protocol execution.
If each of these HSM audits $C$ chunks, then the probability
that no honest HSM audits a particular log chunk is
\[\Pr[\text{fail}] = \big(1 - \tfrac{1}{N}\big)^{(1 - 2 \fevil)N \cdot C} \leq \exp\big((2\fevil - 1) \cdot C\big).\]
(Here, we use the fact that $(1-x) \leq \exp(-x)$.)
If each HSM audits $C = \lambda \approx 128$ chunks, this failure probability is $\ll 2^{-128}$.
In other words, some honest HSM will catch
a cheating service provider with overwhelming
probability.
In addition, since all honest HSMs will expect a signature from all honest HSMs, 
this will cause the updating operation to fail and the system to halt.
For this analysis, we assume that the adversary cannot adaptively compromise HSMs
while the recovery protocol is running without taking them offline.

\itpara{Scalability.}
Each HSM must check the extension proofs on $\lambda$ chunks, where each chunk contains a
$1/N$ fraction of the total updates in each epoch.
Thus each HSM checks a vanishing fraction ($\frac{\lambda}{N}$) 
of log insertions.
Each HSM checks one aggregate signature, which requires
time independent of the number of HSMs~\cite{BGLS03}.
Thus, the total work that each HSM performs per epoch decreases as 
the number of HSMs $N$ increases.

\smallskip
Because we use the log primarily to limit the number of PIN attempts,
garbage collection is straightforward.
The service provider simply creates a new empty log, effectively
resetting the number of PIN attempts for every user (old copies of the
log can still be inspected to monitor recovery attempts).
To ensure that the service provider does not run garbage collection and
clear the state too frequently, each HSM will run garbage collection for
a fixed number of times (e.g. the expected number of garbage collections
over two years) before refusing to respond to further requests. This
bounds the number of times the service provider can garbage collect the log.

\subsection{Transparency and external auditability}
\label{sec:log:audit}

Our log design allows \emph{anyone} to audit the
log to ensure that the service provider correctly maintains the log's
append-only property.
Additional auditors only add to the security of the system by adding
another layer of protection, as they can detect log corruptions in the event that
more than $\fevil$ HSMs are compromised.
In particular, for any two log digests $d$ and $d'$, an auditor can ask
the data center for the entire logs $L$ and $L'$ 
corresponding to both of these digests.
The auditor confirms that $d$ is the root of the log tree for $L$
and that $d'$ is the root of the log tree for $L'$.
Finally, the auditor checks that $L'$~extends~$L$.

As an extra precaution, users could specify external
parties (e.g., Let's Encrypt) as designated auditors during backup. 
During recovery, the HSMs would only complete the recovery if
these auditors sign the latest log digest.
In this way, mounting a brute-force PIN-guessing attempt against 
a user would require compromising the user's external auditors as well.

The transparency log can also help with PIN re-use.
As discussed in Section~\ref{app:deploy}, instead of storing the
salt directly with the service provider, the
salt itself can be encrypted using a second round of location-hiding
encryption and a null PIN\@.  After recovery, the salt will be
destroyed as discussed in the next section.  Once the salt has been
destroyed, the device restoring a backup can use the log to determine
if anyone else has ever fetched the salt.  If not, then it is safe for
the user to re-use the old PIN\@.

As described in \cref{sec:arch}, the log contains usernames, 
which could be sensitive. To prevent leaking usernames, we 
would replace usernames with random device identifiers that
are rerandomized when the device is factory reset. However, even
with this modification,
the log still leaks information about when and how often users
restore backups, which the service provider may not wish to make
public.
While we hope that organizations would make their logs public, we
acknowledge that some may only share their logs with several
hand-picked organizations for auditing or may not share their logs
at all. In these cases, our security guarantees still hold,
although some of the transparency benefits are lost.

 \section{{Forward~security~by~puncturable~encryption}} \label{sec:punc}

We would like our encrypted-backup system to provide
forward secrecy~\cite{CHK03}.
During the recovery process, the client reveals the identity of
the $n \ll N$ HSMs that can decrypt its backup.
Without forward secrecy, an attacker can break into these $n$
HSMs to recover the client's backed-up data.
Forward secrecy ensures that after recovery, an attacker, even one
who compromises all HSMs in the data center, learns no information
about the client's backup.

One seemingly straightforward way to provide forward secrecy would be to use a
new keypair for each backup.
However, because the client cannot interact with the HSMs it is encrypting to
during backup (as this would reveal their identities), using a unique keypair
for every backup would require every HSM in the data center to generate
a new keypair for every backup, running counter to our scalability goals.

\subsection{Background: Puncturable encryption}\label{sec:punc:bg}
We instead achieve forward secrecy using 
puncturable public-key encryption~\cite{GM15,GHJL17,DJSS18,CHNVW18,CRRV17,DKLR+18}.
A puncturable encryption scheme is 
a normal public-key encryption scheme
$(\keygen, \enc, \dec)$, with one extra routine:

\smallskip
\begin{hangparas}{2em}{1}
  $\puncture(\sk, \ct) \to \sk_{\ct}$. Given a decryption key $\sk$
  and a ciphertext $\ct$, output a new
  secret key $\sk_{\ct}$ that can decrypt all
  ciphertexts that $\sk$ could decrypt except for $\ct$.
\end{hangparas}
\smallskip

\noindent

\paragraph{Puncturable encryption for forward secrecy.}
To achieve forward security in \name, after an HSM decrypts its share of a client's
recovery ciphertext $\ct$, the HSM \emph{punctures} its secret decryption key.
The punctured key allows the HSM to decrypt all ciphertexts except for $\ct$.
Thus, if an attacker compromises \emph{all} HSMs in the data center after
a client has recovered its backup, the attacker will be unable to decrypt 
any backup images that clients have already recovered.
Furthermore, if an attacker compromises at most $\fevil \cdot N$ HSMs total, 
where $\fevil$ is a parameter of the system that we define in 
\cref{sec:design}, then
the attacker will not be able to recover any backed-up data whatsoever.

\paragraph{Existing tool: Bloom-filter encryption.}
Our implementation uses 
a puncturable encryption scheme called
Bloom-filter encryption~\cite{DJSS18}.
There are only two details of 
Bloom-filter encryption that are important for this discussion.
\begin{enumerate}
  \item \textit{The secret key is large.} 
          If a key supports $P \in \Zgz$ punctures and we want decryption
          to fail with probability at most $2^{-\lambda}$, then the
          secret key for Bloom-filter encryption 
          is an array of roughly $\lambda P$ elements of a cryptographic group~$\G$.
          After $P$ punctures, the secret key may no longer decrypt messages
          and it is necessary to rotate encryption keys.
  \item \textit{Puncturing is simple.} 
          Puncturing the secret key 
          just requires deleting $\lambda$ elements 
          in the data array that comprises the secret key.
\end{enumerate}

Concretely, when we set the Bloom-filter-encryption parameters
to suitable values for experimental evaluation, each
Bloom-filter encryption secret key has size over
64~MB.
Even high-end HSMs have only 1--2~MB of storage
(\cref{tab:hsm}), so storing such large keys on an
HSM would be impossible.

\subsection{Outsourced storage with secure deletion}
\label{sec:punc:del}

We describe how to efficiently outsource the
storage of this large secret key in a way that
preserves forward secrecy of the punctured key.
In particular, the HSM can outsource the storage
of its secret-key array to the untrustworthy service provider, 
while still retaining the ability to delete portions of the key.
To do this, we draw on techniques for outsourcing the storage of any
data array (not just secret keys) first described by Di Crescenzo et
al.\ \cite{how-to-forget-a-secret} and extended in subsequent work
\cite{delete-btree, vORAM}.

\paragraph{Desired functionality.}
At a high level, the HSM has access to 
(a) a small amount of internal storage and 
(b) a large external block store, run by the service provider.
The HSM wants to store an array of $D$ data blocks
at the provider $(\data_1, \dots, \data_D)$.
The HSM should be able to subsequently \emph{read} 
or \emph{delete} these blocks.

The following security properties should hold, even if the
attacker, controlling the service provider, 
may choose the data-array and sequence of 
operations the HSM performs:
\begin{itemize}
    \item \textbf{Integrity.} 
        If the service provider tampers with the stored data
        in a way that could cause a \emph{read} to return an incorrect result,
        the read operation outputs $\bot$.
        Otherwise, the read operation for a block $i$ returns 
        the value of the last data that the client wrote to block~$i$.

    \item \textbf{Secure deletion.}
        If the service provider compromises the HSM after the HSM has 
        run the \emph{delete} operation for the $i$th data block,
        the attacker learns nothing about the data stored in block $i$.
        (This property implies a confidentiality property: the 
        service provider learns nothing about the outsourced data.) 

\end{itemize}
For efficiency, the HSM storage requirements 
must be small (constant size) 
and the \emph{read} and \emph{delete} routines should run quickly (in time
logarithmic in the size $D$ of the data array).
Unlike in ORAM~\cite{G87,GO96}, our goal is not to hide
the HSM's data-access pattern from the service provider.
We aim only to hide the contents of the array.

\subsection{Building secure outsourced storage}
We explain the construction here in prose, drawing on techniques
first described by Di Crescenzo et al.\ \cite{how-to-forget-a-secret}.
See \iffull \cref{app:puncx} \fi for a more formal description.

\begin{figure}
  \centering
  \usetikzlibrary{patterns}

\definecolor{few-gray-bright}{HTML}{010202}
\definecolor{few-red-bright}{HTML}{EE2E2F}
\definecolor{few-green-bright}{HTML}{008C48}
\definecolor{few-blue-bright}{HTML}{185AA9}
\definecolor{few-orange-bright}{HTML}{F47D23}
\definecolor{few-purple-bright}{HTML}{662C91}
\definecolor{few-brown-bright}{HTML}{A21D21}
\definecolor{few-pink-bright}{HTML}{B43894}
\definecolor{few-gray}{HTML}{737373}
\definecolor{few-red}{HTML}{F15A60}
\definecolor{few-green}{HTML}{7AC36A}
\definecolor{few-blue}{HTML}{5A9BD4}
\definecolor{few-orange}{HTML}{FAA75B}
\definecolor{few-purple}{HTML}{9E67AB}
\definecolor{few-brown}{HTML}{CE7058}
\definecolor{few-pink}{HTML}{D77FB4}
\definecolor{few-gray-light}{HTML}{CCCCCC}
\definecolor{few-red-light}{HTML}{F2AFAD}
\definecolor{few-green-light}{HTML}{D9E4AA}
\definecolor{few-blue-light}{HTML}{B8D2EC}
\definecolor{few-orange-light}{HTML}{F3D1B0}
\definecolor{few-purple-light}{HTML}{D5B2D4}
\definecolor{few-brown-light}{HTML}{DDB9A9}
\definecolor{few-pink-light}{HTML}{EBC0DA}
 
\def\sk{\mathsf{sk}}
\def\data{\mathsf{data}}
\colorlet{cryptcolor}{few-purple-bright}
\colorlet{cryptcolorp}{few-green-bright}

\tikzset{
  crypt/.style={draw=cryptcolor,thick,rounded corners},
  cryptp/.style={draw=cryptcolorp,thick,rounded corners},
  link/.style={draw,thick,->},
  key/.style={at=(#1.south east),
    draw=cryptcolor,fill=cryptcolor,text=white,
    anchor=center, rounded corners=0.5ex, node font=\tiny,
    inner sep=0pt, minimum width=5ex,
    text height=2.25ex, text depth=1.5ex},
  keyp/.style={at=(#1.south east),
    draw=cryptcolorp,fill=cryptcolorp,text=white,
    anchor=center, rounded corners=0.5ex, node font=\tiny,
    inner sep=0pt, minimum width=5ex,
    text height=2.25ex, text depth=1.5ex},
}

\begin{tikzpicture}
\fill [rounded corners,thick,pattern=north east lines,pattern color=gray] (5.675,-2.29) rectangle (7,-1.73);
\fill [color=white] (6.5,-2.29) rectangle (7,-1.73);
\newif\ifleft
\foreach \side in {L,R} {
\if\side L
  \let\pr\relax
  \let\ph\relax
  \newcommand{\topstyle}{crypt}
  \newcommand{\keystyle}{key}
\else
  \let\pr'
  \def\ph##1{\phantom{##1}\llap{}}
  \tikzset{xshift=4.5cm}
  \newcommand{\topstyle}{cryptp}
  \newcommand{\keystyle}{keyp}
\fi

\node at (1,0) (k) {$\sk\pr$};
\node[\topstyle] (n) at (1,-1) {$\sk_0\mid \sk_1\pr$};
\node[\keystyle\relax=n] {$\sk\pr$};
\node[crypt] (n0) at (0,-2) {$\sk_{00}\mid \sk_{01}$};
\node[key=n0] {$\sk_0$};
\node[\topstyle] (n1) at (2,-2) {$\ph{\sk_{10}}\mid \sk_{11}$};
\node[\keystyle\relax=n1] {$\sk_1\pr$};

\foreach \d/\b in {1/00,2/01,3/10,4/11} {
  \node[anchor=east,rotate=90,crypt] at (\d-1.5,-3)
     (data\b) {$\data_{\d}$};
  \node[key=data\b,at=(data\b.south west)] {$\sk_{\tiny \b}$};
}
\draw[link] (k) -- (n);
\foreach \d in {0,1} {
  \draw[link] (n) -- (n\d);
  \draw[link] (n\d) -- (data\d1);
}
\draw[link] (n0) -- (data00);
\if\side L
  \draw[link] (n1) -- (data10);
\fi
}

\draw[dashed,few-gray] (-1,-.5) -- (7.5,-.5);
\node[anchor=south west,few-gray,node font=\footnotesize]
  at (-1,-.5) {\textbf{HSM}};
\node[anchor=north west,few-gray,node font=\footnotesize]
  at (-1,-.5) {\textbf{Servers}};

\node[node font=\small] at (1.05,-4.6) {\emph{Initial state.}};
\node[node font=\small] at (5.55,-4.6) {\emph{After deleting $\data_3$.}};

\end{tikzpicture}

   \caption{The outsourced-storage scheme has a tree of keys.
        An arrow $a \to b$ denotes that value $b$ 
        is stored encrypted under key $a$.
        A service provider that stores all values it sees
        and later compromises the HSM state ($\sk'$) still does not
        learn the deleted $\data_3$ value.}
  \label{figs:tree}
\end{figure}
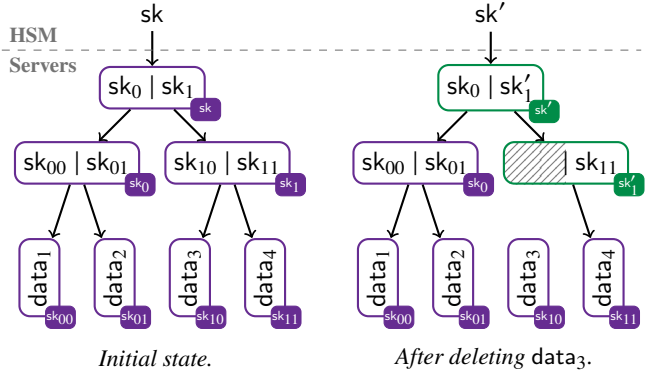

\itpara{Running the setup phase.}
During the setup phase, the outsourced-storage scheme 
builds a binary tree with $D$ leaves.
Every node of the tree contains a fresh symmetric encryption key.
During setup, for each node in the tree with key $\sk_i$, 
we encrypt the keys of the child nodes $\sk_{i0}$ and $\sk_{i1}$ with $\sk_i$ and 
store this ciphertext $\AEenc(\sk_i, \sk_{i0} \| \sk_{i1})$ in outsourced
storage.
At the leaves of the tree, we encrypt the $i$th data block
with the key $\sk_i$ at the $i$th leaf and we store
the ciphertext $\AEenc(\sk_i, \data_i)$ in outsourced storage.

For example, in Figure~\ref{figs:tree}, we use $\sk_0$ 
to encrypt $\sk_{00}$ and $\sk_{01}$ and we store the result in outsourced storage.
We use key $\sk_{01}$ 
to decrypt data item 2. Thus, knowing the root key $\sk$ is
enough to decrypt the entire tree and access every data element in the array.

\itpara{Reading a data block.}
To retrieve the data block at index $i$, the HSM reads in
the ciphertexts along the path from the tree root to leaf $i$. 
The HSM then decrypts the chain of
ciphertexts from the root down to recover the data block at index $i$.
For example, in Figure~\ref{figs:tree}, to retrieve data block 3, the HSM can use $\sk$ to decrypt $\sk_1$, and $\sk_1$ to decrypt $\sk_{10}$, 
which it can use to decrypt data item 1.

\itpara{Deleting a data block.}
To delete the data block at index $i$, the HSM recovers (as in retrieval)
the keys along the path from the root to leaf $i$. 
At the node containing the key to decrypt data block $i$, the
HSM deletes the key. 
It then chooses a fresh key and re-encrypts the other
key at that node using the fresh key. To maintain the ability of the parent
key to decrypt the child ciphertext, the HSM updates the parent of that node to
contain the fresh key for its child and re-encrypts the parent's keys under
a new key. It continues this up the path to the root, where the HSM chooses
a new key $\sk'$ to encrypt the root. The HSM replaces $\sk$ with $\sk'$, deleting
the old $\sk$, and then sends the new ciphertexts along the path from the
root to leaf $i$ back to the service provider.
For example, in Figure~\ref{figs:tree},
to delete data item 3, the HSM decrypts the keys $(\sk_0 \| \sk_1)$ and $(\sk_{10} \| \sk_{11})$.
The HSM then deletes $\sk_{10}$, chooses a new key $\sk'_1$ to encrypt $\sk_{11}$, and then
chooses a new key $\sk'$ to encrypt $\sk_0$ and $\sk'_1$. The HSM then replaces $\sk$
with $\sk'$.

\itpara{Efficiency.}
The setup time is linear in the size of the data array $D$. 
The runtimes of retrieval and deletion are both
logarithmic in $D$, and
require only symmetric-key operations. 
The HSM stores only the constant-sized root encryption key $\sk$. 

\itpara{Security intuition.}
An HSM can always recover the keys necessary to decrypt a data item, provided
the HSM did not previously delete any of the keys necessary for decryption.
Integrity follows immediately from the security of the underlying
authenticated encryption scheme. Finally, we ensure secure deletion by deleting
the key necessary to decrypt a certain data item and updating the root key. Without
the old root key, it is impossible to access the key necessary to decrypt the
deleted data item. 

\paragraph{Putting it together.}
To summarize: the HSMs use a puncturable
encryption scheme to prevent the compromise of HSM secrets at time $T$
from allowing an adversary to learn about backed-up data that was
recovered any time before $T$.
We implement puncturable encryption using Bloom-filter encryption and 
outsource the storage of the large secret decryption key
while allowing secure deletion.

 \section{Extensions and deployment considerations}
\label{app:deploy}

The full \name implementation has to deal with a number
of additional issues, which we discuss now. 

\paragraph{Failure during recovery.}
\xxx{Need to clarify this}
As discussed in Section~\ref{sec:punc}, after participating in
recovery, HSMs revoke their ability to decrypt the recovered ciphertext.
One consequence is that a client cannot recover the same backup ciphertext twice. 
This raises the question of what happens if a replacement device fails during or
shortly after recovery, or if a communication failure during recovery prevents the new
device from receiving the replies from the HSMs.

To solve this problem, 
when a client initiates recovery, it first generates a fresh per-recovery keypair
$(\sk, \pk)$ for a public-key encryption scheme.
The client backs up this secret key $\sk$ using \name
before initiating its recovery.
Next, the client sends the public key $\pk$ to each HSM\@
and then begins the backup-recovery process.
Each HSM encrypts its replies to the client under $\pk$,
and each HSM sends a copy of each reply to the data center.
If a client device fails during recovery,
a second, replacement client device can retrieve
the backed-up secret key $\sk$ and use these to
decrypt the replies stored at the data center.
This scheme nests arbitrarily, thereby handling any number of
consecutive device failures during recovery.

\paragraph{Incremental backups.}
In practice, mobile devices often generate incremental backups rather
than encrypting the entire disk image for each backup. \name supports
incremental backups in the following way. The user uses \name to
store a single AES key, which the user also keeps on her phone.
The user can then encrypt incremental backups under this AES key and
upload the resulting ciphertext to the data center.
When the user recovers, she recovers her AES key and can use this
key to decrypt the incremental updates.

\paragraph{Multiple recovery ciphertexts.}
Clients back up their phones regularly (e.g., every three days), 
and will thus generate a series of recovery ciphertexts.
We want to ensure that after a client recovers her backup 
from time $t$, the HSMs involved in recovery puncture their 
secret decryption keys so that they cannot decrypt that client's
backups from earlier times $t' < t$, even if an attacker compromises
all HSMs in the data center.
To achieve this, in the puncturable-encryption
step (\cref{sec:punc}), we have the client use the same salt 
for each recovery ciphertext it generates. 
In this way, the client will encrypt its series of backups 
to the same set of HSMs. When these HSM puncture their secret keys 
during the recovery process, they will destroy their ability to decrypt
any previous recovery ciphertexts from the given client.
After recovery, the client chooses a new salt to generate subsequent backups
on its new device.

\paragraph{Preventing post-recovery PIN leakage.}\label{par:leakage}
As we have discussed, an attacker that watches the client recover can
learn a salted hash of the user's PIN, which can be used to mount
an offline brute-force attack to learn the user's PIN.

One approach to protect against this attack would be to have each user store
their salt in secret-shared form at a random set $\Ssalt$ of HSMs,
where $\Ssalt$ is included in the client's recovery ciphertext.
Then, provided that the attacker does not compromise this set of HSMs,
the attacker would learn no useful 
information on the user's PIN, even after recovery.
An attacker could always compromise every HSM in $\Ssalt$, 
but an attacker that can compromise only a $\fevil$ fraction of
HSMs in the data center would not be able to
mount this attack against too many clients' salts.
We hope to model and prove this multi-user
PIN-protection property in future work.

\begin{table}[t]
\centering
{\small
\begin{minipage}[t]{0.4\columnwidth}
\centering
\begin{tabular}[t]{lr}
\toprule
    \textbf{Operation} & \textbf{Ops/sec} \\
\midrule
  Pairing & 0.43\\ 
    ECDSA ver& 5.85\\
    ElGamal dec& 6.67\\
    $g^x \in \mathbb{G}_\text{P256}$ & 7.69\\
\bottomrule 
\end{tabular}
\end{minipage}~
\begin{minipage}[t]{0.6\columnwidth}
\centering
\begin{tabular}[t]{llr}
\toprule
    &\textbf{Operation} & \textbf{Ops/sec} \\
\midrule
    &HMAC-SHA256 & 2,173.91 \\
    &AES-128 & 3,703.70 \\ \midrule
    \multirow{3}{0pt}{\rott{\centering \textbf{I/O}}} &RTT, HID (32b) & 71.43 \\
    &RTT, CDC (32b) & 2,277.90 \\
    &Flash read (32b) & $\approx$166,000\\
\bottomrule 
\end{tabular}
\end{minipage}
}
\caption{Microbenchmarks on SoloKey. 
  Pairing is on BLS12-381 curve using the JEDI library~\cite{jedi}. 
  Other public-key operations
  use NIST P256 curve.}
\label{tab:micro}
\end{table}
\section{Implementation and evaluation}
\label{sec:eval}

We implemented \name on an experimental data
cluster of 100 hardware security devices (\cref{fig:datacenter}).

\sitpara{HSM.} 
For the HSMs, we used SoloKeys~\cite{solokeys}, a low-cost
open-source USB FIDO2 security key.
SoloKeys use a STM32L432 microcontroller
with an ARM Cortex-M4 32-bit RISC core clocked at 80MHz and 265KB of
memory. 
The device is not side-channel resistant, but has a true
random number generator and can lock its firmware.
We add roughly 2,500 lines of C code 
to the open-source SoloKey firmware~\cite{solo-code}.

By default, SoloKeys communicate with the USB host via USB HID, an 
interrupt-based USB class used typically for keyboards and mice
that has a maximum throughput of 64KBps.
To improve performance, we
rewrote parts of the firmware to use USB CDC, a high-throughput
USB class commonly used for networking devices. 
This gave a roughly $32 \times$ 
increase in I/O throughput (Table~\ref{tab:micro}).

For the puncturable-encryption scheme (\cref{sec:punc:bg}),
we use a variant of Bloom-filter encryption~\cite{DJSS18}
that avoids the need for pairings~\cite{BF-IBE} 
but increases the size of the HSMs' public keys.
For the aggregate signature scheme needed for the log, we
use BLS-style multisignatures~\cite{bonehbls}
over the JEDI~\cite{jedi} implementation of the BLS12-381 curve.

Our implementation does not encrypt communication between
the client and HSMs. Based on the time to run AES-128 and 
ElGamal encryption on the SoloKeys, we estimate that
transport-layer encryption would add two ElGamal decryptions
and 2KB of AES operations per recovery,
increasing recovery time by approximately 0.3 seconds, 
or 30\%.
This overhead is comparatively high because
processing a recovery only requires a handful of symmetric
and public key operations.

\sitpara{Service Provider.} Our service provider host
is a Linux machine with an Intel Xeon E5-2650 CPU clocked at 2.60GHz. 
Our service-provider implementation is roughly 3,800 lines of C/C++ code (excluding tests)
and uses OpenSSL. 

\sitpara{Client.} Our client device is a Google Pixel~4.
Our implementation is roughly 2,300 lines of C/C++ code (excluding tests)
and uses OpenSSL.

\begin{figure}[t]
\begin{minipage}{0.42\columnwidth}
\centering
\begingroup
\makeatletter
\begin{pgfpicture}
\pgfpathrectangle{\pgfpointorigin}{\pgfqpoint{1.411952in}{1.500000in}}
\pgfusepath{use as bounding box, clip}
\begin{pgfscope}
\pgfsetbuttcap
\pgfsetmiterjoin
\definecolor{currentfill}{rgb}{1.000000,1.000000,1.000000}
\pgfsetfillcolor{currentfill}
\pgfsetlinewidth{0.000000pt}
\definecolor{currentstroke}{rgb}{1.000000,1.000000,1.000000}
\pgfsetstrokecolor{currentstroke}
\pgfsetdash{}{0pt}
\pgfpathmoveto{\pgfqpoint{0.000000in}{0.000000in}}
\pgfpathlineto{\pgfqpoint{1.411952in}{0.000000in}}
\pgfpathlineto{\pgfqpoint{1.411952in}{1.500000in}}
\pgfpathlineto{\pgfqpoint{0.000000in}{1.500000in}}
\pgfpathclose
\pgfusepath{fill}
\end{pgfscope}
\begin{pgfscope}
\pgfsetbuttcap
\pgfsetmiterjoin
\definecolor{currentfill}{rgb}{1.000000,1.000000,1.000000}
\pgfsetfillcolor{currentfill}
\pgfsetlinewidth{0.000000pt}
\definecolor{currentstroke}{rgb}{0.000000,0.000000,0.000000}
\pgfsetstrokecolor{currentstroke}
\pgfsetstrokeopacity{0.000000}
\pgfsetdash{}{0pt}
\pgfpathmoveto{\pgfqpoint{0.314665in}{0.351551in}}
\pgfpathlineto{\pgfqpoint{1.362789in}{0.351551in}}
\pgfpathlineto{\pgfqpoint{1.362789in}{1.500000in}}
\pgfpathlineto{\pgfqpoint{0.314665in}{1.500000in}}
\pgfpathclose
\pgfusepath{fill}
\end{pgfscope}
\begin{pgfscope}
\definecolor{textcolor}{rgb}{0.000000,0.000000,0.000000}
\pgfsetstrokecolor{textcolor}
\pgfsetfillcolor{textcolor}
\pgftext[x=0.314665in,y=0.254329in,,top]{\color{textcolor}\rmfamily\fontsize{8.000000}{9.600000}\selectfont 0}
\end{pgfscope}
\begin{pgfscope}
\definecolor{textcolor}{rgb}{0.000000,0.000000,0.000000}
\pgfsetstrokecolor{textcolor}
\pgfsetfillcolor{textcolor}
\pgftext[x=0.565070in,y=0.254329in,,top]{\color{textcolor}\rmfamily\fontsize{8.000000}{9.600000}\selectfont 2.5K}
\end{pgfscope}
\begin{pgfscope}
\definecolor{textcolor}{rgb}{0.000000,0.000000,0.000000}
\pgfsetstrokecolor{textcolor}
\pgfsetfillcolor{textcolor}
\pgftext[x=0.815475in,y=0.254329in,,top]{\color{textcolor}\rmfamily\fontsize{8.000000}{9.600000}\selectfont 5K}
\end{pgfscope}
\begin{pgfscope}
\definecolor{textcolor}{rgb}{0.000000,0.000000,0.000000}
\pgfsetstrokecolor{textcolor}
\pgfsetfillcolor{textcolor}
\pgftext[x=1.065880in,y=0.254329in,,top]{\color{textcolor}\rmfamily\fontsize{8.000000}{9.600000}\selectfont 7.5K}
\end{pgfscope}
\begin{pgfscope}
\definecolor{textcolor}{rgb}{0.000000,0.000000,0.000000}
\pgfsetstrokecolor{textcolor}
\pgfsetfillcolor{textcolor}
\pgftext[x=1.316285in,y=0.254329in,,top]{\color{textcolor}\rmfamily\fontsize{8.000000}{9.600000}\selectfont 10K}
\end{pgfscope}
\begin{pgfscope}
\definecolor{textcolor}{rgb}{0.000000,0.000000,0.000000}
\pgfsetstrokecolor{textcolor}
\pgfsetfillcolor{textcolor}
\pgftext[x=0.838727in,y=0.099387in,,top]{\color{textcolor}\rmfamily\fontsize{8.000000}{9.600000}\selectfont Data center size (\(\displaystyle N\))}
\end{pgfscope}
\begin{pgfscope}
\pgfpathrectangle{\pgfqpoint{0.314665in}{0.351551in}}{\pgfqpoint{1.048125in}{1.148449in}}
\pgfusepath{clip}
\pgfsetrectcap
\pgfsetroundjoin
\pgfsetlinewidth{0.803000pt}
\definecolor{currentstroke}{rgb}{0.900000,0.900000,0.900000}
\pgfsetstrokecolor{currentstroke}
\pgfsetdash{}{0pt}
\pgfpathmoveto{\pgfqpoint{0.314665in}{0.521045in}}
\pgfpathlineto{\pgfqpoint{1.362789in}{0.521045in}}
\pgfusepath{stroke}
\end{pgfscope}
\begin{pgfscope}
\definecolor{textcolor}{rgb}{0.000000,0.000000,0.000000}
\pgfsetstrokecolor{textcolor}
\pgfsetfillcolor{textcolor}
\pgftext[x=0.154942in, y=0.483379in, left, base]{\color{textcolor}\rmfamily\fontsize{8.000000}{9.600000}\selectfont \(\displaystyle {20}\)}
\end{pgfscope}
\begin{pgfscope}
\pgfpathrectangle{\pgfqpoint{0.314665in}{0.351551in}}{\pgfqpoint{1.048125in}{1.148449in}}
\pgfusepath{clip}
\pgfsetrectcap
\pgfsetroundjoin
\pgfsetlinewidth{0.803000pt}
\definecolor{currentstroke}{rgb}{0.900000,0.900000,0.900000}
\pgfsetstrokecolor{currentstroke}
\pgfsetdash{}{0pt}
\pgfpathmoveto{\pgfqpoint{0.314665in}{0.804808in}}
\pgfpathlineto{\pgfqpoint{1.362789in}{0.804808in}}
\pgfusepath{stroke}
\end{pgfscope}
\begin{pgfscope}
\definecolor{textcolor}{rgb}{0.000000,0.000000,0.000000}
\pgfsetstrokecolor{textcolor}
\pgfsetfillcolor{textcolor}
\pgftext[x=0.154942in, y=0.767142in, left, base]{\color{textcolor}\rmfamily\fontsize{8.000000}{9.600000}\selectfont \(\displaystyle {30}\)}
\end{pgfscope}
\begin{pgfscope}
\pgfpathrectangle{\pgfqpoint{0.314665in}{0.351551in}}{\pgfqpoint{1.048125in}{1.148449in}}
\pgfusepath{clip}
\pgfsetrectcap
\pgfsetroundjoin
\pgfsetlinewidth{0.803000pt}
\definecolor{currentstroke}{rgb}{0.900000,0.900000,0.900000}
\pgfsetstrokecolor{currentstroke}
\pgfsetdash{}{0pt}
\pgfpathmoveto{\pgfqpoint{0.314665in}{1.088571in}}
\pgfpathlineto{\pgfqpoint{1.362789in}{1.088571in}}
\pgfusepath{stroke}
\end{pgfscope}
\begin{pgfscope}
\definecolor{textcolor}{rgb}{0.000000,0.000000,0.000000}
\pgfsetstrokecolor{textcolor}
\pgfsetfillcolor{textcolor}
\pgftext[x=0.154942in, y=1.050905in, left, base]{\color{textcolor}\rmfamily\fontsize{8.000000}{9.600000}\selectfont \(\displaystyle {40}\)}
\end{pgfscope}
\begin{pgfscope}
\pgfpathrectangle{\pgfqpoint{0.314665in}{0.351551in}}{\pgfqpoint{1.048125in}{1.148449in}}
\pgfusepath{clip}
\pgfsetrectcap
\pgfsetroundjoin
\pgfsetlinewidth{0.803000pt}
\definecolor{currentstroke}{rgb}{0.900000,0.900000,0.900000}
\pgfsetstrokecolor{currentstroke}
\pgfsetdash{}{0pt}
\pgfpathmoveto{\pgfqpoint{0.314665in}{1.372334in}}
\pgfpathlineto{\pgfqpoint{1.362789in}{1.372334in}}
\pgfusepath{stroke}
\end{pgfscope}
\begin{pgfscope}
\definecolor{textcolor}{rgb}{0.000000,0.000000,0.000000}
\pgfsetstrokecolor{textcolor}
\pgfsetfillcolor{textcolor}
\pgftext[x=0.154942in, y=1.334668in, left, base]{\color{textcolor}\rmfamily\fontsize{8.000000}{9.600000}\selectfont \(\displaystyle {50}\)}
\end{pgfscope}
\begin{pgfscope}
\definecolor{textcolor}{rgb}{0.000000,0.000000,0.000000}
\pgfsetstrokecolor{textcolor}
\pgfsetfillcolor{textcolor}
\pgftext[x=0.099387in,y=0.925776in,,bottom,rotate=90.000000]{\color{textcolor}\rmfamily\fontsize{8.000000}{9.600000}\selectfont Time to audit log (s)}
\end{pgfscope}
\begin{pgfscope}
\pgfpathrectangle{\pgfqpoint{0.314665in}{0.351551in}}{\pgfqpoint{1.048125in}{1.148449in}}
\pgfusepath{clip}
\pgfsetrectcap
\pgfsetroundjoin
\pgfsetlinewidth{1.505625pt}
\definecolor{currentstroke}{rgb}{0.458824,0.439216,0.701961}
\pgfsetstrokecolor{currentstroke}
\pgfsetdash{}{0pt}
\pgfpathmoveto{\pgfqpoint{1.316285in}{0.403754in}}
\pgfpathlineto{\pgfqpoint{0.815475in}{0.470023in}}
\pgfpathlineto{\pgfqpoint{0.648538in}{0.593161in}}
\pgfpathlineto{\pgfqpoint{0.565070in}{0.659422in}}
\pgfpathlineto{\pgfqpoint{0.514989in}{0.724942in}}
\pgfpathlineto{\pgfqpoint{0.481601in}{0.791725in}}
\pgfpathlineto{\pgfqpoint{0.457753in}{0.913671in}}
\pgfpathlineto{\pgfqpoint{0.439867in}{0.980036in}}
\pgfpathlineto{\pgfqpoint{0.425956in}{1.046298in}}
\pgfpathlineto{\pgfqpoint{0.414827in}{1.126324in}}
\pgfpathlineto{\pgfqpoint{0.405721in}{1.248724in}}
\pgfpathlineto{\pgfqpoint{0.398133in}{1.314438in}}
\pgfpathlineto{\pgfqpoint{0.391712in}{1.379222in}}
\pgfpathlineto{\pgfqpoint{0.386209in}{1.447798in}}
\pgfusepath{stroke}
\end{pgfscope}
\begin{pgfscope}
\pgfpathrectangle{\pgfqpoint{0.314665in}{0.351551in}}{\pgfqpoint{1.048125in}{1.148449in}}
\pgfusepath{clip}
\pgfsetbuttcap
\pgfsetroundjoin
\definecolor{currentfill}{rgb}{0.458824,0.439216,0.701961}
\pgfsetfillcolor{currentfill}
\pgfsetlinewidth{0.000000pt}
\definecolor{currentstroke}{rgb}{0.458824,0.439216,0.701961}
\pgfsetstrokecolor{currentstroke}
\pgfsetdash{}{0pt}
\pgfsys@defobject{currentmarker}{\pgfqpoint{-0.020833in}{-0.020833in}}{\pgfqpoint{0.020833in}{0.020833in}}{
\pgfpathmoveto{\pgfqpoint{0.000000in}{-0.020833in}}
\pgfpathcurveto{\pgfqpoint{0.005525in}{-0.020833in}}{\pgfqpoint{0.010825in}{-0.018638in}}{\pgfqpoint{0.014731in}{-0.014731in}}
\pgfpathcurveto{\pgfqpoint{0.018638in}{-0.010825in}}{\pgfqpoint{0.020833in}{-0.005525in}}{\pgfqpoint{0.020833in}{0.000000in}}
\pgfpathcurveto{\pgfqpoint{0.020833in}{0.005525in}}{\pgfqpoint{0.018638in}{0.010825in}}{\pgfqpoint{0.014731in}{0.014731in}}
\pgfpathcurveto{\pgfqpoint{0.010825in}{0.018638in}}{\pgfqpoint{0.005525in}{0.020833in}}{\pgfqpoint{0.000000in}{0.020833in}}
\pgfpathcurveto{\pgfqpoint{-0.005525in}{0.020833in}}{\pgfqpoint{-0.010825in}{0.018638in}}{\pgfqpoint{-0.014731in}{0.014731in}}
\pgfpathcurveto{\pgfqpoint{-0.018638in}{0.010825in}}{\pgfqpoint{-0.020833in}{0.005525in}}{\pgfqpoint{-0.020833in}{0.000000in}}
\pgfpathcurveto{\pgfqpoint{-0.020833in}{-0.005525in}}{\pgfqpoint{-0.018638in}{-0.010825in}}{\pgfqpoint{-0.014731in}{-0.014731in}}
\pgfpathcurveto{\pgfqpoint{-0.010825in}{-0.018638in}}{\pgfqpoint{-0.005525in}{-0.020833in}}{\pgfqpoint{0.000000in}{-0.020833in}}
\pgfpathclose
\pgfusepath{fill}
}
\begin{pgfscope}
\pgfsys@transformshift{1.316285in}{0.403754in}
\pgfsys@useobject{currentmarker}{}
\end{pgfscope}
\begin{pgfscope}
\pgfsys@transformshift{0.815475in}{0.470023in}
\pgfsys@useobject{currentmarker}{}
\end{pgfscope}
\begin{pgfscope}
\pgfsys@transformshift{0.648538in}{0.593161in}
\pgfsys@useobject{currentmarker}{}
\end{pgfscope}
\begin{pgfscope}
\pgfsys@transformshift{0.565070in}{0.659422in}
\pgfsys@useobject{currentmarker}{}
\end{pgfscope}
\begin{pgfscope}
\pgfsys@transformshift{0.514989in}{0.724942in}
\pgfsys@useobject{currentmarker}{}
\end{pgfscope}
\begin{pgfscope}
\pgfsys@transformshift{0.481601in}{0.791725in}
\pgfsys@useobject{currentmarker}{}
\end{pgfscope}
\begin{pgfscope}
\pgfsys@transformshift{0.457753in}{0.913671in}
\pgfsys@useobject{currentmarker}{}
\end{pgfscope}
\begin{pgfscope}
\pgfsys@transformshift{0.439867in}{0.980036in}
\pgfsys@useobject{currentmarker}{}
\end{pgfscope}
\begin{pgfscope}
\pgfsys@transformshift{0.425956in}{1.046298in}
\pgfsys@useobject{currentmarker}{}
\end{pgfscope}
\begin{pgfscope}
\pgfsys@transformshift{0.414827in}{1.126324in}
\pgfsys@useobject{currentmarker}{}
\end{pgfscope}
\begin{pgfscope}
\pgfsys@transformshift{0.405721in}{1.248724in}
\pgfsys@useobject{currentmarker}{}
\end{pgfscope}
\begin{pgfscope}
\pgfsys@transformshift{0.398133in}{1.314438in}
\pgfsys@useobject{currentmarker}{}
\end{pgfscope}
\begin{pgfscope}
\pgfsys@transformshift{0.391712in}{1.379222in}
\pgfsys@useobject{currentmarker}{}
\end{pgfscope}
\begin{pgfscope}
\pgfsys@transformshift{0.386209in}{1.447798in}
\pgfsys@useobject{currentmarker}{}
\end{pgfscope}
\end{pgfscope}
\begin{pgfscope}
\pgfsetrectcap
\pgfsetmiterjoin
\pgfsetlinewidth{0.501875pt}
\definecolor{currentstroke}{rgb}{0.000000,0.000000,0.000000}
\pgfsetstrokecolor{currentstroke}
\pgfsetdash{}{0pt}
\pgfpathmoveto{\pgfqpoint{0.314665in}{0.351551in}}
\pgfpathlineto{\pgfqpoint{0.314665in}{1.500000in}}
\pgfusepath{stroke}
\end{pgfscope}
\begin{pgfscope}
\pgfsetrectcap
\pgfsetmiterjoin
\pgfsetlinewidth{0.501875pt}
\definecolor{currentstroke}{rgb}{0.000000,0.000000,0.000000}
\pgfsetstrokecolor{currentstroke}
\pgfsetdash{}{0pt}
\pgfpathmoveto{\pgfqpoint{0.314665in}{0.351551in}}
\pgfpathlineto{\pgfqpoint{1.362789in}{0.351551in}}
\pgfusepath{stroke}
\end{pgfscope}
\end{pgfpicture}
\makeatother
\endgroup
     \caption{Log-audit time after inserting 10K recovery attempts
    for a log with roughly 100M recovery attempts. We only measure
    the auditing time for 100 HSMs as we only had 100 SoloKeys; 
    we distribute the work as if there were $N$ HSMs.}
\label{fig:logeval}
\end{minipage}~~~~~~~
\begin{minipage}{0.56\columnwidth}
\centering
\begingroup
\makeatletter
\begin{pgfpicture}
\pgfpathrectangle{\pgfpointorigin}{\pgfqpoint{1.758660in}{1.662663in}}
\pgfusepath{use as bounding box, clip}
\begin{pgfscope}
\pgfsetbuttcap
\pgfsetmiterjoin
\definecolor{currentfill}{rgb}{1.000000,1.000000,1.000000}
\pgfsetfillcolor{currentfill}
\pgfsetlinewidth{0.000000pt}
\definecolor{currentstroke}{rgb}{1.000000,1.000000,1.000000}
\pgfsetstrokecolor{currentstroke}
\pgfsetdash{}{0pt}
\pgfpathmoveto{\pgfqpoint{0.000000in}{0.000000in}}
\pgfpathlineto{\pgfqpoint{1.758660in}{0.000000in}}
\pgfpathlineto{\pgfqpoint{1.758660in}{1.662663in}}
\pgfpathlineto{\pgfqpoint{0.000000in}{1.662663in}}
\pgfpathclose
\pgfusepath{fill}
\end{pgfscope}
\begin{pgfscope}
\pgfsetbuttcap
\pgfsetmiterjoin
\definecolor{currentfill}{rgb}{1.000000,1.000000,1.000000}
\pgfsetfillcolor{currentfill}
\pgfsetlinewidth{0.000000pt}
\definecolor{currentstroke}{rgb}{0.000000,0.000000,0.000000}
\pgfsetstrokecolor{currentstroke}
\pgfsetstrokeopacity{0.000000}
\pgfsetdash{}{0pt}
\pgfpathmoveto{\pgfqpoint{0.400997in}{0.521560in}}
\pgfpathlineto{\pgfqpoint{1.755556in}{0.521560in}}
\pgfpathlineto{\pgfqpoint{1.755556in}{1.562334in}}
\pgfpathlineto{\pgfqpoint{0.400997in}{1.562334in}}
\pgfpathclose
\pgfusepath{fill}
\end{pgfscope}
\begin{pgfscope}
\pgfpathrectangle{\pgfqpoint{0.400997in}{0.521560in}}{\pgfqpoint{1.354559in}{1.040774in}}
\pgfusepath{clip}
\pgfsetbuttcap
\pgfsetroundjoin
\definecolor{currentfill}{rgb}{0.905882,0.160784,0.541176}
\pgfsetfillcolor{currentfill}
\pgfsetlinewidth{0.000000pt}
\definecolor{currentstroke}{rgb}{0.000000,0.000000,0.000000}
\pgfsetstrokecolor{currentstroke}
\pgfsetdash{}{0pt}
\pgfpathmoveto{\pgfqpoint{0.462568in}{0.668736in}}
\pgfpathlineto{\pgfqpoint{0.462568in}{0.521560in}}
\pgfpathlineto{\pgfqpoint{0.616495in}{0.521560in}}
\pgfpathlineto{\pgfqpoint{0.770422in}{0.521560in}}
\pgfpathlineto{\pgfqpoint{0.924349in}{0.521560in}}
\pgfpathlineto{\pgfqpoint{1.078276in}{0.521560in}}
\pgfpathlineto{\pgfqpoint{1.232203in}{0.521560in}}
\pgfpathlineto{\pgfqpoint{1.386130in}{0.521560in}}
\pgfpathlineto{\pgfqpoint{1.540058in}{0.521560in}}
\pgfpathlineto{\pgfqpoint{1.693985in}{0.521560in}}
\pgfpathlineto{\pgfqpoint{1.693985in}{0.669421in}}
\pgfpathlineto{\pgfqpoint{1.693985in}{0.669421in}}
\pgfpathlineto{\pgfqpoint{1.540058in}{0.667202in}}
\pgfpathlineto{\pgfqpoint{1.386130in}{0.670199in}}
\pgfpathlineto{\pgfqpoint{1.232203in}{0.667150in}}
\pgfpathlineto{\pgfqpoint{1.078276in}{0.668056in}}
\pgfpathlineto{\pgfqpoint{0.924349in}{0.669060in}}
\pgfpathlineto{\pgfqpoint{0.770422in}{0.669818in}}
\pgfpathlineto{\pgfqpoint{0.616495in}{0.668701in}}
\pgfpathlineto{\pgfqpoint{0.462568in}{0.668736in}}
\pgfpathclose
\pgfusepath{fill}
\end{pgfscope}
\begin{pgfscope}
\pgfpathrectangle{\pgfqpoint{0.400997in}{0.521560in}}{\pgfqpoint{1.354559in}{1.040774in}}
\pgfusepath{clip}
\pgfsetbuttcap
\pgfsetroundjoin
\definecolor{currentfill}{rgb}{0.105882,0.619608,0.466667}
\pgfsetfillcolor{currentfill}
\pgfsetlinewidth{0.000000pt}
\definecolor{currentstroke}{rgb}{0.000000,0.000000,0.000000}
\pgfsetstrokecolor{currentstroke}
\pgfsetdash{}{0pt}
\pgfpathmoveto{\pgfqpoint{0.462568in}{0.779416in}}
\pgfpathlineto{\pgfqpoint{0.462568in}{0.668736in}}
\pgfpathlineto{\pgfqpoint{0.616495in}{0.668701in}}
\pgfpathlineto{\pgfqpoint{0.770422in}{0.669818in}}
\pgfpathlineto{\pgfqpoint{0.924349in}{0.669060in}}
\pgfpathlineto{\pgfqpoint{1.078276in}{0.668056in}}
\pgfpathlineto{\pgfqpoint{1.232203in}{0.667150in}}
\pgfpathlineto{\pgfqpoint{1.386130in}{0.670199in}}
\pgfpathlineto{\pgfqpoint{1.540058in}{0.667202in}}
\pgfpathlineto{\pgfqpoint{1.693985in}{0.669421in}}
\pgfpathlineto{\pgfqpoint{1.693985in}{1.126502in}}
\pgfpathlineto{\pgfqpoint{1.693985in}{1.126502in}}
\pgfpathlineto{\pgfqpoint{1.540058in}{1.081867in}}
\pgfpathlineto{\pgfqpoint{1.386130in}{1.041307in}}
\pgfpathlineto{\pgfqpoint{1.232203in}{0.996869in}}
\pgfpathlineto{\pgfqpoint{1.078276in}{0.952890in}}
\pgfpathlineto{\pgfqpoint{0.924349in}{0.911731in}}
\pgfpathlineto{\pgfqpoint{0.770422in}{0.868957in}}
\pgfpathlineto{\pgfqpoint{0.616495in}{0.823760in}}
\pgfpathlineto{\pgfqpoint{0.462568in}{0.779416in}}
\pgfpathclose
\pgfusepath{fill}
\end{pgfscope}
\begin{pgfscope}
\pgfpathrectangle{\pgfqpoint{0.400997in}{0.521560in}}{\pgfqpoint{1.354559in}{1.040774in}}
\pgfusepath{clip}
\pgfsetbuttcap
\pgfsetroundjoin
\definecolor{currentfill}{rgb}{0.458824,0.439216,0.701961}
\pgfsetfillcolor{currentfill}
\pgfsetlinewidth{0.000000pt}
\definecolor{currentstroke}{rgb}{0.000000,0.000000,0.000000}
\pgfsetstrokecolor{currentstroke}
\pgfsetdash{}{0pt}
\pgfpathmoveto{\pgfqpoint{0.462568in}{0.793961in}}
\pgfpathlineto{\pgfqpoint{0.462568in}{0.779416in}}
\pgfpathlineto{\pgfqpoint{0.616495in}{0.823760in}}
\pgfpathlineto{\pgfqpoint{0.770422in}{0.868957in}}
\pgfpathlineto{\pgfqpoint{0.924349in}{0.911731in}}
\pgfpathlineto{\pgfqpoint{1.078276in}{0.952890in}}
\pgfpathlineto{\pgfqpoint{1.232203in}{0.996869in}}
\pgfpathlineto{\pgfqpoint{1.386130in}{1.041307in}}
\pgfpathlineto{\pgfqpoint{1.540058in}{1.081867in}}
\pgfpathlineto{\pgfqpoint{1.693985in}{1.126502in}}
\pgfpathlineto{\pgfqpoint{1.693985in}{1.176435in}}
\pgfpathlineto{\pgfqpoint{1.693985in}{1.176435in}}
\pgfpathlineto{\pgfqpoint{1.540058in}{1.126761in}}
\pgfpathlineto{\pgfqpoint{1.386130in}{1.082602in}}
\pgfpathlineto{\pgfqpoint{1.232203in}{1.034023in}}
\pgfpathlineto{\pgfqpoint{1.078276in}{0.985221in}}
\pgfpathlineto{\pgfqpoint{0.924349in}{0.937948in}}
\pgfpathlineto{\pgfqpoint{0.770422in}{0.891487in}}
\pgfpathlineto{\pgfqpoint{0.616495in}{0.842360in}}
\pgfpathlineto{\pgfqpoint{0.462568in}{0.793961in}}
\pgfpathclose
\pgfusepath{fill}
\end{pgfscope}
\begin{pgfscope}
\pgfsetrectcap
\pgfsetmiterjoin
\pgfsetlinewidth{0.501875pt}
\definecolor{currentstroke}{rgb}{0.000000,0.000000,0.000000}
\pgfsetstrokecolor{currentstroke}
\pgfsetdash{}{0pt}
\pgfpathmoveto{\pgfqpoint{0.400997in}{0.521560in}}
\pgfpathlineto{\pgfqpoint{0.400997in}{1.562334in}}
\pgfusepath{stroke}
\end{pgfscope}
\begin{pgfscope}
\pgfsetrectcap
\pgfsetmiterjoin
\pgfsetlinewidth{0.501875pt}
\definecolor{currentstroke}{rgb}{0.000000,0.000000,0.000000}
\pgfsetstrokecolor{currentstroke}
\pgfsetdash{}{0pt}
\pgfpathmoveto{\pgfqpoint{0.400997in}{0.521560in}}
\pgfpathlineto{\pgfqpoint{1.755556in}{0.521560in}}
\pgfusepath{stroke}
\end{pgfscope}
\begin{pgfscope}
\definecolor{textcolor}{rgb}{0.000000,0.000000,0.000000}
\pgfsetstrokecolor{textcolor}
\pgfsetfillcolor{textcolor}
\pgftext[x=1.611754in,y=0.424338in,,top]{\color{textcolor}\rmfamily\fontsize{8.000000}{9.600000}\selectfont 100K}
\end{pgfscope}
\begin{pgfscope}
\definecolor{textcolor}{rgb}{0.000000,0.000000,0.000000}
\pgfsetstrokecolor{textcolor}
\pgfsetfillcolor{textcolor}
\pgftext[x=1.356087in,y=0.424338in,,top]{\color{textcolor}\rmfamily\fontsize{8.000000}{9.600000}\selectfont 10K}
\end{pgfscope}
\begin{pgfscope}
\definecolor{textcolor}{rgb}{0.000000,0.000000,0.000000}
\pgfsetstrokecolor{textcolor}
\pgfsetfillcolor{textcolor}
\pgftext[x=1.100419in,y=0.424338in,,top]{\color{textcolor}\rmfamily\fontsize{8.000000}{9.600000}\selectfont 1K}
\end{pgfscope}
\begin{pgfscope}
\definecolor{textcolor}{rgb}{0.000000,0.000000,0.000000}
\pgfsetstrokecolor{textcolor}
\pgfsetfillcolor{textcolor}
\pgftext[x=0.844752in,y=0.424338in,,top]{\color{textcolor}\rmfamily\fontsize{8.000000}{9.600000}\selectfont 100}
\end{pgfscope}
\begin{pgfscope}
\definecolor{textcolor}{rgb}{0.000000,0.000000,0.000000}
\pgfsetstrokecolor{textcolor}
\pgfsetfillcolor{textcolor}
\pgftext[x=0.589085in,y=0.424338in,,top]{\color{textcolor}\rmfamily\fontsize{8.000000}{9.600000}\selectfont 10}
\end{pgfscope}
\begin{pgfscope}
\definecolor{textcolor}{rgb}{0.000000,0.000000,0.000000}
\pgfsetstrokecolor{textcolor}
\pgfsetfillcolor{textcolor}
\pgftext[x=0.397892in, y=0.138508in, left, base]{\color{textcolor}\rmfamily\fontsize{8.000000}{9.600000}\selectfont Recoveries before key rotation}
\end{pgfscope}
\begin{pgfscope}
\definecolor{textcolor}{rgb}{0.000000,0.000000,0.000000}
\pgfsetstrokecolor{textcolor}
\pgfsetfillcolor{textcolor}
\pgftext[x=0.774037in, y=0.024055in, left, base]{\color{textcolor}\rmfamily\fontsize{8.000000}{9.600000}\selectfont  {\fontsize{7}{4} \selectfont Secret key size}}
\end{pgfscope}
\begin{pgfscope}
\pgfpathrectangle{\pgfqpoint{0.400997in}{0.521560in}}{\pgfqpoint{1.354559in}{1.040774in}}
\pgfusepath{clip}
\pgfsetbuttcap
\pgfsetroundjoin
\pgfsetlinewidth{0.803000pt}
\definecolor{currentstroke}{rgb}{0.900000,0.900000,0.900000}
\pgfsetstrokecolor{currentstroke}
\pgfsetdash{{0.800000pt}{1.320000pt}}{0.000000pt}
\pgfpathmoveto{\pgfqpoint{0.400997in}{0.521560in}}
\pgfpathlineto{\pgfqpoint{1.755556in}{0.521560in}}
\pgfusepath{stroke}
\end{pgfscope}
\begin{pgfscope}
\definecolor{textcolor}{rgb}{0.000000,0.000000,0.000000}
\pgfsetstrokecolor{textcolor}
\pgfsetfillcolor{textcolor}
\pgftext[x=0.154942in, y=0.483894in, left, base]{\color{textcolor}\rmfamily\fontsize{8.000000}{9.600000}\selectfont \(\displaystyle 0.00\)}
\end{pgfscope}
\begin{pgfscope}
\pgfpathrectangle{\pgfqpoint{0.400997in}{0.521560in}}{\pgfqpoint{1.354559in}{1.040774in}}
\pgfusepath{clip}
\pgfsetbuttcap
\pgfsetroundjoin
\pgfsetlinewidth{0.803000pt}
\definecolor{currentstroke}{rgb}{0.900000,0.900000,0.900000}
\pgfsetstrokecolor{currentstroke}
\pgfsetdash{{0.800000pt}{1.320000pt}}{0.000000pt}
\pgfpathmoveto{\pgfqpoint{0.400997in}{0.781754in}}
\pgfpathlineto{\pgfqpoint{1.755556in}{0.781754in}}
\pgfusepath{stroke}
\end{pgfscope}
\begin{pgfscope}
\definecolor{textcolor}{rgb}{0.000000,0.000000,0.000000}
\pgfsetstrokecolor{textcolor}
\pgfsetfillcolor{textcolor}
\pgftext[x=0.154942in, y=0.744088in, left, base]{\color{textcolor}\rmfamily\fontsize{8.000000}{9.600000}\selectfont \(\displaystyle 0.25\)}
\end{pgfscope}
\begin{pgfscope}
\pgfpathrectangle{\pgfqpoint{0.400997in}{0.521560in}}{\pgfqpoint{1.354559in}{1.040774in}}
\pgfusepath{clip}
\pgfsetbuttcap
\pgfsetroundjoin
\pgfsetlinewidth{0.803000pt}
\definecolor{currentstroke}{rgb}{0.900000,0.900000,0.900000}
\pgfsetstrokecolor{currentstroke}
\pgfsetdash{{0.800000pt}{1.320000pt}}{0.000000pt}
\pgfpathmoveto{\pgfqpoint{0.400997in}{1.041947in}}
\pgfpathlineto{\pgfqpoint{1.755556in}{1.041947in}}
\pgfusepath{stroke}
\end{pgfscope}
\begin{pgfscope}
\definecolor{textcolor}{rgb}{0.000000,0.000000,0.000000}
\pgfsetstrokecolor{textcolor}
\pgfsetfillcolor{textcolor}
\pgftext[x=0.154942in, y=1.004281in, left, base]{\color{textcolor}\rmfamily\fontsize{8.000000}{9.600000}\selectfont \(\displaystyle 0.50\)}
\end{pgfscope}
\begin{pgfscope}
\pgfpathrectangle{\pgfqpoint{0.400997in}{0.521560in}}{\pgfqpoint{1.354559in}{1.040774in}}
\pgfusepath{clip}
\pgfsetbuttcap
\pgfsetroundjoin
\pgfsetlinewidth{0.803000pt}
\definecolor{currentstroke}{rgb}{0.900000,0.900000,0.900000}
\pgfsetstrokecolor{currentstroke}
\pgfsetdash{{0.800000pt}{1.320000pt}}{0.000000pt}
\pgfpathmoveto{\pgfqpoint{0.400997in}{1.302140in}}
\pgfpathlineto{\pgfqpoint{1.755556in}{1.302140in}}
\pgfusepath{stroke}
\end{pgfscope}
\begin{pgfscope}
\definecolor{textcolor}{rgb}{0.000000,0.000000,0.000000}
\pgfsetstrokecolor{textcolor}
\pgfsetfillcolor{textcolor}
\pgftext[x=0.154942in, y=1.264474in, left, base]{\color{textcolor}\rmfamily\fontsize{8.000000}{9.600000}\selectfont \(\displaystyle 0.75\)}
\end{pgfscope}
\begin{pgfscope}
\pgfpathrectangle{\pgfqpoint{0.400997in}{0.521560in}}{\pgfqpoint{1.354559in}{1.040774in}}
\pgfusepath{clip}
\pgfsetbuttcap
\pgfsetroundjoin
\pgfsetlinewidth{0.803000pt}
\definecolor{currentstroke}{rgb}{0.900000,0.900000,0.900000}
\pgfsetstrokecolor{currentstroke}
\pgfsetdash{{0.800000pt}{1.320000pt}}{0.000000pt}
\pgfpathmoveto{\pgfqpoint{0.400997in}{1.562334in}}
\pgfpathlineto{\pgfqpoint{1.755556in}{1.562334in}}
\pgfusepath{stroke}
\end{pgfscope}
\begin{pgfscope}
\definecolor{textcolor}{rgb}{0.000000,0.000000,0.000000}
\pgfsetstrokecolor{textcolor}
\pgfsetfillcolor{textcolor}
\pgftext[x=0.154942in, y=1.524668in, left, base]{\color{textcolor}\rmfamily\fontsize{8.000000}{9.600000}\selectfont \(\displaystyle 1.00\)}
\end{pgfscope}
\begin{pgfscope}
\definecolor{textcolor}{rgb}{0.000000,0.000000,0.000000}
\pgfsetstrokecolor{textcolor}
\pgfsetfillcolor{textcolor}
\pgftext[x=0.099387in,y=1.041947in,,bottom,rotate=90.000000]{\color{textcolor}\rmfamily\fontsize{8.000000}{9.600000}\selectfont Decrypt + Puncture time (s)}
\end{pgfscope}
\begin{pgfscope}
\definecolor{textcolor}{rgb}{0.000000,0.000000,0.000000}
\pgfsetstrokecolor{textcolor}
\pgfsetfillcolor{textcolor}
\pgftext[x=1.611754in,y=0.261367in,,base]{\color{textcolor}\rmfamily\fontsize{6.000000}{7.200000}\selectfont 30MB}
\end{pgfscope}
\begin{pgfscope}
\definecolor{textcolor}{rgb}{0.000000,0.000000,0.000000}
\pgfsetstrokecolor{textcolor}
\pgfsetfillcolor{textcolor}
\pgftext[x=1.356087in,y=0.261367in,,base]{\color{textcolor}\rmfamily\fontsize{6.000000}{7.200000}\selectfont 3MB}
\end{pgfscope}
\begin{pgfscope}
\definecolor{textcolor}{rgb}{0.000000,0.000000,0.000000}
\pgfsetstrokecolor{textcolor}
\pgfsetfillcolor{textcolor}
\pgftext[x=1.100419in,y=0.261367in,,base]{\color{textcolor}\rmfamily\fontsize{6.000000}{7.200000}\selectfont 300KB}
\end{pgfscope}
\begin{pgfscope}
\definecolor{textcolor}{rgb}{0.000000,0.000000,0.000000}
\pgfsetstrokecolor{textcolor}
\pgfsetfillcolor{textcolor}
\pgftext[x=0.844752in,y=0.261367in,,base]{\color{textcolor}\rmfamily\fontsize{6.000000}{7.200000}\selectfont 30KB}
\end{pgfscope}
\begin{pgfscope}
\definecolor{textcolor}{rgb}{0.000000,0.000000,0.000000}
\pgfsetstrokecolor{textcolor}
\pgfsetfillcolor{textcolor}
\pgftext[x=0.589085in,y=0.261367in,,base]{\color{textcolor}\rmfamily\fontsize{6.000000}{7.200000}\selectfont 3KB}
\end{pgfscope}
\begin{pgfscope}
\pgfsetbuttcap
\pgfsetmiterjoin
\definecolor{currentfill}{rgb}{1.000000,1.000000,1.000000}
\pgfsetfillcolor{currentfill}
\pgfsetfillopacity{0.800000}
\pgfsetlinewidth{1.003750pt}
\definecolor{currentstroke}{rgb}{0.800000,0.800000,0.800000}
\pgfsetstrokecolor{currentstroke}
\pgfsetstrokeopacity{0.800000}
\pgfsetdash{}{0pt}
\pgfpathmoveto{\pgfqpoint{0.555897in}{1.302140in}}
\pgfpathlineto{\pgfqpoint{1.615320in}{1.302140in}}
\pgfpathquadraticcurveto{\pgfqpoint{1.634765in}{1.302140in}}{\pgfqpoint{1.634765in}{1.321585in}}
\pgfpathlineto{\pgfqpoint{1.634765in}{1.627784in}}
\pgfpathquadraticcurveto{\pgfqpoint{1.634765in}{1.647228in}}{\pgfqpoint{1.615320in}{1.647228in}}
\pgfpathlineto{\pgfqpoint{0.555897in}{1.647228in}}
\pgfpathquadraticcurveto{\pgfqpoint{0.536453in}{1.647228in}}{\pgfqpoint{0.536453in}{1.627784in}}
\pgfpathlineto{\pgfqpoint{0.536453in}{1.321585in}}
\pgfpathquadraticcurveto{\pgfqpoint{0.536453in}{1.302140in}}{\pgfqpoint{0.555897in}{1.302140in}}
\pgfpathclose
\pgfusepath{stroke,fill}
\end{pgfscope}
\begin{pgfscope}
\pgfsetbuttcap
\pgfsetmiterjoin
\definecolor{currentfill}{rgb}{0.458824,0.439216,0.701961}
\pgfsetfillcolor{currentfill}
\pgfsetlinewidth{0.000000pt}
\definecolor{currentstroke}{rgb}{0.000000,0.000000,0.000000}
\pgfsetstrokecolor{currentstroke}
\pgfsetstrokeopacity{0.000000}
\pgfsetdash{}{0pt}
\pgfpathmoveto{\pgfqpoint{0.575341in}{1.540284in}}
\pgfpathlineto{\pgfqpoint{0.769786in}{1.540284in}}
\pgfpathlineto{\pgfqpoint{0.769786in}{1.608339in}}
\pgfpathlineto{\pgfqpoint{0.575341in}{1.608339in}}
\pgfpathclose
\pgfusepath{fill}
\end{pgfscope}
\begin{pgfscope}
\definecolor{textcolor}{rgb}{0.000000,0.000000,0.000000}
\pgfsetstrokecolor{textcolor}
\pgfsetfillcolor{textcolor}
\pgftext[x=0.847564in,y=1.540284in,left,base]{\color{textcolor}\rmfamily\fontsize{7.000000}{8.400000}\selectfont I/O}
\end{pgfscope}
\begin{pgfscope}
\pgfsetbuttcap
\pgfsetmiterjoin
\definecolor{currentfill}{rgb}{0.105882,0.619608,0.466667}
\pgfsetfillcolor{currentfill}
\pgfsetlinewidth{0.000000pt}
\definecolor{currentstroke}{rgb}{0.000000,0.000000,0.000000}
\pgfsetstrokecolor{currentstroke}
\pgfsetstrokeopacity{0.000000}
\pgfsetdash{}{0pt}
\pgfpathmoveto{\pgfqpoint{0.575341in}{1.451180in}}
\pgfpathlineto{\pgfqpoint{0.769786in}{1.451180in}}
\pgfpathlineto{\pgfqpoint{0.769786in}{1.519236in}}
\pgfpathlineto{\pgfqpoint{0.575341in}{1.519236in}}
\pgfpathclose
\pgfusepath{fill}
\end{pgfscope}
\begin{pgfscope}
\definecolor{textcolor}{rgb}{0.000000,0.000000,0.000000}
\pgfsetstrokecolor{textcolor}
\pgfsetfillcolor{textcolor}
\pgftext[x=0.847564in,y=1.451180in,left,base]{\color{textcolor}\rmfamily\fontsize{7.000000}{8.400000}\selectfont Symmetric key ops}
\end{pgfscope}
\begin{pgfscope}
\pgfsetbuttcap
\pgfsetmiterjoin
\definecolor{currentfill}{rgb}{0.905882,0.160784,0.541176}
\pgfsetfillcolor{currentfill}
\pgfsetlinewidth{0.000000pt}
\definecolor{currentstroke}{rgb}{0.000000,0.000000,0.000000}
\pgfsetstrokecolor{currentstroke}
\pgfsetstrokeopacity{0.000000}
\pgfsetdash{}{0pt}
\pgfpathmoveto{\pgfqpoint{0.575341in}{1.362077in}}
\pgfpathlineto{\pgfqpoint{0.769786in}{1.362077in}}
\pgfpathlineto{\pgfqpoint{0.769786in}{1.430133in}}
\pgfpathlineto{\pgfqpoint{0.575341in}{1.430133in}}
\pgfpathclose
\pgfusepath{fill}
\end{pgfscope}
\begin{pgfscope}
\definecolor{textcolor}{rgb}{0.000000,0.000000,0.000000}
\pgfsetstrokecolor{textcolor}
\pgfsetfillcolor{textcolor}
\pgftext[x=0.847564in,y=1.362077in,left,base]{\color{textcolor}\rmfamily\fontsize{7.000000}{8.400000}\selectfont Public key ops}
\end{pgfscope}
\end{pgfpicture}
\makeatother
\endgroup
   \caption{Time to run puncturable encryption on
  a single HSM as the maximum number of allowed punctures (and also secret key size)
  grows. The cost of our outsourced storage scheme dominates, though
  the access time is logarithmic in the size of the key.}
\label{fig:scale-punc}
\end{minipage}
\end{figure}

\subsection{Microbenchmarks}

\paragraph{Log.}
\cref{fig:logeval} demonstrates how increasing the number of HSMs
reduces the log-digest update time. 
We assume that the log is periodically garbage
collected (i.e., approximately once a month), so
that it holds at most a hundred million
recovery attempts at once.
If the HSMs run the log-update process every 10 minutes, each HSM spends approximately 11\% of
its active cycles auditing the log. 
The choice of how often to update the log is a tradeoff between how long users must wait to
recover their backups and the total number of write cycles to non-volatile storage permitted
by the hardware.

\paragraph{Puncturable encryption.}
\cref{fig:scale-punc} shows the cost of performing
a decrypt-and-puncture operation as the number of 
supported punctures increases.
The AES operations associated with our scheme for outsourced
storage with secure deletion (\cref{sec:punc:del}) dominate the cost.

Another way to implement outsourced storage with secure
deletion would be to have the HSM store the outsourced array 
encrypted under a single AES key $k$. To delete an item, the HSM would
read in the entire array, delete the item, and write out the 
entire array encrypted under a fresh key $k'$.
With this approach, a deletion takes 48 minutes
for a 64~MB array (the size of our outsourced secret keys).
Our scheme thus 
improves system throughput by roughly $4,423\times$.

Each HSM punctures its secret key
(\cref{sec:punc:bg}) once after each decryption it performs.
Since our puncturable-encryption scheme only
supports a fixed number of punctures, each HSM
must periodically rotate its encryption keys.
We configure our puncturable-encryption scheme to
allow each HSM to perform roughly $2^{18}$
decryptions before it must rotate its keys (rotation is
triggered when half of the elements of the secret key
have been deleted).
Key rotation is expensive: we estimate 
(based on the number of public-key operations required) 
that key rotation on our HSMs will take roughly 75 hours.
Each HSM spends approximately 139.4 hours 
processing recoveries and maintaining the log between key rotations.
Therefore, each HSM spends roughly 56\% of its
cycles rotating its keys, and each HSM can
process 1,503.9 recoveries per hour on average.

\begin{figure}[t]
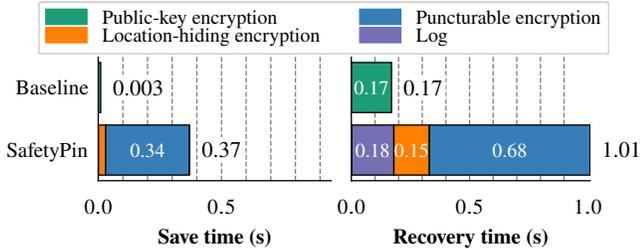

\centering
\begingroup
\makeatletter

\makeatother
\endgroup
 \caption{Breakdown of time to save (on Android Pixel~4 phone) and recover
    (using our SoloKey cluster). We do not consider the time to encrypt or
    decrypt disk images.}
\label{figs:breakdown}
\end{figure}

\subsection{End-to-end costs}
\label{sec:eval:e2e-costs}

\paragraph{Parameters.}
We estimate that on average, each user will run recovery once a year.  
(There are 3.8B smartphone users~\cite{phone-users} and 1.5B smartphones sold
annually~\cite{phone-sales}, so we expect $1.5/3.8 = 0.39 \ll 1$ recovery/user/year.)
We calculate that a \name deployment of $N=3,100$ HSMs could support
one billion users. 
So, we treat our small cluster of 100 HSMs as a representative
slice of a larger data center of $N=3,100$ HSMs.
Within this larger data center, 
each client shares its recovery keys among a cluster of $n=40$ HSMs.
This choice of $n$ is based on the size of the data center $N$ and
PINs with six decimal digits, and is dictated by \iffull \cref{thm:lhe-sec}.
\fi
We set the puncturable encryption keys to allow $2^{20}$ punctures, as we
found this provides a reasonable tradeoff between the time to decrypt and puncture
and the time between key rotations.
With these parameters,
we maintain secrecy if at most an $\fevil = \smash{\tfrac{1}{16}}$
fraction of the HSMs are compromised (or $\fevil \cdot N \approx 194$ total). 
We allow data recovery if at most an $\ffail = \smash{\tfrac{1}{64}}$ 
fraction fail due to benign hardware failures (or $\ffail \cdot N \approx 48$ total).

\paragraph{Baseline.}
We compare against an encrypted-backup system modeled on the ones
that Google and Apple use~\cite{google-key-vault, apple-key-vault}.
To backup, the client selects a fixed cluster of five HSMs and encrypts her
recovery key and a hash of her PIN under the cluster's public key.
At recovery, the client sends the recovery ciphertext and a hash
of her PIN to the cluster, and any HSM in the cluster can decrypt
the ciphertext, check that the PIN hashes match, and return the
recovery key. 
To defeat brute-force PIN-guessing attacks, each
HSM independently limits the number of recovery
attempts allowed on a given ciphertext.

\paragraph{Client overhead.}
Figure~\ref{figs:breakdown} gives the overhead of
generating a backup in \name, compared to the baseline.
The backup process takes 0.37 seconds.
\name recovery ciphertexts are 16.5KB, versus 
130B for our baseline, though
we expect encrypted disk image 
to dominate the ciphertext size.

\name increases the bandwidth cost at the client.
In the baseline scheme, the client downloads five public keys---one
from each of its five chosen HSMs.
In \name, the client must fetch a copy of all HSMs' public keys.
(This way, the service provider does not learn the subset of HSMs
to which the client is encrypting its backup.)
So, when a client first joins the system, the client
must download all these keys (11.5MB).
Whenever an HSM rotates its puncturable-encryption keys, 
clients must download the HSM's new public key.
In a deployment of $N=3,100$ HSMs supporting one billion recoveries 
annually, we estimate that each \name client must download
1.97MB of keying material daily.
Increasing the puncturable
encryption failure probability would decrease client bandwidth, although this
would require decreasing the fraction of HSMs allowed to fail,~$\ffail$.
If a client goes offline for several days, it must
download the rotated public keys for each day it
spent offline (roughly 2MB/day), up to a maximum
of 11.5MB (the size of all HSMs' keys).
However, the client only needs to store the public keys for 
the $n$ HSMs comprising its chosen recovery cluster which amounts to 9.02KB.

\paragraph{Recovery time.}
At a cluster size of $n=40$ HSMs, \cref{fig:scale-breakdown}
shows that the end-to-end recovery time takes 1.01 seconds. 
Puncturable-encryption operations dominate recovery time (\cref{figs:breakdown}),
since these require expensive elliptic-curve operations for ElGamal decryption
and many I/O and AES operations in order to perform secure deletion (\cref{sec:punc:del}).

\paragraph{Tail latency.}
In a deployment of \name, it will be important to consider not only the
average throughput of the \name cluster, but also the request latency.
Since recovery requests will arrive concurrently and in a bursty fashion,
we will need to overprovision the system slightly to ensure that request tail
latency does not grow too high, even under large transient loads.
In \cref{figs:concurrent}, we model how many HSMs are required to achieve
various 99th-percentile latencies, while handling different average throughputs.
We compute these values by modeling incoming requests using a Poisson
process and each HSM using a M/M/1 queue with service times derived
from our experimental results. 
As the figure demonstrates, 
by increasing the total number of HSMs, we can reduce the tail latency even
when accounting for request contention.
We anticipate that recovery time will in practice be dominated by the time
to download the encrypted disk image, and so as long as the tail latency
is less than or close to this time, any delay is unlikely to be noticed by
the user.

\paragraph{Financial cost.}
\cref{figs:throughput} shows how throughput scales
as the outlay on HSMs increases and
\cref{tab:e2e-cost} presents dollar-cost estimates
for \name deployments with different types of HSMs.
For a configuration that tolerates the
compromise of 50 high-quality HSMs, we estimate
that adding \name to an unencrypted backup system
would increase the system's dollar cost by $2.5\%$.

\begin{figure}[t]
\begin{minipage}{0.48\columnwidth}
\centering
\begingroup
\makeatletter

\makeatother
\endgroup
     \caption{Recovery time grows slowly as cluster size $n$ increases.
    The bits of security lost refers to the difference between the
    advantage of an attacker against a \name deployment with the given
    value of $n$ and an attacker trying to guess the user's PIN.\xxx{Explain more about security loss} \xxx{Needs to be updated after artifact evaluation}}
\label{fig:scale-breakdown}
\end{minipage}~~~~~~~
\begin{minipage}{0.48\columnwidth}
\centering
\begingroup
\makeatletter
%
\makeatother
\endgroup
 \caption{Estimated number of \name-protected recoveries per year supported
  by clusters of different HSM models and costs.
We use $g^x$/sec to compute the expected throughput of more powerful HSMs
based on our measurements using SoloKeys (Table~\ref{tab:hsm}).}
\label{figs:throughput}
\end{minipage}
\end{figure}

\begin{figure}[t]
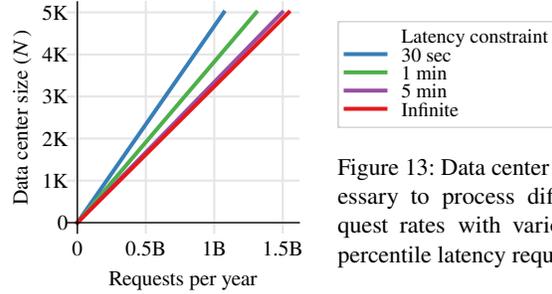

    \begin{minipage}{0.5\columnwidth}
    \begingroup
\makeatletter
%
\makeatother
\endgroup
         \caption{Data center sizes necessary to process different request rates with various 99th-percentile latency requirements.}
    \label{figs:concurrent}
    \end{minipage}
\end{figure}

\begin{table}
\centering
{\small
\begin{tabular}[t]{l|rcr|r}
\toprule
  &\multicolumn{1}{c}{\textbf{HSM Qty.}} &\multicolumn{1}{c}{$\fevil$} &\multicolumn{1}{c|}{$N_\mathsf{evil}$}
  &\multicolumn{1}{c}{\textbf{Cost}} \\
  \midrule
  SoloKey~\cite{solokeys} &  3,037 & $1/16$ & 189 & \$60.7K \\
  YubiHSM2~\cite{yubico-hsm} &  1,732 & $1/16$ & 108 & \$1.1M \\
  SafeNet A700~\cite{safenet}&  40 & $1/20$ & 2& \$738.7K \\
  -- 10 evil HSMs & 320 & $1/32$ & 10 &\$3.0M \\
  -- 50 evil HSMs & 800 & $1/16$ & 50 & \$14.8M \\ \midrule
\multicolumn{4}{l}{\emph{Estimated cost of storing 4GB $\times$ $10^9$ users per year:}} & \$600M \\
\bottomrule
\end{tabular}
}
  \caption[]{The estimated hardware cost of a \name deployment supporting
one billion users, if each user recovers once per year.
  The $N_\textsf{evil}$ number is how many corrupt
  HSMs the deployment tolerates.\par
\smallskip
{\footnotesize {\setstretch{0.8}
We estimate the storage cost using
AWS S3 infrequent access~\cite{s3}
(\$0.0125 per GB/month). 
We estimate YubiHSM2 and SafeNet HSM throughput using
their data sheets (Table~\ref{tab:hsm}). 
When computing the number of HSMs necessary
to service a billion users, we account for 
key-rotation time. 
A cluster of 40 SafeNet HSMs can meet the
throughput demands of one billion users, so
we also consider larger deployments tolerating
more compromised HSMs. 
\par }}}
\label{tab:e2e-cost}
\end{table}
  \section{Related work} 

Today's encrypted-backup systems rely either on the security of
hardware security modules~\cite{Green16,apple-key-vault}, 
secure microcontrollers~\cite{titan}, or secure enclaves~\cite{signal-recovery,sgx}.
Vulnerabilities in these hardware components 
leave encrypted-backup systems open to attack.
And there is ample evidence of vulnerabilities in both HSMs~\cite{xbox-tpm, 
tpm-fail, coppersmith, HSPK18, Kauer07, BKKH13, blackhat, CVE-2015-5464}
and enclaves~\cite{foreshadow, 
BMDK+17, GESM17, LSGK+17, branchscope, VPS17, VWKP+17, plundervolt, sgxpectre, GVPS18,
LJJK+17, BCDF+18}, and reason for concern about
hardware backdoors as well~\cite{TK10,A2,BRPB13,KJBP14}.

Many companies 
including Anchorage~\cite{anchorage}, Unbound Tech~\cite{unbound-tech},
Curv~\cite{curv}, and Ledger Vault~\cite{ledger-vault},
offer systems for secret-sharing
cryptocurrency secret keys across multiple hardware devices.
Unlike \name, these solutions use a small fixed set of HSMs, 
so they cannot simultaneously provide scalability 
and protection against adaptive HSM compromise.

Mavroudis et al.~\cite{mavroudis2017touch} propose building a single trustworthy
hardware security module from a large array of potentially faulty hardware devices.
To achieve this, they use cryptographic protocols
for threshold key-generation, decryption, and signing.
Like Myst, \name distributes trust over a large number of hardware devices.
Unlike Myst, \name focuses on hardware-protected
PIN-based encrypted backups, rather than 
more traditional HSM operations, such as
decryption and signing.

In recent theoretical work,
Benhamouda et al.\ show how to scalably store secrets on proof-of-stake
blockchains when an adversary can adaptively 
corrupt some fraction of the stake~\cite{BGGH+20}.
They face many of the same cryptographic challenges that we tackle
in \cref{sec:hide}; their theoretical treatment complements
our implementation-focused approach.
While they use proactive secret sharing to periodically re-share the
secret and hide the secret from an adversary controlling some fraction of the
stake, our approach allows a party with
some low-entropy secret to recover the high-entropy secret.

Di Crescenzo et al.\ show how to use a small amount of erasable
memory to outsource the storage of a much larger data array in
non-erasable memory while
providing secure deletion~\cite{how-to-forget-a-secret}.
Their construction uses a tree-based approach where reading, writing,
or deleting an element requires a number of symmetric
key operations logarithmic in the size of the data array.
Subsequent work has applied similiar techniques
to B-trees~\cite{delete-btree} and dynamically sized arrays~\cite{vORAM}.
These works are part of a larger body of work on secure deletion using
cryptography~\cite{sok-secure-deletion, BL96, vanish, fade, PBHS+05,
ficklebase, Perlman05, fs-eRAM}.

Transparency logs inspire our log design~\cite{wave,ct,CONIKS, trillian, LMW19}.
While these logs allow a powerful auditor to verify correctness,
they do not easily allow distributing the work of auditing
across many less powerful participants.
The proofs we provide to the HSMs about the state of the log draw on
work on authenticated data structures~\cite{Tamassia03, MNDG+04,
PT07} and cryptocurrency light clients~\cite{nakamoto08}.
Kaptchuk et al.\ show how public ledgers can be used to build stateful
systems from stateless secure hardware~\cite{KGM19}, and they show how
their techniques can be applied to Apple's encrypted-backup system.
This work is complementary to ours, as they show how to securely manage
state in cases where HSMs do not have secure internal non-volatile storage
(an assumption we make in \name).
 \section{Conclusion}
\name is an encrypted backup system that 
(a) requires its users to only remember a short PIN,
(b) defeats brute-force PIN-guessing attacks using hardware protections, and 
(c) provides strong protection against hardware compromise. 
\name demonstrates that it is possible to reap the benefits
of hardware security protections without turning these 
hardware devices into single points of security failure.

{
\paragraph{Acknowledgments.}
We would like to thank Raluca Ada Popa and Bryan Ford
for their support throughout this project.
Albert Kwon,
Anish Athalye,
Christian Mouchet,
David Lazar, 
Dima Kogan, and
Lefteris Kokoris-Kogias,
offered thoughtful criticism on drafts of this work.
We thank our shepherd Sebastian Angel for his 
work reviewing our camera-ready.
We thank Dan Boneh for useful suggestions on how to 
simplify the security analysis of our location-hiding
encryption scheme, we thank Vinod Vaikuntanathan for the 
suggestion to avoid pairings in our puncturable-encryption scheme, 
and we thank Keith Winstein 
for early brainstorming on passwords and PINs.
Conor Patrick and Nicolas Stalder gave helpful suggestions for
modifying the SoloKey firmware to support USB CDC, and
Vivian Fang provided advice on assembling the system.
Finally, we thank the anonymous reviewers of USENIX Security 2020
and OSDI 2020
for their thorough and detailed feedback.
This research is also supported in part by the RISELab, 
a Facebook Research Award, and a NSF GRFP fellowship. 
}
 
{\small
\setstretch{0.95}
\bibliography{refs}
\bibliographystyle{plain} 
}

\appendix
\iffull

\section{Analysis: Location-hiding encryption}
\label{sec:hide-defs}

\itpara{Additional notation.}
We use $x \gets 7$ to indicate assignment.
For a finite set $S$, we use $x \getsr S$ to denote
taking a uniform random sample from $S$.
The notation $\poly(\cdot)$ refers to a fixed polynomial function
and $\negl(\cdot)$ refers to a fixed negligible function.

\subsection{Syntax}

We first define the syntax of a location-hiding encryption scheme.
The scheme is parameterized by a total number of HSMs $N$,
a cluster size $n \ll N$, a PIN space $\calP$, a message space $\calM$, and 
two constants:
\begin{itemize}
  \item [(a)] the fraction $\ffail$ of the $N$ HSMs in a cluster whose benign failure
we can tolerate while still allowing message recovery, and
\item[(b)] the fraction $\fevil$ of the $N$ total HSMs whose compromise we can
tolerate while still providing security. 
\end{itemize}
In our construction, we take $\ffail = \tfrac{1}{64}$ and $\fevil = \tfrac{1}{16}$, but
any constants $0 < \ffail, \fevil < 1$ would suffice
with an adjustment to the parameters.

Formally, a location-hiding encryption scheme consists of the
following five algorithms:

\medskip

\begin{hangparas}{\defhangindent}{1}
$\keygen(1^\lambda) \to (\pk, \sk)$.
      On input security parameter $\lambda$, expressed in unary,
      output a public-key encryption keypair.

$\enc(\pk_1, \dots, \pk_N, \salt, \pin, \msg) \to \ct$.
      Given a list of $N$ public keys, generated using $\keygen$, 
      a randomizing salt $\salt \in \zo^*$,
      a PIN $\pin \in \calP$, and a message $\msg \in \calM$, output 
      a ciphertext $\ct$.

$\select(\salt, \pin) \to ( i_1, \dots, i_n )$. 
      Given a salt and PIN, output
      the indices $(i_1, \dots, i_n) \in [N]^n$ 
      of the $n$ secret keys 
      needed to reconstruct the encrypted message. 

$\dec(\sk_{i_j}, i_j, \ct) \to \sigma_{j}$.
      Given a ciphertext $\ct$, a secret key $\sk_{i_j}$
      and an index $i_j$, produce the $j$-th secret share
      $\sigma_{i_j}$ of the plaintext message. 

$\reconst(\sigma_1, \dots, \sigma_n) \to \msg$.
      Given $n$ shares of the plaintext message $\msg$
      returned by $\dec$, reconstruct $\msg$.

\end{hangparas}

\xxx{Explain how we handle chosen-ciphertext security}

To understand the syntax, it may be helpful for us
to explain how we use these routines in our encrypted-backup
application:
\begin{enumerate}
  \item \textbf{Setup.} During device provisioning, each HSM runs
        $(\pk_i, \sk_i) \gets \keygen(1^\lambda)$ to 
        generate its keypair $(\pk_i, \sk_i)$.
        The HSM stores the secret $\sk_i$ in its
        internal memory and it publishes $\pk_i$.
  \item \textbf{Backup.}
        To back up its message $\msg$ with salt $\salt$ and PIN $\pin$, 
        given HSM public keys $(\pk_1, \dots, \pk_N)$ the client runs
        \[ \ct \gets \enc(\pk_1, \dots, \pk_N, \salt, \pin, \msg).\]
        The client uploads its recovery ciphertext $\ct$ to the
        service provider.
  \item \textbf{Recovery.}
        During recovery, the client fetches its salt $\salt$ and recovery
        ciphertext $\ct$ from the service provider.
        It then computes
        \[ (i_1, \dots, i_n) \gets \select(\salt, \pin)\]
        to identify the cluster of $n$ HSMs it must communicate
        with during recovery.

        For each HSM $i_j \in \{i_1, \dots, i_n\}$, the client
        asks HSM $i_j$ to decrypt the ciphertext $\ct$.
        (Our full protocol requires interaction with the service
        provider, but we elide those details here.)
        The HSM, who holds the secret key $\sk_{i_j}$, computes:
        \[ \sigma_{j} \gets \dec(\sk_{i_j}, i_j, \ct).\]
        
        Finally, the client recovers its plaintext as
        \[ \msg \gets \reconst(\sigma_1, \dots, \sigma_n).\]

\end{enumerate}

\subsection{Definitions}
\label{sec:hide-defs:defs}

Intuitively, the correctness property states that
if a random $\ffail$ fraction of the secret keys
are unavailable, decryption still succeeds.

\begin{defn}[Correctness]\label{defn:lhe-cor}
Formally, we say that a location-hiding encryption scheme
on parameters $(N, n, \calP, \calM, \ffail, \fevil)$ is
\textit{correct} if, for all PINs $\pin \in \calP$,
all messages $\msg \in \calM$, on security parameter $\lambda \in \N$
\cref{exp:lhe-cor} outputs $1$, except with probability 
$\negl(\lambda)$.
\end{defn}

\noindent
This notion of correctness only guarantees
message reconstruction if some shares $\sigma_i$, 
computed using the $\dec$ routine, are
deleted or unavailable.
We do not consider the stronger notion of correctness,
in which message reconstruction is possible even if some
of the shares $\sigma_i$ are corrupted (rather than just missing).

\begin{figure}
\begin{framed}
\begin{experiment}[Location-hiding encryption: Correctness]\label{exp:lhe-cor}
We define the following correctness experiment,
which is parameterized by a location-hiding encryption
scheme with parameters $(N,n,\calP,\calM,\ffail,\fevil)$,
a PIN $\pin \in \calP$, 
a message $\msg \in \calM$, and 
a security parameter $\lambda \in \N$.

The experiment consists of the following steps:
\begin{align*}
  (\pk_i, \sk_i) &\gets \keygen(1^\lambda), \text{ for $i =1, \dots, N$}\\
    \salt &\getsr \zo^{\lambda}\\
    \ct &\gets \enc(\pk_1, \dots, \pk_N, \salt, \pin, \msg)\\
    F&\gets \emptyset\\
    &\text{Add each element of $\{1, \dots, N\}$ to $F$}\\[-4pt]
    &\text{with independent probability $\ffail$.}\\
    \{i_1, \dots, i_n\} &\gets \ms{Select}(\salt, \pin)\\
    \sigma_j &\gets \begin{cases}
      \bot & \text{for $j \in F$}\\
      \dec(\sk_{i_j}, i_j, \ct)&\text{otherwise}\\
          \end{cases}\\
    \msg' &\gets \reconst(\sigma_1, \dots, \sigma_n)
\end{align*}
The output of the experiment is 
$1$ if $\msg = \msg'$ and $0$ otherwise.
\end{experiment}
\end{framed}
\end{figure}

\medskip
The security property for location-hiding encryption 
states that no efficient adversary can
distinguish the encryption of two chosen
messages with probability much better than
guessing the PIN after roughly $N$ guesses.
This property should hold 
\emph{even if the adversary may adaptively corrupt} 
up to $\fevil \cdot N$ of the $N$ total secret keys.

For a location-hiding encryption scheme
on parameters $(N, n, \calP, \calM, \ffail, \fevil)$, 
an adversary $\calA$, and 
a security parameter $\lambda \in \N$, 
let $W_{\lambda,\beta}$ denote the probability
that the adversary outputs ``1'' in \cref{exp:lhe-sec}
with bit $\beta \in \zo$.
Then we define the advantage of $\calA$ at attacking
a location-hiding encryption scheme $\calE$ as:
\[ \LHEncAdv[\calA, \calE](\lambda) \deq \abs{W_{\lambda,0} - W_{\lambda,1}}.\]

\begin{defn}[Security]\label{defn:lhe-sec}
Let $\calE$ be a location-hiding encryption scheme. 
on parameters $(N, n, \calP, \calM, \ffail, \fevil)$, such that $N = \poly(\lambda)$, 
$n = \poly(\lambda)$, and $|\calP|$ and $|\calM|$ grow as
(possibly superpolynomial) functions of $\lambda$.
Then we say that the location-hiding encryption scheme $\calE$
is \textit{secure} if, for all efficient adversaries $\calA$, 
\[ \LHEncAdv[\calA, \calE](\lambda) \leq O(N/|\calP|) + \negl(\lambda).\]
\end{defn}

\begin{figure}[t]
\begin{framed}
  \begin{experiment}[Location-hiding encryption: Security]\label{exp:lhe-sec}
We define the following security experiment,
which is parameterized by an adversary $\calA$,
a location-hiding encryption
scheme with parameters $(N,n,\calP,\calM,\ffail,\fevil)$,
a security parameter $\lambda$,
and a bit $\beta \in \zo$.

\begin{itemize}
  \item The challenger runs:
    \begin{align*}
  (\pk_i, \sk_i) &\gets \keygen(1^\lambda), \text{ for $i =1, \dots, N$}\\
    \salt &\getsr \zo^\lambda\\
    \pin &\getsr \calP
    \end{align*}
    and sends $(\pk_1, \dots, \pk_N)$ to the adversary.
  \item The adversary chooses two messages $\msg_0, \msg_1 \in \calM$ and sends these
        to the challenger.
  \item The challenger computes
        \[\ct \gets \enc(\pk_1, \dots, \pk_N, \salt, \pin, \msg_\beta)\]
        and sends $(\salt, \ct)$ to the adversary.
  \item The adversary may make $\fevil \cdot N$ \emph{corruption queries}. 
        At each query:
        \begin{itemize}
          \item the adversary sends to the challenger an index $i \in [N]$ and
          \item the challenger sends to the adversary $\sk_i$.
        \end{itemize}
  \item Finally, the adversary outputs a bit $\beta'\in \zo$.
\end{itemize}

The output of the experiment is the bit $\beta'$.
\end{experiment}
\end{framed}
\end{figure}

\begin{remark}[Understanding the security definition.]\label{rem:under}
Since our security definition allows the adversary to corrupt up to
$\fevil \cdot N$ of the secret keys, the adversary can always reconstruct the plaintext
with probability roughly $\frac{\fevil \cdot N}{n|\calP|}$ using the following attack
to decrypt a ciphertext $\ct$ with salt $\salt$:
  \begin{itemize}
    \item Pick two messages $\msg_0,\msg_1 \getsr \calM$.
    \item Pick a candidate PIN $\pin' \getsr \calP$.
    \item Run $I = (i_1, \dots, i_n) \gets \select(\pin', \salt)$.
    \item Corrupt the keys in $I$ and use them to decrypt $\ct$.
      \begin{itemize}
        \item If $\ct$ decrypts to either $\msg_0$ or $\msg_1$, 
          guess the corresponding bit. 
        \item Otherwise, try again with another PIN.
    \end{itemize}
  \end{itemize}
If the attacker can corrupt $\fevil \cdot N$ keys, it can try at least $\fevil \cdot N/n$
PINs and its success probability is at least $\fevil \cdot N/(n|\calP|)$.
Our security analysis shows that this is nearly the best attack
possible against our location-hiding encryption scheme, up to low-order terms.
\end{remark}

\begin{remark}[Chosen-ciphertext security]
A stronger and more realistic definition of security
would allow the adversary to make decryption queries to
the HSMs in the penultimate stage of the attack game,
as in the standard chosen-ciphertext attack (CCA) security
game~\cite{RS91,CS98}.
Since our underlying public-key encryption scheme (hashed ElGamal)
is already CCA secure, we believe that---at the cost of some additional
complexity in the definitions and proofs---it would be able to achieve
a CCA-type notion for location-hiding encryption.
\end{remark}

\subsection{A tedious combinatorial lemma}

We will need the following lemma about random sets
to prove \cref{thm:lhe-sec}.
The proof of this lemma uses no difficult ideas, but keeping
track of the constants involved is a bit tedious.

\begin{defn}\label{defn:cover}
Let $n$, $N$, and $\Phi$ be positive integers. 
Say that a set $S \subseteq [N]$ ``$n/2$-covers''
a list $L \in [N]^n$ if at least $n/2$ elements
of the list $L$ appear in the set $S$.

Then, let $\Cover(\alpha, \beta)$ denote the probability, over the
random choice of lists $L_1, \dots, L_\Phi \getsr [N]^n$,
that there exists a set $S \subseteq [N]$ of size $\alpha N$ that $n/2$-covers
more than $\beta N$ of the lists.
\end{defn}

\begin{lemma}\label{lemma:annoying}
For $N > en \approx 2.71n$ and $\Phi < 2^{n/2}$, 
we have 
  \[ \Cover(\tfrac{1}{16}, \tfrac{3}{n}) \leq 2^{-N/4}.\]
\end{lemma}

\begin{proof}
  Let $\alpha = 1/16$ and $\beta = 3/n$. Our task is
then to bound the probability that there exists a set of size $\alpha N$ 
that $n/2$-covers more than $\beta N$ lists.

Fix a set $S \subseteq [N]$ of size $\alpha N$.
We compute the probability that $S$ $n/2$-covers more than $\beta N$
of the chosen lists.

The probability that a fixed set $S$ $n/2$-covers 
a list $L$ is at most ${n \choose {n/2}} \cdot \alpha^{n/2}$,
over the random choice of $L \getsr [N]^n$.
Using the inequality ${n \choose k} \leq (ne/k)^k$, we can 
bound this probability by $(2e \alpha )^{n/2}$.

Then, the probability that a fixed set $S$ $n/2$-covers 
some subset of $\beta N$ of the $\Phi$ lists
is then
\begin{align*}
  {\Phi \choose {\beta N}} \left( (2e\alpha)^{n/2} \right)^{\beta N} 
    \leq \left( \frac{\Phi e}{\beta N} \cdot (2e \alpha)^{n/2} \right)^{\beta N}
\end{align*}
Finally we apply the union bound over all ${N \choose {\alpha N}}$
possible choices of the set $S$ to get the final probability:
\begin{align*}
  &{N \choose {\alpha N}} \cdot \left( \frac{\Phi e}{\beta N} \cdot (2e \alpha)^{n/2} \right)^{\beta N}\\
  &\leq (e/\alpha)^{\alpha N} \cdot \left( \frac{\Phi e}{\beta N} \cdot (2e \alpha)^{n/2} \right)^{\beta N}\\
  &\leq (e/\alpha)^{\alpha N} \cdot \left( \frac{\Phi e}{\beta N} \cdot (2e \alpha)^{n/2} \right)^{\beta N} \\
  \intertext{and since $(e/\alpha)^\alpha = (16e)^{1/16} < 2^{1/2}$ and $\beta = 3/n$,}
  &\leq 2^{N/2} \cdot \left( \frac{\Phi e n}{3N} \cdot (2e \alpha)^{n/2} \right)^{\beta N} \\
  \intertext{and since we have assumed $N > en$, $en/N \leq 1$, so}
  &\leq \left[ 2^{1/2} \cdot \left( \Phi \cdot (2e \alpha)^{n/2} \right)^{\beta} \right]^N \\
  \intertext{and letting $\alpha = 1/16$ implies $2e \alpha = e/16 < 2^{-3/2}$, so}
  &\leq 2^{N/2} \left( \Phi \cdot 2^{-3n/4} \right)^{3N/n}\\
  \intertext{and since $\Phi < 2^{n/2}$,} 
  &\leq 2^{N/2} \left( 2^{-n/4} \right)^{3N/n} \leq 2^{-N/4}.
\end{align*}
\end{proof}

\subsection{Our construction}
\label{sec:hide-defs:const}

\renewcommand{\labelitemi}{$-$}
\begin{figure}[ht!]
\begin{framed}
\textbf{Location-hiding encryption scheme.} 
{\small 
The construction is parameterized by:
  \begin{itemize}
    \item a universe size $N \in \N$, 
    \item a cluster size $n \in \N$, 
    \item a PIN-space $\calP$,
    \item a recovery threshold $t \in \N$, 
    \item a public-key encryption scheme $(\PKEgen,\allowbreak \PKEenc,\allowbreak \PKEdec)$, 
and 
    \item an authenticated encryption scheme $(\AEenc,\allowbreak \AEdec)$ 
          with keyspace $\F$ (some finite field), 
      and 
    \item a security parameter $\lambda \in \N$.
\end{itemize}
The message space is $\zo^*$.
We use a hash function $\hash \colon \zo^\lambda \times \calP \to [N]^n$
(e.g., built from SHA-256)
which we model as a random oracle~\cite{BR93}.
We use $t$-out-of-$n$ Shamir secret sharing~\cite{S79};
we denote the sharing and reconstruction algorithms over a finite field
$\F$ by $(\Sshare_\F,\allowbreak\Sreconst_\F)$.
}

\medskip

  $\keygen(1^\lambda) \to (\pk, \sk)$. 
      \begin{itemize}
        \item On input security parameter $\lambda$, expressed in unary,
              run the key-generation routine for the underlying public-key
              encryption scheme: $(\pk, \sk) \gets \PKEgen(1^\lambda)$. 
      \end{itemize}

$\enc(\pk_1, \dots, \pk_N, \pin, \msg) \to (\salt, \ct)$.
      \begin{itemize}
        \item Sample a random salt $\salt \getsr \zo^\lambda$ and
              compute $(i_1, \dots, i_n) \gets \hash(\salt, \pin)$. 

        \item Sample a transport key $k \getsr \F$ 
              and split it using $t$-out-of-$n$-Shamir secret sharing~\cite{S79}:\\
              $(k_1, \dots, k_n) \gets \Sshare_{\F}(k)$.

        \item Encrypt each share of the key using the underlying public-key encryption scheme. 
              That is, for 
          $j \in \{1, \dots, n\}$, set  $C_i \gets \PKEenc\big(\pk_{i_j}, k_i\big)$.

        \item Set $M \gets \AEenc(k, \msg)$, 
              $\ct \gets (M, C_1, \dots, C_n)$, and 
              output $(\salt, \ct)$.
      \end{itemize}

$\select(\salt, \pin) \to (i_1, \dots, i_n) \in [N]^n$.
      \begin{itemize}
        \item Output $\hash(\salt, \pin) \in [N]^n$.
      \end{itemize}

$\dec(\sk, i, \ct) \to \sigma_i \in \G$.
      \begin{itemize}
        \item Parse $\ct$ as $(M, C_1, \dots, C_n)$.
        \item Output $\sigma \gets \big(M,\PKEdec(\sk, C_i)\big)$.
      \end{itemize}
$\reconst(\sigma_1, \dots, \sigma_n) \to \msg$. 
      \begin{itemize}
        \item For $i \in [n]$, Parse each share $\sigma_i$ as a pair $(M_i, k_i)$.
        \item Let $k \gets \Sreconst_\F(k_1, \dots, k_n) \in \F$.
        \item Let $M$ be the most common value in $\{M_1, \dots, M_n\}$.
        \item Output $\AEdec(k, M)$.
      \end{itemize}
      
\end{framed}
\caption{Our construction of location-hiding encryption.}
\label{fig:lhe-const}
\end{figure}
\renewcommand{\labelitemi}{$\bullet$}

Our construction, which is parameterized by a
standard public-key encryption scheme,
appears in \cref{fig:lhe-const}.
We instantiate the public-key encryption
scheme with hashed ElGamal encryption, 
which we recall here.

\paragraph{Hashed ElGamal} \label{sec:hide-defs:elgamal}
We instantiate the location-hiding encryption scheme of
\cref{sec:hide} with the ``Hashed ElGamal'' 
encryption scheme. 
The scheme uses a cyclic group $\G = \langle g \rangle$ 
of prime order $p$,
in which we assume that the computational Diffie-Hellman 
problem is hard.
It also uses an authenticated encryption scheme
$(\AEenc, \AEdec)$ with key space $\calK$ and a
hash function $\hash': \G \to \calK$.

A keypair is a pair $(x, g^x) \in \Z_p \times \G$,
where $x \getsr \Z_p$.
To encrypt a message $m \in \zo^\ell$ 
to public key $X \in \G$, the encryptor
computes $\big(g^r, \AEenc(\hash'(X^r), m)\big)$.

A standard argument~\cite{BSbook} shows that Hashed ElGamal
satisfies semantic security against
chosen-ciphertext attacks~\cite{RS91,CS98}.

\medskip

In our use of hashed ElGamal, we can provide domain separation
between different ciphertexts by prepending inputs to the
hash function $\hash'$ during encryption and decryption with:
(1) the client's username,
(2) the salt associated with the ciphertext, and
(3) the public keys of the $n$ public keys to which the
    client encrypted the ciphertext.
All of these values are available to the client during
encryption and to the HSMs during decryption.

\subsection{Proof of correctness}

We now prove:
\begin{theorem} \label{thm:lhe-cor}
Consider an instantiation of the encryption scheme $\calE$ of \cref{fig:lhe-const} 
with parameters $(N, n, \calP, t)$,
with cluster size $n = \Omega(\lambda)$, on security parameter $\lambda$, 
and recovery threshold $t = n/2$.
Then the resulting scheme is a \emph{correct} location-hiding encryption
scheme (in the sense of \cref{defn:lhe-cor})
for any fault-tolerance parameter $\ffail \leq \tfrac{1}{8}$.
\end{theorem}

\begin{proof}
To prove the claim, we need to bound the probability that,
out of a random size-$n$ subset of all $N$ keys,
no more than $t = n/2$ are failed.
Fix a size-$(n/2)$ subset $S$ of the $n$ chosen keys.
The probability that the keys in $S$ are all failed
is $(\ffail)^{n/2}$.
We can now take a union bound over all ${n \choose {n/2}}$
choices of $S$ to bound the final failure probability:
\begin{align*}
  \Pr[\text{fail}] &\leq {n \choose {n/2}} \cdot (\ffail)^{n/2}
    \leq 2^n \cdot \left(\frac{1}{2^3}\right)^{n/2} = 2^{-n/2}.
\end{align*}
Since we have taken $n = \Omega(\lambda)$, this probability is negligible in~$\lambda$.
\end{proof}

\subsection{Proof of security}
\label{sec:hide-defs:proof}

We state the main theorem and then prove it here.

\begin{theorem} \label{thm:lhe-sec}
The location-hiding encryption scheme $\calE$ of \cref{fig:lhe-const} 
instantiated with the hashed-ElGamal 
public-key encryption system (\cref{sec:hide-defs:elgamal})
is \emph{secure}, in the sense of
\cref{defn:lhe-sec}, in the random-oracle model,
for recovery threshold $t = n/2$ and
corruption threshold $\fevil = 1/16$, 

More precisely, we consider an
instantiation of our location-hiding
encryption scheme $\calE$ with
parameters $(N, n, \calP, t= n/2)$,
where $N > e\cdot n \approx 2.71n$ and $|\calP| < 2^{n/2}$,
using hashed ElGamal encryption for the public-key encryption scheme,
and using an arbitrary authenticated encryption scheme $\AEscheme$.

Then, let $\calA$ be an adversary that attack $\calE$ in the
location-hiding encryptions security game with corruption threshold
$\fevil = 1/16$.
Denote $\calA$'s attack advantage as $\LHEncAdv[\calA,\calE]$,
Then if $\calA$ makes at most $Q$ queries
to the random oracle $\hash'$, used in hashed ElGamal encryption,
we construct:
\begin{itemize}
\item an efficient algorithm $\calB_\CDH$ that breaks CDH in group $\G$ with
advantage $\CDHAdv[\calB_\CDH, \G]$ and 
\item an efficient
algorithm $\calB_\AEscheme$ that breaks the authenticated
encryption scheme $\AEscheme$ with advantage $\AEAdv[\calB_\AEscheme,\AEscheme]$
\end{itemize}
such that
\begin{multline*}
    \LHEncAdv[\calA, \calE] \leq 2^{-N/4} 
    + N \cdot Q \cdot \CDHAdv[\calB_\CDH,\G] \\
    + \frac{3N}{n\abs{\calP}} 
    + \AEAdv[\calB_\AEscheme, \AEscheme].
\end{multline*}
\end{theorem}

\begin{figure}
{\small 
\begin{framed}
\paragraph{Game 0.}
We define Game 0 by instantiating Experiment~\ref{exp:lhe-sec}
our location-hiding encryption scheme, instantiated in turn 
with hashed-ElGamal encryption (\cref{sec:hide-defs:elgamal}).
Game 0 is parameterized by an adversary $\calA$,
a group $\G$ of prime order $p$ with generator $g$,
a universe size $N \in \N$,
a cluster size $n \in \N$,
a PIN-space $\calP$, 
a recovery threshold $t$, 
a security parameter $\lambda$,
hash functions $\Hash: \zo^\lambda \times \calP \to [N]^n$ and
$\Hash': \G \to \F$ (which we model as random oracles
where the challenger plays the role of the random oracle),
an authenticated encryption scheme $(\AEenc, \AEdec)$ with
keyspace $\F$ and
message space $\zo^*$,
a security parameter $\lambda \in \N$, and
a bit $\beta \in \zo$.

\begin{itemize}
  \item The challenger runs
    \begin{align*}
    \sk_i &\getsr \Z_p,\ \pk_i \gets g^{\sk_i} \ \text{for $i \in \{1, \dots, N$\}}
    \end{align*}
    and sends $(\pk_1, \dots, \pk_N)$ to the adversary.
    \item The adversary chooses two messages $\msg_0, \msg_1 \in \calM$ and sends these
        to the challenger.
    \item The challenger then computes the ciphertext using our location-hiding
    encryption scheme instantiated with hashed-ElGamal.
    That is, the challenger computes
    \begin{align*}
      \salt &\getsr \zo^\lambda\\
      \pin &\getsr \calP\\
      (i_1, \dots, i_n) &\gets \hash(\salt, \pin)\\
      k &\getsr \F \\
      (k_1, \dots, k_n) &\gets \Sshare_{\F}(k).
      \intertext{Here, $\Sshare_\F$ denotes $t$-out-of-$n$ Shamir secret sharing over
        the field $\F$.
        The challenger then encrypts the shares
      $(k_1, \dots, k_n)$
      of the transport key $k$ using hashed ElGamal encryption.
      For $j \in \{1, \dots, n\}$, the challenger computes:}
      &\quad r_j \getsr \Z_p \\
        &\quad \kappa_j \gets \Hash'((\pk_{i_j})^{r_j})\\
      &\quad C_j \gets (g^{r_j}, \kappa_j \oplus k_j).
      \intertext{Finally, the challenger encrypts the message $\msg$ with
      the transport key $k$ and outputs the ciphertext:}
      M &\gets \AEenc(k, \msg_\beta) \\
      \ct &\gets (M, C_1, \dots, C_n)
    \end{align*}
    and sends $(\salt, \ct)$ to the adversary.
  \item The adversary may make adaptive $\fevil \cdot N$ \emph{corruption queries}
        and may perform computation between its queries.
        At each query:
        \begin{itemize}
          \item the adversary sends to the challenger an index $i \in [N]$ and
          \item the challenger sends to the adversary $\sk_i$.
        \end{itemize}
  \item Finally, the adversary outputs a bit $\beta'\in \zo$.
\end{itemize}

The output of the experiment is the bit $\beta'$.
\end{framed}
}
\end{figure}

Notice that if the number of HSMs $N$ grows much larger
than the size of the PIN space $|\calP|$, the bound of \cref{thm:lhe-sec} 
on the adversary's advantage becomes vacuous.
(In fact, this limitation is inherent---see \cref{rem:under}
in \cref{sec:hide-defs}.)
To support extremely large data deployments, it would be
possible to shard users into data centers of moderate
size (e.g., $N \approx 50,000$). 

The main technical challenge in proving \cref{thm:lhe-sec} comes
from the fact that the adversary may adaptively compromise a subset
of the secret keys.
This is very similar to the issues that arise when proving 
a cryptosystem secure against ``selective opening attacks''~\cite{DNRS99,BHY09} 

\begin{proof}[Proof of \cref{thm:lhe-sec} (Security)]
We construct the following series of games and show that the difference between
the adversary's advantage from one game to the next is small.
For $i \in \{0, \dots, 4\}$, let $W_i$ the event that the adversary wins in 
Game $i$, where we define the winning condition for each game below.

\paragraph{Game 0.} Game 0 proceeds according to the location-hiding security game
where the challenger plays the role of the random oracle.
For location-hiding encryption scheme $\cal$ and adversary $\calA$ in Game 0,
by definition,
\begin{align} 
  \Pr[W_0] = \LHEncAdv[\calA, \calE]~. \label{eq:game0}
\end{align}

\paragraph{Game 1.} Game 1 proceeds in the same way as Game 0, except that for the
adversary to win in Game 1, we also require that a certain ``bad event'' does not take place.
This bad event is that many different PINs hash to the same set of HSMs.
Formally, the bad event is that there exists a set $S$ of $\ffail \cdot N = N/16$ HSMs
and a set of $\tfrac{3}{n}\cdot N$ PINs $\{\pin_1, \pin_2, \dots\}$
such that, for each $i \in [N/n]$, it holds that:
$|\hash(\salt, \pin_i) \cap S| \geq t = n/2$.
Using Definition~\ref{defn:cover}, we can simply write the probability of this bad
event as $\Cover(\frac{1}{16}, \frac{3}{n})$.

By Lemma~\ref{lemma:annoying}, we have that
\begin{align} 
\abs{\Pr[W_0] - \Pr[W_1]} \leq 2^{-N/40}~. \label{eq:game1}
\end{align} 

\paragraph{Game 2.} Game 2 proceeds in the same way as Game 1 except that
for the adversary to win, we require that
\begin{itemize}
    \item the adversary wins in Game 1, and
    \item the adversary never makes a random-oracle query for $(\pk_i)^{r_i}$
          where $i \in (i_1, \dots, i_n)$ unless the adversary
          issued a corruption query for $i$ first.
\end{itemize}

Then by Lemma~\ref{lemma:game2},
\begin{align} 
\abs{\Pr[W_1] - \Pr[W_2]} < N \cdot Q \cdot \CDHAdv[\calB_\CDH, \G]~. \label{eq:game2}
\end{align} 

\paragraph{Game 3.} Game 3 proceeds as in Game 2 except that for the
adversary to win, we require that
\begin{itemize}
    \item the adversary wins in Game 2, and
    \item the adversary makes fewer than $t=n/2$ corruption queries
        to elements in $I = (i_1, \dots, i_n)$.
\end{itemize}
Then, by Lemma~\ref{lemma:game3},
\begin{align} 
\abs{\Pr[W_2] - \Pr[W_3]} \leq \frac{3N}{n|\calP|} \label{eq:game3}
\end{align} 

\paragraph{Game 4.} In Game 4, 
we modify the behavior of the challenger. 
Rather than encrypting $\msg$ with $k$, instead the challenger samples an
additional key $k' \getsr \F$ for an authenticated encryption scheme
and uses $k'$ to encrypt $\msg$.

Then, by Lemma~\ref{lemma:game4},
\begin{align} 
\abs{\Pr[W_3] - \Pr[W_4]} = 0. \label{eq:game4}
\end{align} 

\paragraph{Final reduction.}
In Lemma~\ref{lemma:final}, we show that
\begin{align} 
\abs{\Pr[W_4]} = \AEAdv[\calB_\AEscheme, \AEscheme]~. \label{eq:final}
\end{align} 

\paragraph{Putting it together.}

We can write $\Pr[W_0]$ as
\begin{align*}
  \Pr[W_0] \leq& \abs{\Pr[W_0] - \Pr[W_1]} \\
      +&\abs{\Pr[W_1] - \Pr[W_2]} \\
      +&\abs{\Pr[W_2] - \Pr[W_3]} \\
      +&\abs{\Pr[W_3] - \Pr[W_4]} + \Pr[W_4].
\end{align*}

Then, by (\ref{eq:game0})-(\ref{eq:final}), 
we can bound $\LHEncAdv$ as follows:
\begin{align*}
  \LHEncAdv[A, \calE] \leq &2^{-N/40} \\
+&N \cdot Q \cdot \CDHAdv[\calB_\CDH, \G] \\
+&\frac{3N}{n|\calP|} \\
+&\AEAdv[\calB_\AEscheme, \AEscheme]~.
\end{align*}
\end{proof}

\begin{lemma} \label{lemma:game2}
Let $W_1$ be the event that the adversary wins in Game 1 and $W_2$
be the event that the adversary wins in Game 2.
In particular, for every adversary $\calA$ in Game 1,
we construct a CDH adversary $\calB_\CDH$ in group $\G$
with advantage $\CDHAdv[\calB_\CDH, \G]$
that runs in time linear in the runtime of $\calA$ and
makes $Q$ $\Hash'$-oracle queries
such that
\[ \abs{\Pr[W_1] - \Pr[W_2]} \leq N \cdot Q \cdot \CDHAdv[\calB_\CDH, \G]~.\]
\end{lemma}

\xxx{Clarify that there is a hash oracle and a corruption oracle}
\begin{proof}
Let $F$ be the event that the adversary makes a $\Hash'$-oracle 
query at the point $\pk_i^{r_i}$ \textit{before} making a corruption query
for key $i$.
Then, by the definition of Games 1 and 2, and the Difference Lemma~\cite{difference-lemma},
we have
\[ \abs{\Pr[W_1] - \Pr[W_2]} \leq \Pr[F].\]
So, to prove the lemma we need only bound $\Pr[F]$.

Given an adversary $\calA$ in Game 1, we will construct a CDH adversary
$\calB_\CDH$ such that $\Pr[F] \leq N\cdot Q\cdot  \CDHAdv[\calB_\CDH, \G]$.
We construct $\calB_\CDH$ as follows:
\begin{itemize}
  \item The algorithm $\calB_\CDH$ receives a challenge tuple 
        $(g, g^r, g^x) \in \G^3$ from the CDH challenger.
  \item The algorithm $\calB_\CDH$ plays the role of the Game-1 challenger.
    The algorithm $\calB_\CDH$ deviates from the normal behavior of the 
    Game-1 challenger in two ways:
    \begin{enumerate}
      \item The challenger chooses a random value $i^* \getsr [N]$.
            For the $i^*$th keypair, algorithm $\calB_\CDH$ sets
            $\pk_i \gets g^x$, where $g^x$ is the value from the CDH
            challenger.
            (Algorithm $\calB_\CDH$ generates all other keypairs 
            as the challenger in Game 1 does.)
      \item When constructing the final ciphertext, algorithm $\calB_\CDH$
            sets the encryption nonce in the $i^*$th ciphertext
            to be the value $g^r$, where $g^r$ is the value 
            from the CDH challenger.
            For the value $\hash'((\pk_{i^*})^r)$ needed to construct
            the ciphertext, the algorithm $\calB_\CDH$ just chooses
            a random element in $\F$.
    \end{enumerate}
  \item The algorithm $\calB_\CDH$ then proceeds in the same way as the Game 1 challenger
        with the following modification:
        \begin{itemize}
        \item If $\calA$ makes a corruption query for $i^*$, $\calB_\CDH$ outputs $\bot$.
        \end{itemize}
    \item If $\calB_\CDH$ did not output $\bot$, 
          then $\calB_\CDH$ randomly chooses one of the
          points at which $\calA$ made an $\hash'$-oracle query 
          and returns the queried point to the CDH challenger.
\end{itemize}

We now compute algorithm $\calB_\CDH$'s CDH advantage.
Whenever event $F$ occurs, the algorithm $\calA$ makes a random-oracle
query to a point $(\pk_i)^{r_i}$, for some $i \in [N]$, before issuing
a corruption query at $i$.
Notice that $\calB_\CDH$ succeeds whenever:
\begin{enumerate}
  \item event $F$ occurs, 
  \item $i=i^*$, and
  \item the algorithm $\calB_\CDH$ guesses the correct random-oracle query
        to output.
\end{enumerate}

These three events are independent.
Furthermore, their probabilities are:
\begin{enumerate}
  \item $\Pr[F]$ -- to be computed later, 
  \item $\Pr[i=i^*] = 1/N$, and
  \item $\Pr[\text{guesses correct r.o.~query}] = 1/Q$.
\end{enumerate}
Therefore, 
\begin{align*}
  \CDHAdv[\calB_\CDH, \G] \geq \Pr[F] \cdot \left(\frac{1}{N}\right) \cdot \left(\frac{1}{Q}\right),
\end{align*}
which proves the lemma.
\end{proof}

\begin{lemma} \label{lemma:game3}
Let $W_2$ be the event that the adversary wins in Game 2 and $W_3$
be the event that the adversary wins in Game 3. Then
    \[ \abs{\Pr[W_2] - \Pr[W_3]} \leq \frac{3N}{n|\calP|}~.\]
\end{lemma}

\begin{proof}
Let $F$ be the event that the adversary $\calA$ wins in Game 2
but not in Game 3.
By construction of the games, we have 
\[ \abs{\Pr[W_2] - \Pr[W_3]} \leq \Pr[F], \]
so our task is to bound $\Pr[F]$.

To analyze $\Pr[F]$, we consider a modified game between $\calA$
and its challenger in Games 2 and 3.
Here, we modify the challenger to halt the execution of $\calA$
as soon as event $F$ occurs.
This modification cannot increase $\Pr[F]$.

However, in the modified game, $\calA$'s view is 
\emph{independent of the uncorrupted elements of $I$}
until event $F$ occurs.
This is so because the only values that the adversary
sees that depend on the set $I$ are the $\kappa_j$ values.
Since these are computed as the output of $\hash'((\pk_i)^{r_j})$,
until the adversary queries the random oracle $\hash'$
at the point $(\pk_i)^{r_j}$, these $\kappa$
values are also independent of $I$.

Since we have already argued (Game 2)
that the challenger never queries the random oracle at
a point $(\pk_i)^{r_j}$ without having corrupted $i$,
the ciphertext values that the adversary sees
are independent of the uncorrupted elements of $I$.

Therefore, the answers to the corruption queries give the 
adversary no information on the uncorrupted elements of $I$.
Then $\Pr[F]$ is just the probability that the adversary makes
$N/16$ \emph{non-adaptive} corruption queries and is able to 
cause event $F$ to occur.

By the winning condition of Game 1, for any set of $N/16$ corrupted
HSMs $S$,
there are at most $3N/n$ PINs $\pin$ such that
$\abs{\hash(\salt, \pin) \cap S} > n/2$.
Therefore, the probability that $F$ occurs is
$\Pr[F] \leq (3N/n)/\abs{\calP} = 3N/(n \abs{\calP})$.
\end{proof}

\begin{lemma} \label{lemma:game4}
Let $W_3$ be the event that the adversary wins in Game 3 and $W_4$
be event that the adversary wins in Game 4. Then
\[ \abs{\Pr[W_3] - \Pr[W_4]} = 0~.\]
\end{lemma}

\begin{proof}
We use the security of Shamir secret sharing to prove this lemma.
For each of the indexes that the adversary does not corrupt, we
can replace $\kappa_i \oplus k_i$ with a random element in $\F$ in both
games.
We know that in both games, the adversary corrupts fewer than
$n/2$ of the keys in $I$, and so the adversary learns fewer
than $n/2$ of the shares $k_1, \dots, k_n$ of the transport key.
Therefore, we can replace the transport key $k$ with $k' \getsr \F$,
completing the proof.
\end{proof}

\begin{lemma} \label{lemma:final}
Let $W_4$ be event that the adversary wins in Game 4. Then
given an adversary $\calA$ in Game 4, we construct an AE adversary
$\calB_\AEscheme$ for AE scheme $\AEscheme$ with advantage
$\AEAdv[\calB_\AEscheme, \AEscheme]$ that runs in time linear
in the runtime of $\calA$ such that
\[ \Pr[W_4]  = \AEAdv[\calB_\AEscheme, \AEscheme]~.\]
\end{lemma}

\begin{proof}
Given an adversary $\calA$ in Game 4, we will construct an adversary $\calB_\AEscheme$
attacking the authenticated encryption scheme $\AEscheme$ such that the advantage
of $\calB_\AEscheme$ is identical to that of $\calA$. We construct $\calB_\AEscheme$
as follows:
\begin{itemize}
    \item The algorithm $\calB_\AEscheme$ computes $(\pk_1, \dots, \pk_N)$
        in the same way as the Game 4 challenger and sends them to $\calA$.
    \item When $\calA$ returns the messages $\msg_0, \msg_1 \in \calM$,
        $\calB_\AEscheme$ forwards the messages to the AE challenger
        and receives the ciphertext $c^*$ from the AE challenger.
    \item The algorithm $\calB_\AEscheme$ then computes the ciphertext in the same way
        as the Game 4 challenger except that instead of computing the ciphertext
        as $(\AEenc(k,\msg),C_1, \dots, C_n)$, $\calB_\AEscheme$ sends $\calA$
        $(c^*, C_1, \dots, C_n)$.
    \item The algorithm $\calB_\AEscheme$ responds to queries in the same way as the
        Game 4 challenger.
    \item When $\calA$ outputs bit $\beta' \in \zo$, $\calB_\AEscheme$ forwards the
        response to the AE challenger.
\end{itemize}
Because in Game 4, we sample $k'$ independently from the rest of the messages, $\calA$
cannot distinguish between interaction with $\calB_\AEscheme$ and the Game 4 challenger.
The advantage of $\calA$ is exactly $\AEAdv[\calB_\AEscheme, \AEscheme]$, completing
the proof.
\end{proof}
 \section{Implementing the log data structure}
\label{app:logx}

\subsection{Security properties}
\label{app:logx:sec}

We need the following properties to hold:

\paragraph{Inclusion completeness.}
Informally, $\VerifyIncludes$ accepts valid log-inclusion
proofs produced by $\ProveIncludes$.
That is, for all logs $L$, 
all identifier-value pairs $(\id, \val) \in L$,
if $d \gets \Digest(L)$ and 
$\pfinc \gets \ProveIncludes(L, \id, \val)$, then
$\VerifyIncludes(d, \id, \val, \pfinc) = 1$.

\paragraph{Inclusion soundness.}
Informally, for all efficient adversaries that output
a log $L$, an identifier-value pair $(\id, \val) \not \in L$, and
a false inclusion proof $\pfinc^*$,
it holds that, for digest $d \gets \Digest(L)$, the probability
$\Pr[\VerifyIncludes(d, \id, \val, \pfinc^*) = 1]$
is negligible. 

\paragraph{Extension completeness.}
Informally, $\VerifyExtends$ accepts valid log-extension proofs
produced by $\ProveExtends$.
That is, for all logs $L$ and $L'$, where $L'$ extends $L$,
if $d \gets \Digest(L)$, $d' \gets \Digest(L')$, and
$\pfext \gets \ProveExtends(L, L')$, then 
$\VerifyExtends(d, d', \pfext) = 1$.

\paragraph{Extension soundness.}
Informally, for all efficient adversaries that output
a pair of logs $L$ and $L'$
(such that $L'$ does \emph{not} extend $L$
and $L$ does not contain duplicate identifiers),
and a false proof $\pfext^*$, if we compute
$d \gets \Digest(L)$ and $d' \gets \Digest(L')$, the probability
$\Pr[\VerifyExtends(d, d', \pfext) = 1]$ is negligible.

\subsection{Implementation using Merkle trees}
\label{app:logx:impl}

We now describe how to implement the 
routines defined in \cref{sec:log:syntax}, 
using a construction as in 
Nissim and Naor~\cite{NN98}. 
In the following discussion, we define the
``log tree'' for a log $L$ to be a binary 
search tree, ordered by identifiers $\id$.
Each internal node in the tree---in addition to containing
a value---contains the cryptographic hash of its two child 
nodes, as in a Merkle tree.

\smallskip
\begin{hangparas}{\defhangindent}{1}

$\Digest(L) \to d$.
Construct the  log tree for $L$ by inserting the
elements of $L$ into a binary-search tree 
one at a time.
(We can use any type of self-balancing binary-search tree here.)
As the digest $d$, output the hash of the root of the log tree. 

$\ProveIncludes(L, \id, \val) \to \{\pfinc, \bot\}$.
If $(\id, \val) \not \in L$, output $\bot$.
Otherwise, output the Merkle 
inclusion proof that proves that $(\id, \val)$
is in the log tree rooted at $d = \Digest(L)$.

$\VerifyIncludes(d, \id, \val, \pfinc) \to \zo$. 
Treating the digest $d$ as a log-tree root and $\pfinc$
as a Merkle proof of inclusion, verify that $(\id, \val)$
is included in the log tree rooted at~$d$.

$\ProveExtends(L, L') \to \{\pfext, \bot\}$.
We show how the routine works in the special case
in which $L'$ contains exactly one entry $(\id, \val)$ 
that does not appear in $L$.
To generalize to the case in which there are many 
new entries in $L'$, we run this routine once
for each new entry and output the concatenation of all
of the resulting proofs. \\[2pt]
If $L'$ does not extend $L$, output $\bot$.
Otherwise, find the identifiers $\id_\mathsf{left}$ and $\id_\mathsf{right}$
that appear just before and after $\id$ in the lexicographical
ordering of identifiers in the old log $L$.
Let their corresponding values be $\val_\mathsf{left}$
and $\val_\mathsf{right}$.\\[2pt]
The first portion of the proof
$\pfext$ is a Merkle proof of inclusion 
of $(\id_\mathsf{left}, \val_\mathsf{left})$
and $(\id_\mathsf{right}, \val_\mathsf{right})$
in the old log tree rooted at digest $d = \Digest(L)$.
These proofs prove that the new identifier $\id$ is not in the log 
represented by the old digest $d$.\\[2pt]
Next, insert $(\id, \val)$ into the log tree for $L$ to get a
log tree for $L'$ and its corresponding digest.
The second portion of the proof $\pfext$ is a Merkle proof of 
inclusion of every node in the log tree for $L'$ that
does not appear in the log tree for $L$.
These proofs prove that the new digest $d'$ represents
the root of $L$'s log tree, with the new pair $(\id, \val)$
inserted.

$\VerifyExtends(d, d', \pfext) \to \zo$. 
Parse $\pfext$ as a series of Merkle inclusion proofs
over log trees. 
Check that $\id$ lies between $\id_\mathsf{left}$
and $\id_\mathsf{right}$ in lexicographic order.
Verify each of the Merkle proofs generated
in $\ProveExtends$.
Checking the Merkle proofs in this case also requires
verifying that the nodes in log trees satisfy the
proper ordering constraints of a binary search tree: 
the left child's value
is less than the parent's value and the right child's
value is greater than the parent's value.
Accept if all of these proofs accept.
\end{hangparas}

\paragraph{Security properties.}
The completeness properties are immediate.
Inclusion soundness follows directly from the analysis
of Merkle proofs, which in turn rely on the collision-resistance
of the underlying hash function.
To sketch the argument for extension soundness:
The first part of the proof $\pfext$ convinces
the verifier that the new identifier $\id$ does not 
appear in the log $L$ represented by the old digest $d$.
Therefore, the log $L'$ that digest $d'$ represents must not 
contain any duplicate identifiers, since $L$ contains no 
duplicate identifiers.
The second part of the proof $\pfext$ convinces
the verifier that log tree rooted at $d'$ is
identical to the log tree rooted at $d$, except
with the pair $(\id, \val)$ added.

\subsection{Making progress in spite of failures}
\label{sec:logx:fail}

If frequent node failures prevent the log from making progress, the
HSMs can run the following more complicated log-update protocol.
In this variant, each HSM chooses which log chunks to audit 
as a deterministic function of the Merkle root $R$
and the HSM's own node ID.

Provided that we choose the constant $C$ large enough,
for \emph{every} choice of the Merkle root, 
$R$ at least one honest HSM will audit each log chunk.
In particular, by taking $C \geq 384$ the probability that no
honest HSM audits a particular chunk is less than $2^{-384}$.
The probability that there exists a 256-bit Merkle root $R$ that causes
no honest HSM to audit a particular chunk is then at most 
$2^{256} \cdot 2^{-384} = 2^{-128}$.
So, no matter how the service provider influences the Merkle root $R$,
at least one honest HSM will attempt to audit each chunk.

Using this method, given the Merkle root $R$,
every HSM can deterministically compute
the set of log chunks that every other HSM will audit.
Then, if any HSM fails during the audit process, 
the remaining non-failed HSMs can recursively 
run our randomized-checking protocol to check the log 
chunks that the failed HSMs would have checked.

In this way, the protocol can make progress even if 
HSMs fail during the log-update process.

 \section{Details on outsourced storage\\ with secure deletion}
\label{app:puncx}

We model the external service provider as an oracle $\calS$ to 
which the HSM has access. To read and write blocks from external
storage at address $\addr \in \N$ 
the HSM can call $\calS.\mathsf{Get}(\addr) \to \block$
or $\calS.\mathsf{Put}(\addr, \block)$.

The HSM needs to implement the following functions:

\medskip 

\begin{hangparas}{1em}{1}
$\setup^\calS(\data_1, \dots, \data_D) \to \sk$. 
  Store the given data blocks in the external
  storage system $\calS$ encrypted under a secret key $\sk$.
  The HSM stores only $\sk$ in its internal memory.

$\Read^\calS(\sk, i) \to \data_i$ or $\bot$.
  Use key $\sk$ to read the $i$th logical
  data block from the storage system $\calS$.
  Return $\bot$ on failure.

$\Delete^\calS(\sk, i) \to \sk'$ or $\bot$.
  Use key $\sk$ to delete the $i$th logical
  data block from the external storage system $\calS$.
  Return a new key $\sk'$ to store in 
  internal memory.
  Return $\bot$ on failure.

\end{hangparas}

\medskip

Let $(\allowbreak \AEenc,\allowbreak\AEdec)$ be a symmetric-key
authenticated encryption scheme such as AES-GCM, with key space $\calK$\@.  We store data
blocks encrypted as leaves of a binary tree of height $h =
1+\lceil\log_2(D)\rceil$.  We store leaf $i$ at address $2^h+i$, block
$a$'s parent at address $\lfloor a/2\rfloor$, and $a$'s left and right
children at $2a$ and $2a+1$, respectively.  

For convenience, we have
optional parameters that take a non-default value only on recursive
calls.  We bundle outsourced blocks into a 0-indexed vector that can
be sliced.  We assume slicing beyond the end of an array yields a
zero-length slice, and assigning $\sk_\ell\|\sk_r\gets\msg$, if $|\msg|$
is the length of only one key, sets $\sk_\ell\gets\msg$ and
$\sk_r\gets\mathsf{empty}$.  We assume if $\AEdec$ fails, it raises an
exception caught outside of these routines.

\renewcommand{\labelitemi}{$-$}
\renewcommand{\labelitemii}{$\star$}
\def\addr{\mathsf{addr}}
\def\level{\mathsf{level}}
\medskip
\medskip
\noindent
$\setup^\calS(\data, \addr=1) \to \sk$.
\begin{itemize}
  \item Let $D\gets\mathrm{len}(\data)$
  \item If $D = 0$, set $\msg \gets 00\cdots0$
  \item Else if $D = 1$, set $\msg \gets \data[0]$
  \item Else ($D > 1$)
  \begin{itemize}
  \item Run $\sk_\ell \gets \setup^\calS(\data[0:\lfloor D/2\rfloor],
    2\cdot\addr)$
  \item Run $\sk_r \gets \setup^\calS(\data[\lfloor D/2\rfloor:n],
    2\cdot\addr+1)$
  \item Let $\msg \gets \sk_\ell \| \sk_r$
  \end{itemize}
  \item Let $\sk\getsr \calK$
  \item $\calS.\mathsf{Put}(\addr, \AEenc(\sk, \msg))$
  \item Return $\sk$
\end{itemize}

\medskip
\noindent
$\Read^\calS_h(\sk, i, \level=h-1) \to \msg$.
\begin{itemize}
\item Let $\addr\gets \lfloor (2^h+i)/2^\level\rfloor$
\item Let $c\gets \calS.\mathsf{Get}(\addr)$
\item If $\level=0$ then return $\AEdec(\sk, c)$
\item If $i\&2^{\level-1} = 0$ (i.e., $i$ in left child),\\
  let $\sk\|\_\gets \AEdec(\sk,c)$
\item Otherwise (right child), let $\_\|\sk\gets \AEdec(\sk,c)$
\item Return $\Read^\calS(\sk, i, \level-1)$
\end{itemize}

\medskip
\noindent
$\Delete^\calS_h(\sk, i, \level=h-1) \to \sk$.
\begin{itemize}
\item If $\level=0$ return $00\cdots0$ (useless encryption key)
\item Let $\addr\gets \lfloor (2^h+i)/2^\level\rfloor$
\item Let $\sk_\ell\|\sk_r\gets \AEdec(\sk, \calS.\mathsf{Get}(\addr))$
\item If $i\&2^{\level-1} = 0$ (i.e., $i$ in left child)\\
  let $\sk_\ell \gets\Delete^\calS_h(\sk, i, \level-1)$
\item Else (right child) let $\sk_r \gets\Delete^\calS_h(\sk, i, \level-1)$
  \item Let $\sk\getsr \calK$
  \item $\calS.\mathsf{Put}(\addr, \AEenc(\sk, \sk_\ell\|\sk_r))$
  \item Return $\sk$
\end{itemize}
 \fi

\section{Artifact Appendix}

\subsection{Abstract}

The SafetyPin implementation is split into two components:
\begin{itemize}
    \item \textbf{HSM:} The HSMs (hardware security modules) are used to recover user secrets. Our implementation
        uses SoloKeys\com{~\cite{solokeys}}, which are low-cost HSMs. We add roughly 2,500
        lines of C code to the open-source SoloKey firmware.
    \item \textbf{Host:} The host implements functionality for the user and data center,
       including saving secrets, maintaining the log, and coordinating HSMs. Our
        implementation is roughly 3,800 lines of C/C++ code.
\end{itemize}
We implement the protocol described in the paper above.
To improve performance, we rewrote parts of the SoloKey firmware to use USB CDC,
a high-throughput USB class commonly used for networking devices.
This results in roughly a $32\times$ increase in I/O throughput.
Our artifact is available at:
\begin{center}
\url{https://github.com/edauterman/SafetyPin}
\end{center}

\subsection{Artifact check-list}

\begin{itemize}
  \item {\bf Hardware: }
      \begin{itemize}
          \item 100 SoloKeys\com{~\cite{solokeys}}
          \item 10 Anker SuperSpeed USB 3.0 hubs\com{~\cite{usbhub}}
          \item 2 4-port USB PCIe controller cards\com{~\cite{pcie-cards}}
        \item Linux machine with Intel Xeon E5-260 CPU clocked at 2.60GHz
      \end{itemize}
  \item {\bf Compilation: } The ARM compiler for the SoloKeys, and \texttt{gcc} for the host.
  \item {\bf Metrics: } Latency
  \item {\bf Experiments: } Log-audit time, puncturable encryption overhead, breakdown of recovery time, cluster size vs. recovery time
  \item {\bf Required disk space: } 14MB
  \item {\bf Expected experiment run time: } 50 minutes
  \item {\bf Public link: } \url{https://github.com/edauterman/SafetyPin} 
  \item {\bf Code licenses: } Apache v2
\end{itemize}

\subsection{Description}

\subsubsection{How to access}

We provide reviewers with credentials to remotely access our system. 
Instructions for assembling a similar system are available here:
\begin{center}
\url{https://github.com/edauterman/SafetyPin/blob/master/SETUP.md}.
\end{center}

\subsubsection{Hardware dependencies}

Our artifact uses SoloKeys\com{~\cite{solokeys}} as low-cost HSMs. We use 100
SoloKeys for our experiments, although other deployments could use a different
number of HSMs. 
Because of limitations in the Intel XCHI controller, which supports a maximum of
96 endpoints (each USB 3.0 device has 3 endpoints), we installed
PCIe cards\com{~\cite{pcie-cards}} to support additional endpoints.
This is not necessary for smaller-scale deployments, but larger deployments
should choose a host that supports installing such PCIe cards or find another
solution.
We also recommend USB hubs with an external power source such as the Anker hubs\com{~\cite{usbhub}}.

\subsubsection{Software dependencies}

The firmware for the SoloKeys builds on the original SoloKey firmware, which
already includes several libraries for cryptographic primitives on embedded devices:
\begin{center}
\url{https://github.com/solokeys/solo}.
\end{center}
To support pairings for aggregate signatures, we use the jedi-pairing library
for embedded devices:
\begin{center}
\url{https://github.com/ucbrise/jedi-pairing/}.
\end{center}
For USB HID support, we use the Signal11 library:
\begin{center}
\url{https://github.com/signal11/hidapi}.
\end{center}
We implement our cryptographic primitives that do not require pairings at the host using OpenSSL.

\subsection{Installation}

Instructions for building the host are available under \texttt{host/}.
Instructions for building the firmware and flashing the SoloKeys are available under
\texttt{hsm/}.
The SoloKey documentation
provides additional details and troubleshooting for building and flashing the SoloKeys:
\begin{center}
\url{https://docs.solokeys.io/}.
\end{center}
When experimenting with SafetyPin, you should not boot SoloKeys in DFU mode, as this
locks the firmware and will prevent you from modifying the firmware later (e.g. to
load an updated version of the SafetyPin source).

\subsection{Experiment workflow}

Reviewers can remotely access our machine and run all experiments by executing
\texttt{./runAll.sh} in \texttt{bench/}.
More detailed instructions for running individual experiments are available
here:
\begin{center}
\url{https://github.com/edauterman/SafetyPin#instructions-for-artifact-evaluation}.
\end{center}

\subsection{Evaluation and expected result}

Run all experiments by executing \texttt{./runAll.sh} in
\texttt{bench/}. This will produce figures in \texttt{bench/out} that
match \cref{fig:logeval}, \cref{fig:scale-punc}, \cref{figs:breakdown},
and \cref{fig:scale-breakdown}.
Note that we only reproduce the recovery time breakdown in \cref{figs:breakdown}.
Additionally, the configuration we set up for the reviewers only uses
90 HSMs for \cref{fig:logeval} and \cref{fig:scale-breakdown}. We do this to
keep different firmware on the remaining 10 HSMs to measure the breakdown
in puncturable encryption time as the secret key size increases (\cref{fig:scale-punc}).
For the experiments we show in the body of the paper, we re-flashed HSMs
between experiments so that we could use all 100 HSMs to generate \cref{fig:logeval}
and \cref{fig:scale-breakdown}.

\subsection{Experiment customization}

The experiment for \cref{fig:logeval} can be modified to measure different
data center sizes without changing the firmware on the HSMs.
The experiment for \cref{fig:scale-punc} can likewise be modified to measure
different secret key sizes, although this requires changing HSM firmware.
If we had more HSMs, we could easily expand \cref{fig:scale-breakdown} to show
the effect of larger cluster sizes. We do not measure cluster sizes less than
40 because our analysis shows that our security guarantees begin to break down
below this point.

\subsection{Notes}

To switch between USB CDC and USB HID, change the HID flag on both the host
and the HSMs (this requires loading new firmware on the HSMs).
More detailed instructions are available here:
\begin{center}
\url{https://github.com/edauterman/SafetyPin/blob/master/SETUP.md}.
\end{center}

Note that rather than generating puncturable encryption
secret keys on the HSM (a process we estimate would take roughly 75 hours), to run our 
experiments efficiently, we generate the secret key on the host (for security, a real-world
deployment would need to generate this secret key on the HSM). 

\subsection{AE Methodology}

Submission, reviewing and badging methodology:

\begin{center}
\url{https://www.usenix.org/conference/osdi20/call-for-artifacts}
\end{center}
 
\end{document}